\PassOptionsToPackage{authoryear,round}{natbib}
\documentclass{article}
\usepackage[preprint]{neurips_2024}
\usepackage[T1]{fontenc}
\usepackage{amsmath,amssymb,amsthm}
\usepackage{mathrsfs}
\usepackage{alphalph}
\usepackage{booktabs,longtable,array,multirow}
\usepackage{graphicx}
\usepackage{enumitem}
\usepackage[table]{xcolor}
\usepackage{tikz}
\usetikzlibrary{arrows.meta,positioning}
\usepackage{microtype}
\makeatletter
\def\verbatim@font{\scriptsize\ttfamily}
\makeatother
\definecolor{ink}{RGB}{43,38,34}
\definecolor{accent}{RGB}{156,86,54}
\definecolor{sand}{RGB}{246,241,233}
\definecolor{sage}{RGB}{117,132,106}
\definecolor{warmgray}{RGB}{117,105,94}
\definecolor{ruleline}{RGB}{214,202,188}
\definecolor{lightaccent}{RGB}{248,238,229}

\newtheorem{definition}{Definition}
\newtheorem{proposition}{Proposition}
\newtheorem{remark}{Remark}
\newtheorem{assumption}{Assumption}
\newcommand{\D}{\mathcal{D}}
\newcommand{\M}{\mathcal{M}}
\newcommand{\Spec}{\mathcal{S}}
\newcommand{\F}{\mathcal{F}}
\newcommand{\argmax}{\operatorname*{arg\,max}}

\title{Objective-Oriented Quantitative Investment:\\A Specification-Driven Framework for Automated\\Synthesis of Trading Strategy Pipelines}

\author{Liangliang Zhang\\
Shanghai Liangbai Technology Co., Ltd.\\
\texttt{liangliangzhang81@qq.com}\\
Sole and corresponding author}

\begin{document}
\maketitle

\begin{abstract}
Automated quantitative research has made striking progress---AutoML pipeline synthesis, neural architecture search, reinforcement-learned alpha mining, and LLM-agent research loops all promise to automate the construction of trading strategies. Yet every one of these systems answers the same question: \emph{which strategy scores highest on a scalar metric?} We argue this question is \emph{incomplete}. Professional investors do not order ``the highest return''; they order an \emph{identity}---pure stock-selection alpha uncontaminated by style exposure, resilient in unilateral market declines, within turnover and capacity budgets. We call the incumbent paradigm \emph{result-oriented} and propose, in its place, \textbf{Objective-Oriented Quantitative Investment (OOQI)}: a specification-driven framework in which (i) the full strategy pipeline is modeled as a typed design space of interchangeable modules with explicit interface contracts ($8.85\times 10^{8}$ assemblies in our reference instantiation, high $10^{8}$s after legality pruning); (ii) investor intent is formalized as a \emph{strategy profile specification}---a composable language of measurable, falsifiable clauses drawn from eight requirement families, with hard/soft semantics, priorities, and an interaction algebra capturing synergy and conflict among objectives; and (iii) a compiler translates specifications into constrained assemblies, verifies satisfaction clause-by-clause, and iterates with attribution. Because search over large assembly spaces inflates apparent satisfaction, we develop a \emph{verification protocol} that treats the satisfaction rate itself as a statistical object subject to deflation for search width, isolated temporal holdout, and random-assembly null models. An illustrative synthetic demonstration with 32 pipeline assemblies shows that result-oriented selection attains the top in-sample information ratio while satisfying only 25\% of the specification, whereas specification-driven selection satisfies 100\% of it at a measurable 5.5\% score cost---the two paradigms provably select different strategies. OOQI reframes automated quantitative research from \emph{optimization of a number} to \emph{satisfaction of an identity}, and supplies the formal apparatus---language, compiler, and statistics---to make that reframing rigorous. A theory of optimal assembly accompanies the framework: satisfaction-driven synthesis is NP-hard in general yet constant-factor approximable in a conflict-free (submodular) regime; specifications form a lattice dual to assemblies, yielding identity inference and zero-cost clause completion; each clause carries a Lagrangian shadow price---a clause-level transfer coefficient; and rolling re-certification can be made anytime-valid via e-processes, turning alpha decay into certificate expiry.
\end{abstract}

\section{Introduction: From Results to Objectives}\label{sec:intro}

Two decades of automation have transformed how trading strategies are built. Combined algorithm selection and hyperparameter optimization (CASH) searches over preprocessing-and-model pipelines \citep{thornton2013autoweka,feurer2015autosklearn,olson2016tpot}; neural architecture search (NAS) explores discrete spaces of predictor structures \citep{liu2019darts,pham2018enas,white2023nas1000}; evolutionary and reinforcement-learning systems mine formulaic alphas \citep{zhang2020autoalpha,yu2023alphagen}; and, most recently, LLM-agent frameworks close the loop of hypothesis generation, implementation, and backtest feedback \citep{li2025rdagentquant,wang2023alphagpt,novikov2025alphadevolve}. Across this entire landscape, one element has remained constant: the \emph{objective is a scalar}. Validation loss, information coefficient, information ratio, Sharpe ratio---systems differ in what they search but not in what they seek.

This constancy is, on reflection, strange, because it is not how professional capital is allocated. A pension fund does not mandate ``maximize Sharpe''; a quantitatively minded principal does not instruct a researcher to ``maximize IC.'' Practitioners issue \emph{identity constraints}: the return stream must be attributable to stock selection rather than to market direction or style timing; exposure to dividend and size factors must be zero; the strategy must remain profitable in unilateral declines; turnover must respect a cost budget; capacity must respect market impact. The celebrated decomposition of active returns into factor exposure and residual selection \citep{sharpe1992style,ang2009norway,clarke2017purefactor} exists precisely because sophisticated capital refuses to pay for disguised beta. \emph{What is ordered is not a result but an objective profile; what is delivered must satisfy a specification, not merely win a leaderboard.}

We call the incumbent paradigm \textbf{result-oriented}: strategy construction is cast as $d^* = \argmax_{d \in \D}\, g(\theta(d))$, where $\D$ is a space of candidate strategies, $\theta(d)$ a vector of evaluation statistics, and $g$ a scalarization---typically one coordinate of $\theta$. We propose \textbf{Objective-Oriented Quantitative Investment (OOQI)}, in which the primitive object is a \emph{specification} $\Spec$---a structured collection of measurable, falsifiable clauses over $\theta$---and the synthesis problem is to find assemblies that \emph{satisfy} $\Spec$, optimizing a scalar score only within the satisfying set:
\[
d^*_{\F} \;=\; \argmax_{d \in \F(\Spec)}\, g(\theta(d)), \qquad \F(\Spec) = \{ d \in \D : d \models \Spec \}.
\]
The two paradigms are not merely different in emphasis; we show formally (\S\ref{sec:formal}) and demonstrate empirically (\S\ref{sec:demo}) that they \emph{select different strategies}: result-oriented search systematically harvests whatever earns the score---including style beta---while specification-driven search purchases identity compliance at a measurable score cost.

\paragraph{Why now: the LLM compiler.} Specification-driven synthesis requires mapping semantic clauses (``zero style exposure'') onto structural decisions (label neutralization, exposure-projection post-processing, attribution clauses in evaluation). Until recently this mapping demanded a senior researcher. Large language models change the economics of the compiler: they carry the literature priors that turn blind combinatorial search into informed proposal-and-verification \citep{romera2023funsearch,novikov2025alphadevolve,li2025rdagentquant}. OOQI is deliberately agnostic about whether the compiler is an LLM agent, a Bayesian optimizer, or a human committee---but we note that LLM agents make the framework practical today.

\paragraph{Contributions.}
\begin{enumerate}[leftmargin=1.6em,itemsep=1pt]
\item \textbf{A problem reframing with formal teeth} (\S\ref{sec:formal}). We define the strategy-pipeline design space with typed interface contracts, the strategy profile specification language with hard/soft clauses and priorities, and the satisfaction semantics $d \models \Spec$. Three propositions establish: style contamination of result-oriented selection (P\ref{prop:contamination}); the non-negative, transfer-coefficient-like \emph{cost of specification} (P\ref{prop:cost}); and \emph{satisfaction inflation} under search width, motivating deflated reporting (P\ref{prop:inflation}).
\item \textbf{A specification language for strategy identity} (\S\ref{sec:framework}, App.~\ref{app:spec}). Eight requirement families---signal quality, return--risk, purity/style, regime robustness, execution feasibility, portfolio structure, temporal behavior, transparency/compliance---comprising 30+ composable clauses, each with a measurement protocol grounded in an established literature prototype, plus an interaction algebra (synergy/conflict) and an arbitration semantics for conflicting objectives (App.~\ref{app:algebra}).
\item \textbf{A verification protocol that treats satisfaction as a statistic} (\S\ref{sec:protocol}, App.~\ref{app:protocol}). Search ledger accounting, temporal holdout isolated from search, random-assembly null models, multi-seed confidence intervals, and a deflated satisfaction rate adapting PBO/CSCV and Deflated-Sharpe corrections \citep{bailey2017pbo,bailey2014dsr,harvey2016crosssection} from single-metric selection bias to multi-clause specification satisfaction.
\item \textbf{An illustrative demonstration} (\S\ref{sec:demo}, App.~\ref{app:repro}). On a fully disclosed synthetic market with 32 pipeline assemblies, result-oriented selection attains the highest in-sample IR (7.10) while satisfying 25\% of the specification; specification-driven selection satisfies 100\% of it at 5.5\% IR cost, and held-out satisfaction reveals the stochasticity that motivates our protocol.
\item \textbf{An inverse-problem formalization} (\S\ref{sec:inverse}, App.~\ref{app:inverse}). We identify specification-driven synthesis with the inverse problem of strategy synthesis, import Hadamard's well-posedness triad (existence/uniqueness/stability) as feasibility/underspecification/satisfaction-stability, and give margin-regularized selection a Tikhonov interpretation with a stability guarantee.
\item \textbf{A theory of optimal assembly} (\S\ref{sec:optimal}, App.~\ref{app:proofs}, \ref{app:inverse}). We establish NP-hardness of satisfaction-driven assembly via a 3-SAT reduction (P\ref{prop:nphard}); a constant-factor approximation regime under conflict-free (coverage/submodular) composition (P\ref{prop:submodular}); quality measures for the satisfying set---margin profiles, certified radius, Pareto fronts, robust stress depth, and a quality-diversity archive reading; a Galois-connection duality between specifications and assemblies that yields identity inference and zero-cost clause completion (P\ref{prop:galois}); Lagrangian shadow prices per clause (P\ref{prop:shadow}); and anytime-valid rolling certification via e-processes (P\ref{prop:evalue}).
\item \textbf{A framework walkthrough and a real-complexity case study} (\S\ref{sec:walkthrough}, \S\ref{sec:case}). We fix the locus of each object---which decisions live in the design space, which in the specification, which in the protocol---on a production-shaped pipeline, state the locus of accountability that distinguishes objective-oriented selection from result-chasing, and decompose a documented macro--meso--micro coupling architecture into the framework's grammar: every architectural decision becomes a switch, every engineering red line finds its formal home (legality, clause, or protocol), and an interpretability clause grown from practice ($w_t \ge 0.70$) enters the transparency family.
\end{enumerate}

\paragraph{Two gates against overfitting.} The discipline proposed here operates through two sequential gates, neither of which is itself a test statistic. The first is \emph{architectural transparency}: a strategy is not a monolithic model but an assembly over a contracted design space (Def.~\ref{def:space}), in which every module declares a side-effect signature. One knows what one is building: each switch's effect on each evaluation coordinate is declared ex ante (the contract's side-effect signature) and measurable ex post. The two disciplines are complementary, and the measurement half is concrete: averaging the sixteen matched assembly pairs of the demonstration ledger that differ only in portfolio width, narrowing 50$\to$20 raises mean $|\beta|_{\max}$ by $7.12$ (unstandardized DGP units, \S\ref{sec:demo} footnote) and lowers mean IR by $1.39$---the signature says \emph{which dials} a switch moves; the ledger says \emph{by how much}. Transparency is what makes the declared search width $N$ computable---the nominal count of legal assemblies---and without a declared $N$ the deflated statistics of \S\ref{sec:protocol} lack their essential input. The declared $N$ does not exhaust the effective trials: threshold deliberation, discarded specification drafts, module authors' tacit choices, and fixed environment or context configurations (\S\ref{sec:case}) sit outside any architecture contract, and we flag them as residual channels rather than claim they are closed. The second gate is \emph{demand-side priority}: the specification---what return, what risk, obtained through which declared exposures---is written down as clauses \emph{before} any ledger is consulted. This blocks the psychology of overfitting at its root: the canonical single-shot HARKing mechanism---hypothesizing after the results are known \citep{kerr1998harking}---is structurally blocked when the target is registered before the arrow is shot, while residual channels (specification-shopping, informed threshold deliberation, repeated registration) are governed, not eliminated, by the disciplines of App.~\ref{app:gov}. The two gates compose into a mechanism: transparency makes $N$ computable, giving deflated statistics their input; priority makes violations definable, turning overfitting from an invisible statistical drift into a detectable clause breach; the certificate and audit trail (\S\ref{sec:protocol}) make both attributable. The paradigm does not fight overfitting with cleverer models, but with the discipline of \emph{declare--compile--certify--audit}.

\paragraph{Scope and honesty.} OOQI is not a promise that any specification can be met: feasibility is decided by data, market, and budget, and the framework's duty is to \emph{report satisfaction or its absence with statistical integrity}, never to manufacture it. The demonstration of \S\ref{sec:demo} is synthetic by design---it validates the decision logic of the paradigm, not a live track record; a production case study is part of the research roadmap (App.~\ref{app:roadmap}).

\section{Related Work}\label{sec:related}

We organize four literature streams; an extended treatment is deferred to App.~\ref{app:related}. Table~\ref{tab:positioning} summarizes the positioning.

\paragraph{Supply-side automation: pipeline and architecture search.}
CASH formalizes joint algorithm selection and hyperparameter optimization over two-layer pipelines \citep{thornton2013autoweka}; auto-sklearn adds meta-learning and ensembling \citep{feurer2015autosklearn}; TPOT evolves operator trees and is notable for a rare second objective (pipeline complexity) \citep{olson2016tpot}; CASH+ acknowledges that modern adaptation stages induce mutually incompatible search spaces \citep{balef2025cashplus}. NAS searches discrete predictor structures \citep{liu2019darts,pham2018enas}; a meta-analysis of 1000 papers confirms that NAS targets accuracy, occasionally augmented by hardware-aware proxies \citep{white2023nas1000}. In finance, Qlib provides full-chain modularity \emph{without} search \citep{yang2020qlib}; Quant 4.0 articulates full-chain automation as a \emph{vision} without formalization \citep{guo2024quant4}; evolutionary and RL systems search \emph{one stage}---factor mining---with scalar fitness \citep{zhang2020autoalpha,yu2023alphagen}. \emph{Modularity and combinatorial search have never met inside one quantitative system, and---with the exceptions we now acknowledge---search objectives remain scalars.}

Two neighboring literatures constrain search rather than merely optimize scalars, and we position against them honestly. Constrained Bayesian optimization \citep{gelbart2014constrained,hernandez2016general} and safe reinforcement learning \citep{achiam2017cpo} maximize objectives subject to probabilistic or expected-cost constraints; multi-objective and hardware-aware NAS \citep{tan2019mnasnet,lu2019nsganet} return Pareto fronts over accuracy and deployment proxies. These are the closest formal ancestors of our constrained selection, and \S\ref{sec:demo} and App.~\ref{app:repro} therefore include explicit baselines (penalty scalarization; Pareto filtering). Three deltas remain: constraints there are (i) probabilistic-in-expectation rather than clause-wise auditable, (ii) defined over continuous knobs rather than typed pipeline identity, and (iii) unaccompanied by search-width deflation statistics. Specification formalisms from temporal logic (STL/MTL) offer hard/soft semantics we adopt, without strategy-space instantiation.

\paragraph{LLM-agent quantitative R\&D.}
R\&D-Agent evolves hypotheses under benchmark feedback and its quantitative instantiation jointly optimizes factors and models against IC/return metrics \citep{yang2025rdagent,li2025rdagentquant}. Alpha-GPT compiles natural-language \emph{trading ideas} into formulaic alphas with human-in-the-loop review \citep{wang2023alphagpt,yuan2024alphagpt2}. FunSearch and AlphaEvolve establish the LLM$+$evaluator$+$evolution paradigm against programmatically computable scalar evaluators \citep{romera2023funsearch,novikov2025alphadevolve}. Decision-type agents trade directly rather than synthesize research \citep{zhang2024finagent,xiao2024tradingagents}. The closest specification-flavored precedent is an option-strategy DSL compiling intent into executable option structures \citep{luo2026oql}; its requirement space (option legs, Greeks) shares nothing with strategy identity profiles, and it lacks satisfaction statistics. \emph{Existing systems start from data, metrics, or ideas; none starts from a verifiable multi-dimensional specification.}

\paragraph{Demand-side formalization.}
Goals-based wealth management replaces the scalar with target-achievement probability \citep{das2018gbwm,deguest2015gbwm,kim2020goalprog}; institutional mandates formalize tracking-error and beta constraints \citep{roll1992te,jorion2003te}; the transfer coefficient prices the alpha destroyed by constraints \citep{clarke2002tc}; pure-factor portfolios solve ``exposure one to X, zero to all else'' \citep{clarke2017purefactor}; returns-based style analysis provides the canonical identity measurement \citep{sharpe1992style}; the Norwegian-commission evaluation established that apparent alpha is largely compensated factor exposure \citep{ang2009norway}; regime models make market state a decision variable \citep{ang2002regime,shu2024jump}; multi-objective evolutionary methods trace Pareto fronts \citep{deb2002nsga2,anagnostopoulos2009moea}. \emph{Every clause of our language has an isolated mathematical prototype; no work unifies them into a composable specification language driving synthesis.}

\paragraph{Evaluation credibility.}
Backtest overfitting is quantified by PBO/CSCV \citep{bailey2017pbo}; the Deflated Sharpe Ratio corrects selection bias, non-normality, and trial counts \citep{bailey2014dsr}; factor-zoo multiple testing demands thresholds rising with test counts \citep{harvey2016crosssection,harvey2015backtesting}; reality-check and SPA tests control data snooping pool-wide \citep{white2000rc,hansen2005spa}; model selection and search-phase evaluation overfit in AutoML/NAS \citep{cawley2010overfitting,yang2020nasfrustrating,sciuto2020evalnas}; and underspecification motivates multi-property stress certification in place of single-metric validation \citep{damour2022underspecification}. \emph{All correct single scalars; none treats multi-clause specification satisfaction as the statistical object.}

\begin{table}[t]
\caption{Positioning: what starts the process, what is searched, what is sought. OOQI is the only row whose input is a verifiable identity specification and whose output is a satisfaction certificate.}
\label{tab:positioning}
\centering\small
\begin{tabular}{@{}p{3.1cm}p{3.1cm}p{3.2cm}p{3.6cm}@{}}
\toprule
\rowcolor{sand}\textbf{System family} & \textbf{Starts from} & \textbf{Searches} & \textbf{Seeks / outputs} \\
\midrule
CASH / AutoML / NAS \citep{thornton2013autoweka,liu2019darts} & dataset $+$ metric & preprocess $\times$ model, network cells & scalar-optimal pipeline \\
Quant factor mining \citep{zhang2020autoalpha,yu2023alphagen} & data $+$ IC fitness & one stage (factors) & scalar-optimal factor set \\
Qlib / Quant 4.0 \citep{yang2020qlib,guo2024quant4} & human workflow & \emph{none} (modular) & configured pipeline \\
R\&D-Agent-Quant \citep{li2025rdagentquant} & data $+$ backtest metrics & factor $\times$ model co-evolution & metric-improving hypotheses \\
Alpha-GPT 1/2 \citep{wang2023alphagpt,yuan2024alphagpt2} & trading \emph{idea} (NL) & alpha expressions & idea-consistent alphas \\
Option-strategy DSL \citep{luo2026oql} & option intent (DSL) & leg structures & executable option strategy \\
Goals-based WM \citep{das2018gbwm} & wealth targets & static allocations & goal-achievement probability \\
\midrule
\rowcolor{lightaccent}\textbf{OOQI (this work)} & \textbf{strategy profile specification} & \textbf{full-chain assemblies} & \textbf{spec satisfaction $+$ certificate} \\
\bottomrule
\end{tabular}
\end{table}

\section{Problem Formulation}\label{sec:formal}

\subsection{The design space: modules with typed contracts}

\begin{definition}[Module and assembly space]\label{def:space}
A strategy pipeline comprises $K$ \emph{stages} (in our reference instantiation, $K{=}9$: universe/data, label, feature processing, feature selection, predictor, training regime, post-processing, portfolio construction, risk overlay; App.~\ref{app:space}). Stage $i$ offers a finite set $\M_i$ of \emph{modules}. Each module $m \in \M_i$ declares (i) an \emph{input contract} $\iota(m)$ and \emph{output contract} $o(m)$ over typed channels (panel shape, frequency, exposure structure, weight continuity); and (ii) a \emph{side-effect signature} $\sigma(m) \in \mathbb{R}^{d_s}$ describing its first-order effect channels on evaluation coordinates (e.g., turnover pressure, exposure leakage, capacity footprint). An \emph{assembly} $d = (m_1,\dots,m_K)$ is \emph{legal} if contracts chain compatibly; the \emph{design space} is
\[
\D \;=\; \Big\{ d \in \textstyle\prod_{i=1}^{K} \M_i \;:\; \mathrm{compat}(d) = 1 \Big\}.
\]
\end{definition}

With the module counts of App.~\ref{app:space}, $|\D| \approx 4.4\times 10^{7}$ at the module-family level and $\approx 8.8\times 10^{8}$ once the families' discrete hyperparameter variants are counted (App.~\ref{app:comb} gives the two-level factorization). Three properties preclude exhaustive enumeration and motivate \emph{compilation}: evaluation cost is minutes-to-hours per assembly; legality constraints couple stages (e.g., an exposure-projection post-processor requires a predictor whose scores admit exposure regression); and stage effects interact non-additively (rank labels compound rank preprocessing; warm-start chaining offsets recency weighting).

\begin{definition}[Evaluation functional]\label{def:eval}
Fix data and a backtest protocol. Evaluation maps an assembly to a statistics vector $\theta(d) \in \mathbb{R}^{d_\theta}$ whose coordinates include: information-coefficient family (RankIC mean, ICIR, positive-day share); return--risk coordinates (annualized excess, IR, max drawdown); identity coordinates (style exposure $|\beta|_{\max}$ vs.\ a factor model, correlation with style indices, RBSA residual share); regime-conditional coordinates (conditional excess return and drawdown in down-market states); feasibility coordinates (turnover, cost-sensitivity slope, capacity utilization); and structural coordinates (holdings count, industry deviation).
\end{definition}

\subsection{The specification: a language for strategy identity}

\begin{definition}[Specification]\label{def:spec}
A \emph{strategy profile specification} $\Spec = \{ c_j \}_{j=1}^{J}$ is a set of \emph{clauses} $c_j = (\phi_j, \kappa_j, \pi_j)$ where: $\phi_j : \mathbb{R}^{d_\theta} \to \{0,1\}$ is a measurable, falsifiable predicate over evaluation coordinates (a \emph{threshold predicate} $\phi_j(\theta) = \mathbf{1}[\theta_{j} \ \mathrm{rel}_j\ \tau_j]$ with $\mathrm{rel}_j \in \{\le,\ge\}$); $\kappa_j \in \{\textsc{hard},\textsc{soft}\}$ is the clause's force; and $\pi_j \in \mathbb{N}$ is a priority. Write $d \models c_j$ iff $\phi_j(\theta(d)) = 1$, and $d \models \Spec$ iff $d \models c_j$ for every \textsc{hard} $c_j$. The \emph{satisfaction rate} is
\[
\mathrm{SR}(d, \Spec) \;=\; \frac{1}{J} \sum_{j=1}^{J} \phi_j(\theta(d)),
\]
reported with the clause-wise vector $\big(\phi_j(\theta(d))\big)_j$ (the \emph{satisfaction profile}).
\end{definition}

Clauses compose like building blocks; pairs interact as \emph{synergistic} (joint satisfaction cheaper than separate), \emph{conflicting} (joint satisfaction strictly costlier, up to infeasibility), or \emph{neutral} (App.~\ref{app:algebra} formalizes via the cost functional below). Conflicts are arbitrated by a two-level semantics: \textsc{hard} clauses carve the feasible set $\F(\Spec)$; \textsc{soft} clauses enter a weighted objective or, under lexicographic priority, are satisfied in order $\pi_1 \succ \pi_2 \succ \cdots$ in the manner of goal programming \citep{kim2020goalprog}.

\subsection{Two paradigms, three propositions}

\begin{definition}[Result-oriented vs.\ objective-oriented selection]\label{def:paradigms}
Let $g : \mathbb{R}^{d_\theta} \to \mathbb{R}$ be a scalar score (e.g., in-sample IR). \emph{Result-oriented} selection returns $d^* = \argmax_{d\in\D} g(\theta(d))$. \emph{Objective-oriented} selection returns $d^*_\F = \argmax_{d \in \F(\Spec)} g(\theta(d))$, with $\F(\Spec)$ from Def.~\ref{def:spec}.
\end{definition}

\begin{assumption}[Style premium]\label{ass:premium}
There exist module choices that raise the expected in-sample score by loading return on a priced style factor: formally, writing $X_d = g(\theta(d))$, a subset $\D_S \subset \D$ of style-loading assemblies with common mean $\mathbb{E}[X_d] = \mu_S$ for $d \in \D_S$, and a purity clause $c_{\mathrm{pure}}$ that separates the classes exactly: the style-pure class is $\D_P := \{d \in \D : d \models c_{\mathrm{pure}}\}$ with common mean $\mathbb{E}[X_d] = \mu_P$ for $d \in \D_P$, $\D_S \cap \D_P = \emptyset$, $\D = \D_S \cup \D_P$, and $\mu_S > \mu_P$.
\end{assumption}

\begin{proposition}[Style contamination of result-oriented selection]\label{prop:contamination}
Under Assumption~\ref{ass:premium}, suppose each score estimate $X_d = g(\theta(d))$ has sub-Gaussian noise with variance proxy $s^2$, i.e.\ $\Pr[\pm(X_d - \mathbb{E} X_d) \ge u] \le e^{-u^2/(2s^2)}$ for $u \ge 0$. Write $\Delta = \mu_S - \mu_P > 0$. Then
\[
\Pr\big[ d^* \not\models c_{\mathrm{pure}} \big] \;\ge\; 1 - \big(1 + |\D_P|\big)\,\exp\!\Big(-\frac{\Delta^2}{8s^2}\Big),
\]
which exceeds $1/2$ once $\Delta > 2s\sqrt{2\log\!\big(2(1+|\D_P|)\big)}$ and tends to $1$ as the premium grows or the noise shrinks, for fixed finite classes. Hence result-oriented selection violates identity clauses with high probability exactly when style premia are large relative to estimation noise and to the breadth of the style-pure class---a wide pure class imposes a finite-class penalty of order $s\sqrt{\log |\D_P|}$ before the guarantee bites. No such failure attaches to $d^*_\F$, which satisfies every hard clause by construction.
\end{proposition}

\begin{proposition}[Cost of specification]\label{prop:cost}
Define $\mathrm{CoS}(\Spec) = g(\theta(d^*)) - g(\theta(d^*_\F))$. Then $\mathrm{CoS}(\Spec) \ge 0$, with $\mathrm{CoS} = 0$ iff some score argmax is feasible (if several assemblies tie for the argmax, $\mathrm{CoS} = 0$ iff \emph{at least one} of them lies in $\F(\Spec)$). Moreover $\mathrm{CoS}$ is monotone in specification strength: adding hard clauses weakly increases it. Assume further that the specification contains the purity clause, so $\F(\Spec) \subseteq \D_P$, and adopt Assumption~\ref{ass:premium} with sub-Gaussian noise of variance proxy $s^2$. Then, writing $\tilde\mu_P = \max_{d \in \F(\Spec)} \mathbb{E}[X_d]$ for the best expected score of a feasible assembly ($\tilde\mu_P = \mu_P$ under the common-mean clause of Assumption~\ref{ass:premium}),
\[
\mathrm{CoS} \;\le\; \mu_S - \tilde\mu_P + 4s\sqrt{2\log |\D|}
\quad\text{with probability at least } 1 - 2|\D|^{-3}:
\]
the price of identity is bounded by the temptation it forgoes plus an estimation-width term---a direct analogue of the transfer coefficient \citep{clarke2002tc}, computed per specification rather than per portfolio constraint.
\end{proposition}

\begin{proposition}[Satisfaction inflation under search]\label{prop:inflation}
Suppose the truth is \emph{null}: no assembly satisfies clause set $\{c_j\}_{j=1}^J$, and each clause's IS estimate is satisfied spuriously with probability $q_j$ independently. After evaluating $N$ assemblies,
\[
\Pr\big[ \exists\, d \in \D_N : d \models \Spec \ \text{(apparently)} \big] \;=\; 1 - \Big(1 - \textstyle\prod_j q_j\Big)^{N},
\]
which approaches $1$ exponentially fast in $N$ even when $\prod_j q_j$ is tiny. Consequently, an observed satisfaction profile from a searched space is upward-biased, and reporting $\mathrm{SR}$ without the \emph{search ledger} $N$ is statistically meaningless; \S\ref{sec:protocol} defines the deflated counterpart $\widehat{\mathrm{SR}}_{\mathrm{def}}$.
\end{proposition}

Proofs are given in App.~\ref{app:proofs}. The propositions are elementary by design---their role is to make the paradigm's claims \emph{falsifiable and budgetable}: contamination is not an accusation but a probability; cost is not a regret but a quantity with an upper bound; inflation is not a worry but an exponential law.

\subsection{The inverse-problem reading}\label{sec:inverse}

There is a classical mathematical duality that subsumes the two paradigms of Def.~\ref{def:paradigms}. \emph{Forward problems} infer effects from causes: given an operator and its inputs, solve, and check whether the solution conforms to reality. Result-oriented research is the forward problem of quantitative finance---given assembly $d$ and market realization $\omega$, evaluate $\theta(d,\omega)$ and inspect. \emph{Inverse problems} infer causes from effects: given requirements on the solution, recover the operator. OOQI is exactly the inverse problem of strategy synthesis---given a specification on the return stream's identity, recover assemblies whose evaluation conforms. Inverse problem theory then organizes, under one roof, questions this paper has been asking all along (full formalization in App.~\ref{app:inverse}): Hadamard's triad---\emph{existence, uniqueness, stability}---maps respectively onto \emph{feasibility of the specification} (is $\F(\Spec)$ empty?), \emph{underspecification} (the satisfying set is typically an equivalence class, not a point), and \emph{stability of satisfaction} (Proposition~\ref{prop:inflation} is precisely an instability law); and Tikhonov-style regularization corresponds to our margin-maximizing selection and priority structure, which purchase stability at the price of bias. The market's operator being stochastic and nonstationary, the correct theoretical home is \emph{statistical} inverse problems \citep{kaipio2005statistical}, and rolling re-certification is mandatory. We stress that the analogy is structural, and where it ends (App.~\ref{app:inverse}), it ends honestly.

\section{Optimal Assembly: How to Compose, How to Judge, What Each Clause Costs}\label{sec:optimal}

Definition~\ref{def:paradigms} poses the synthesis problem; this section develops its theory. Three questions a principal inevitably asks next organize the development: \emph{how hard} is optimal assembly (\S\ref{sec:complexity}); \emph{which} member of the satisfying set should be returned, and \emph{how should the many tied solutions be judged} (\S\ref{sec:quality}, \S\ref{sec:lattice}); and \emph{what does each clause cost} in equilibrium (\S\ref{sec:shadow}). A statistical consequence for certification over time closes the section (\S\ref{sec:anytime}). Proofs are collected in App.~\ref{app:proofs}.

\subsection{Composition is combinatorially hard---but not uniformly}\label{sec:complexity}

The design space (Def.~\ref{def:space}) motivates proposal-and-verification because exhaustion is infeasible. We now show the difficulty is \emph{intrinsic}, not merely practical---and then isolate the structural regime in which principled approximation exists.

\begin{proposition}[NP-hardness of satisfaction-driven assembly]\label{prop:nphard}
Deciding whether $\F(\Spec) \neq \emptyset$ is NP-hard in the joint (stages $\times$ clauses) dimension, even when contract legality is trivial (all module pairs compatible) and every clause predicate is evaluable in polynomial time. With polynomial-time predicates, the problem is NP-complete.
\end{proposition}

The reduction (App.~\ref{app:proofs}) is from 3-SAT: stages play the role of variables (two modules each), and each specification clause is satisfied iff the assembly contains at least one designated \emph{literal module}; a fully satisfying assembly exists iff the formula is satisfiable. Two readings matter for practice. First, the source of hardness is \emph{interaction}, not size: if every clause's satisfaction decomposes over single stages, feasibility is decidable in $\sum_i |\M_i|$ by per-stage filtering---hardness enters exactly when clauses couple stages, the formal version of the observation in \S\ref{sec:formal} that ``legality constraints couple stages.'' Second, the compiler's proposer is thereby a heuristic oracle whose outputs must be \emph{certified} rather than trusted: OOQI needs soundness only from the verifier, never from the proposer---the same asymmetry that underlies syntax-guided synthesis \citep{alur2013sygus}, where a specification, a grammar (here: typed contracts), and an untrusted synthesizer meet.

\begin{proposition}[A constant-factor regime under conflict-free composition]\label{prop:submodular}
Suppose soft-clause satisfaction is \emph{coverage-structured}: each clause $c_j$ designates a module family $A_j \subseteq \bigcup_i \M_i$ and is satisfied iff the assembly's module set meets $A_j$, with clause weights $w_j \ge 0$. Then the weighted satisfaction $f(d) = \sum_j w_j\, \phi_j(d)$ is a monotone submodular function of the module set, and, under the one-module-per-stage (partition matroid) constraint, (i) a stage-wise greedy pass attains at least $1/2$ of the optimum \citep{nemhauser1978}, and (ii) the multilinear relaxation with pipage rounding attains at least $1 - 1/e \approx 0.63$ \citep{calinescu2011}. Conflicting clause pairs (App.~\ref{app:algebra}) break submodularity; characterizing approximation under supermodular conflict is open (roadmap R5).
\end{proposition}

The practical content is a license and a warning. The license: when mechanisms compose without conflict, the compiler's greedy translation--assembly loop is not ad hoc but carries a constant-factor guarantee---certification then audits a search that was already principled. The warning: where conflicts dominate (the IC$\,\ominus\,$turnover pair of Table~\ref{tab:interaction} being the canonical instance), no such promise exists, and the certificate's deflation machinery is the only honest output. The theory thus tells the compiler \emph{when it may be fast} and \emph{when it must be humble}.

\subsection{Many solutions: quality measures for the satisfying set}\label{sec:quality}

Remark~\ref{prop:uniqueness} (App.~\ref{app:inverse}) records that $\F(\Spec)$ is an equivalence class, not a point. Selection within the class requires quality measures. We define three at the level of an assembly and one at the level of the set; all are computable from quantities the compiler already produces.

\begin{definition}[Margin profile and certified radius]\label{def:margin}
For assembly $d$ and clause $c_j$, let $m_j(d)$ be the signed distance of $\hat\theta_j(d)$ inside the feasible half-space (App.~\ref{app:inverse}). The \emph{margin profile} is the vector $m(d) = (m_j(d))_j$; the \emph{certified radius} $\rho(d) = \min_j m_j(d) / s_j$ normalizes by estimation scale: by Proposition~\ref{prop:stability}, $d$ retains full satisfaction under any evaluation perturbation with $|\varepsilon_j| < m_j(d)$ in every coordinate. Bare satisfaction is a coin flip; deep satisfaction is a property.
\end{definition}

\begin{definition}[Pareto-optimal assembly]\label{def:pareto}
$d \in \F(\Spec)$ is \emph{Pareto-optimal} for $(g, m)$ if no $d' \in \F(\Spec)$ weakly dominates it in score and in every hard-clause margin with at least one strict improvement. The certificate should return a Pareto \emph{front}, not a point: the principal's risk appetite selects along it, and assemblies that tie in-sample are separated exactly by their margins---the stress-certification logic of \citep{damour2022underspecification} internalized into the solution concept.
\end{definition}

\begin{definition}[Robust satisfaction and stress depth]\label{def:robust}
Let the evaluation vary over an uncertainty set $\mathcal{U}$ of regime perturbations (drift shifts, volatility scaling, cost multipliers---the stress suite of App.~\ref{app:protocol}). The \emph{robustly satisfying set} is $\F_{\mathcal{U}}(\Spec) = \{ d \in \D : \mathscr{F}(d,\omega) \in \mathcal{T}\ \text{for all}\ \omega \in \mathcal{U} \}$---the robust-optimization counterpart of satisfaction \citep{bental2009robust}. $\F_{\mathcal{U}}$ is monotone decreasing in $\mathcal{U}$; the largest perturbation level at which $d$ survives is its \emph{stress depth}, a per-clause scalar the certificate can quote next to each margin.
\end{definition}

\paragraph{Set-level quality: the specification archive.} When the deliverable is a menu rather than a single assembly, the satisfying set itself needs a quality measure. Borrowing quality-diversity \citep{pugh2016qd}: take the side-effect signature $\sigma(d)$ (Def.~\ref{def:space}) as the behavior descriptor, partition signature space into mechanism niches, and retain the best-margin certified assembly per niche. The archive's \emph{coverage} (number of certified niches) times its \emph{depth} (margins) is the quality score of the satisfying set---and it has an economic reading: niche-diverse certified assemblies fail through \emph{different mechanisms} (crowding, regime change, cost shocks), so a principal allocating across the archive diversifies mechanism risk, which return-stream correlation alone cannot see. The satisfying set, properly measured, is a portfolio.

\subsection{The specification lattice and identity inference}\label{sec:lattice}

\begin{definition}[Strength order and the identity map]\label{def:lattice}
Order specifications by entailment: $\Spec_1 \preceq \Spec_2$ iff $\F(\Spec_2) \subseteq \F(\Spec_1)$ (``$\Spec_2$ is stronger''); quotienting by semantic equivalence yields the \emph{specification lattice}, with join given by clause union---$\F(\Spec_1 \cup \Spec_2) = \F(\Spec_1) \cap \F(\Spec_2)$ is the least upper bound of the pair. For $X \subseteq \D$, define the \emph{identity} $\mathcal{I}(X)$ as the strongest specification satisfied by every $d \in X$.
\end{definition}

\begin{proposition}[Galois structure]\label{prop:galois}
$(\F, \mathcal{I})$ is an antitone Galois connection between the specification lattice and the assembly powerset \citep{cousot1977}:
\[
X \subseteq \F(\Spec) \;\Longleftrightarrow\; \Spec \preceq \mathcal{I}(X).
\]
Consequences: (i) \emph{completion}---$\Spec \preceq \mathcal{I}(\F(\Spec))$: every specification entails additional clauses that all its satisfying assemblies meet \emph{for free}, certifiable at zero marginal search cost; (ii) \emph{identity inference}---$\mathcal{I}(\{d\})$ is the strategy's \emph{inferred identity}: the compiler run in reverse answers ``what is this strategy?'' from its evaluation vector, the inverse of the inverse problem of \S\ref{sec:inverse}; (iii) \emph{equivalence}---two specifications denote the same feasible set iff their completions coincide, making semantic equivalence decidable through satisfaction rather than syntax.
\end{proposition}

Consequence (i) has immediate engineering value: after certifying $\Spec$, the certifier may report the \emph{bonus clauses} $\mathcal{I}(\F(\Spec)) \setminus \Spec$---identity the principal did not order but provably receives. Consequence (ii) formalizes what due diligence has always done informally: reading a return stream and naming what it is.

\subsection{The price of a clause: Lagrangian shadow prices}\label{sec:shadow}

The interaction algebra (App.~\ref{app:algebra}) measures conflict through finite CoS differences. Constrained-optimization theory supplies the marginal, equilibrium version. Write the synthesis problem with hard clauses as constraints $h_j(d) = \theta_j(d) - \tau_j \ge 0$ (oriented so feasibility is non-negativity), and consider the smoothed/relaxed formulation with Lagrangian $\mathcal{L}(d, \lambda) = g(\theta(d)) + \sum_j \lambda_j\, h_j(d)$.

\begin{proposition}[Shadow price of a clause]\label{prop:shadow}
Under standard regularity (local uniqueness of the optimum; linear independence of binding-constraint gradients), the envelope theorem \citep{milgrom2002envelope} applied to the value function gives
\[
\frac{\partial\, \mathrm{CoS}(\Spec)}{\partial \tau_j} \;=\; \lambda^*_j \;\ge\; 0:
\]
the optimal dual variable is the marginal score price of tightening clause $j$. The vector $(\lambda^*_j)_j$ is thus a \emph{clause-level transfer coefficient} in the sense of \citep{clarke2002tc}---the price list of identity---estimable in practice from the compiler's search history by local re-solves of the certified optimum.
\end{proposition}

\begin{remark}[Conflict as cross-price]\label{rem:crossprice}
In the smooth relaxation, pairwise conflict (Def.~\ref{def:interaction}) corresponds to positive cross-partials: tightening clause $j$ raises the shadow price of a conflicting neighbor, $\partial \lambda^*_i / \partial \tau_j > 0$ for $c_i \ominus c_j$. This reading is local and regularity-dependent; in the discrete design space the compiler estimates it by sensitivity of the certified optimum rather than by duality. We state it as a remark, not a proposition, deliberately.
\end{remark}

Two interpretations close the loop with neighboring fields. \emph{Economic}: certificates can quote not just $\mathrm{CoS}$ but the per-clause price list---which identity is expensive, and whose price a conflicting neighbor is driving up (Remark~\ref{rem:crossprice}); budgeting identity becomes as routine as budgeting tracking error. \emph{AI-theoretic}: the compiler loop is constrained optimization in the sense of constrained MDPs \citep{altman1999cmdp}, and Proposition~\ref{prop:inflation} is precisely Goodhart's law \citep{manheim2019goodhart} formalized for satisfaction statistics---any satisfaction measure used as a search target without ledger accounting inflates, so the deflation protocol of \S\ref{sec:protocol} is the Goodhart-resistant counterpart. In the alignment vocabulary: result-oriented research optimizes a proxy; OOQI specifies the objective and polices the proxy--objective gap statistically.

\subsection{Certification over time: anytime-valid satisfaction}\label{sec:anytime}

Rolling re-certification (App.~\ref{app:inverse}) tests the same clauses repeatedly as new data arrive, and classical $p$-value machinery inflates under such continuous monitoring. The e-value literature \citep{ramdas2023savi,howard2021cs} supplies the instrument built for exactly this regime.

\begin{proposition}[Anytime-valid clause certification]\label{prop:evalue}
For hard clause $c_j$, let $\{ E_t^{(j)} \}_{t \ge 0}$ be a nonnegative supermartingale under the null that clause $j$ fails (an \emph{e-process}; for example, a mixture likelihood ratio computed on weekly margin observations). Then, by Ville's inequality, under the null
\[
\Pr\big[ \exists\, t \ge 0 :\; E_t^{(j)} \ge 1/\alpha \big] \;\le\; \alpha:
\]
a certificate that affirms clause $j$ when $E_t^{(j)}$ crosses $1/\alpha$ is valid at level $\alpha$ \emph{simultaneously over all future inspection times}---continuous re-certification pays none of the multiplicity penalty that rolling $p$-based tests incur \citep{ramdas2023savi}.
\end{proposition}

The operational reading upgrades the certificate from a document to a process: each certified clause ships with its e-process; certification is \emph{alive} while $E_t^{(j)} \ge 1/\alpha$ and \emph{lapses} when it falls below. Alpha decay---a coupling classical theory does not model (App.~\ref{app:inverse})---thereby becomes \emph{certificate expiry}, a first-class, budgetable event rather than a silent drift of the truth. Nonparametric confidence sequences for the margins themselves \citep{howard2021cs} give the complementary estimation statement.

\section{The OOQI Framework}\label{sec:framework}

OOQI comprises three layers (Fig.~\ref{fig:arch}): a \textbf{demand layer} (the specification language of Def.~\ref{def:spec}), a \textbf{supply layer} (the design space of Def.~\ref{def:space}), and a \textbf{compiler} mediating them. We describe the compiler's three stages, then the arbitration of conflicting objectives.

\begin{figure}[t]
\centering
\begin{tikzpicture}[
  font=\small,
  layer/.style={draw=ruleline,fill=sand,rounded corners=1.5mm,minimum width=0.94\textwidth,minimum height=1.15cm,align=center},
  arr/.style={-{Stealth[length=2.6mm]},line width=0.9pt,accent},
  darr/.style={-{Stealth[length=2.6mm]},line width=0.9pt,warmgray,dashed}
]
\node[layer] (L1) at (0,0) {\textbf{\color{accent}Demand layer}: strategy profile specification $\Spec$\\
{\color{warmgray} 8 families $\cdot$ hard/soft clauses $\cdot$ priorities $\cdot$ interaction algebra (App.~\ref{app:spec}, \ref{app:algebra})}};
\node[layer] (L2) at (0,-2.1) {\textbf{\color{accent}Compiler}: translate $\rightarrow$ assemble $\rightarrow$ certify\\
{\color{warmgray} clause-indexed feedback $\cdot$ attribution $\cdot$ satisfaction certificate with ledger \& deflation}};
\node[layer] (L3) at (0,-4.2) {\textbf{\color{accent}Supply layer}: contracted design space $\D$\\
{\color{warmgray} 9 stages $\cdot$ typed module contracts $\cdot$ side-effect signatures $\cdot$ $|\D| \approx 8.8\times10^8$ (App.~\ref{app:space})}};
\draw[arr] (L1.south) -- (L2.north) node[midway,right=1pt,color=warmgray,font=\footnotesize]{specification};
\draw[arr] (L2.south) -- (L3.north) node[midway,right=1pt,color=warmgray,font=\footnotesize]{assembly directives};
\draw[darr] ([xshift=-3.2cm]L3.north) .. controls +(-0.5,0.7) and +(-0.5,-0.7) .. ([xshift=-3.2cm]L1.south)
  node[midway,left=1pt,color=warmgray,align=center,font=\footnotesize]{backtest evidence\\satisfaction feedback};
\end{tikzpicture}
\caption{The OOQI stack. Specifications flow down; assemblies are proposed and verified in the contracted design space; satisfaction evidence flows back for attribution and iteration.}
\label{fig:arch}
\end{figure}
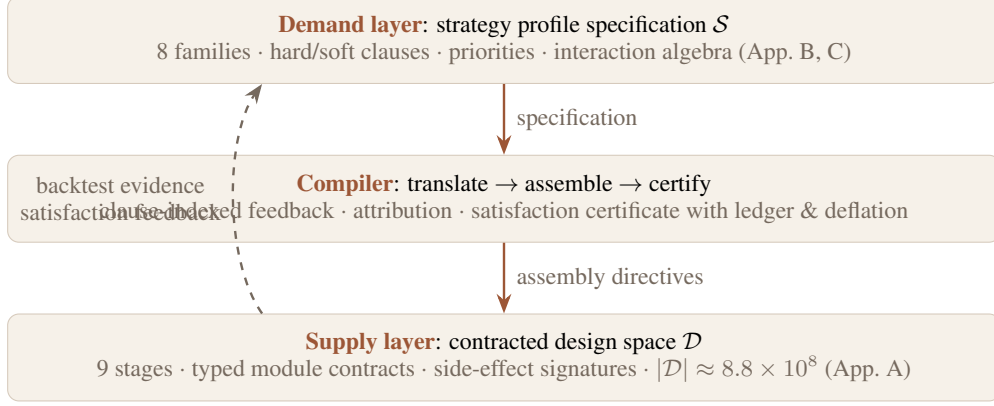

\subsection{Compiler stage I: translation (from specification to stage constraints)}

Each clause is compiled into (i) constraints on stage-module admissibility, (ii) structural obligations, and (iii) evaluation obligations. Examples: a \emph{style-purity} clause prunes non-neutralizing label modules (stage 2), mandates an exposure-projection or residualization post-processor (stage 7), and registers an attribution regression in evaluation; a \emph{turnover-budget} clause admits only smoothing/buffer post-processors and aim-portfolio-style portfolio modules \citep{garleanu2013dynamic}; a \emph{down-market resilience} clause obliges regime-aware risk overlays \citep{ang2002regime,shu2024jump} and adds regime-conditional coordinates to $\theta$. Translation is where domain knowledge enters: the clause-to-mechanism mapping is curated from the literature prototypes of \S\ref{sec:related} (App.~\ref{app:spec} gives the full table), and an LLM agent operationalizes it by proposing, for each clause, the module set consistent with the prototype's mechanism.

\subsection{Compiler stage II: assembly (constrained proposal and verification)}

Exhaustion of $|\D| \approx 8.8\times10^{8}$ is impossible; assembly is \emph{proposal-and-verification}. A proposer---LLM agent, Bayesian optimizer over the contracted space, or hybrid---emits candidate assemblies that (a) satisfy legality (Def.~\ref{def:space}) and (b) are predicted-feasible under the translated constraints. Each candidate is evaluated on the search split, its satisfaction profile computed, and failures \emph{attributed}: side-effect signatures (Def.~\ref{def:space}) localize the responsible stage (e.g., a turnover-clause failure with an unsmoothed post-processor indicts stage 7), and the attribution re-enters the proposal prior. This loop differs from scalar AutoML in exactly one structural respect: \emph{feedback is clause-indexed, not scalar}---the system learns \emph{which dimension} failed, not merely \emph{that} the score fell.

\subsection{Compiler stage III: certification}

The surviving assembly is re-evaluated on a temporally isolated holdout under the protocol of \S\ref{sec:protocol}, and the framework emits a \emph{satisfaction certificate}: the clause-wise profile $\big(\phi_j\big)_j$, the satisfaction rate with deflation $\widehat{\mathrm{SR}}_{\mathrm{def}}$, the search ledger ($N$ evaluated, search budget, seeds), the cost of specification $\mathrm{CoS}$, and per-clause confidence statements. A specification that cannot be met is reported as \emph{infeasible-with-evidence}---which clauses failed, at what magnitudes, and which relaxations the interaction algebra identifies as cheapest (App.~\ref{app:algebra}). Honest infeasibility is a first-class output, not a failure mode.

\subsection{Arbitrating conflicting objectives}

Requirements interact. Three mechanisms, in order of application: (i) \emph{feasibility first}---\textsc{hard} clauses define $\F(\Spec)$; if $\F = \emptyset$, the compiler returns the minimal-cost relaxation identified by the conflict graph; (ii) \emph{weighted soft optimization}---soft clauses contribute to a composite objective with principal-set weights; (iii) \emph{lexicographic fallback}---where a principal declares priorities, clauses are satisfied in order, accepting zero improvement on lower priorities once higher ones bind \citep{kim2020goalprog,deguest2015gbwm}. The interaction algebra (App.~\ref{app:algebra}) records, per clause pair, whether synergy or conflict was observed empirically across the search---turning folk knowledge (``turnover limits cost IC'') into measured structure.

\section{A Verification Protocol for Specification Satisfaction}\label{sec:protocol}

Proposition~\ref{prop:inflation} implies that satisfaction claims from searched spaces are presumptively inflated. The protocol below treats $\mathrm{SR}$ as a statistical object. Full technical details, including the deflation estimator, appear in App.~\ref{app:protocol}.

\begin{enumerate}[leftmargin=1.7em,itemsep=1.5pt]
\item \textbf{Search ledger.} Every evaluated assembly is logged: the effective trial count $N$ is a first-class reported quantity, exactly as DSR requires $n_{\mathrm{trials}}$ \citep{bailey2014dsr} and multiple-testing thresholds require test counts \citep{harvey2016crosssection}.
\item \textbf{Temporal isolation.} A terminal holdout period is never touched by search (no proposal, no early stopping, no attribution feedback), answering the nested-resampling requirement for unbiased model-selection evaluation \citep{cawley2010overfitting}.
\item \textbf{Random-assembly null.} Following NAS evaluation best practice \citep{yang2020nasfrustrating,sciuto2020evalnas}, $M$ uniformly random legal assemblies are evaluated: the compiler must beat the null not on score but on \emph{feasibility discovery}---finding assemblies in $\F(\Spec)$ at rates exceeding chance.
\item \textbf{Deflated satisfaction.} For each clause, IS satisfaction is corrected for the expected spurious-satisfaction maximum under the ledger $N$ (a clause-wise analogue of the expected-max-Sharpe correction \citep{bailey2014dsr}); the deflated rate $\widehat{\mathrm{SR}}_{\mathrm{def}}$ discounts clauses whose apparent satisfaction is within the null band.
\item \textbf{Cross-validated PBO over specifications.} Adapting CSCV \citep{bailey2017pbo}: within each symmetric split, assemblies are ranked by IS satisfaction; we compute the \emph{rank} of the IS-optimal assembly's satisfaction in the OOS half and report the fraction of splits where this OOS rank falls below the median---the \emph{probability of satisfaction backtest overfitting}---with $\mathrm{PBO} < 0.1$ as deployability guidance (App.~\ref{app:protocol}).
\item \textbf{Stress certification.} Each clause ships with a stress test in the spirit of underspecification certification \citep{damour2022underspecification}: purity clauses are re-examined under rotated factor premia; resilience clauses under synthetic unilateral-decline regimes; cost clauses across a cost-sensitivity curve rather than a point assumption \citep{novymarx2016costs,frazzini2018costs}.
\end{enumerate}

\section{Synthetic Illustration of Decision Logic}\label{sec:demo}

We demonstrate that the two paradigms of Def.~\ref{def:paradigms} \emph{select different objects} on a fully disclosed synthetic market (all parameters in App.~\ref{app:repro}; single fixed seed, no tuning). $N{=}300$ stocks, $T{=}208$ weeks (156 IS / 52 held-out); returns load on market, size, and dividend-style factors plus a mild-momentum$+$quality idiosyncratic alpha; the style premium is \emph{positive in-sample and reverses held-out}---a stylized style-rotation regime. The assembly space varies five binary stage choices (label horizon, signal neutralization, signal smoothing, portfolio width, turnover buffer): $|\D| = 32$ legal pipelines. The specification $\Spec$ comprises four hard clauses: style purity ($|\beta|_{\max} \le 6$ vs.\ size and dividend factors), dividend-index correlation ($|\rho| \le 0.15$), down-market resilience (annualized conditional excess $\ge 15\%$ in the worst market quintile), and a turnover budget ($\le 10\%$/week).\footnote{Coordinate units: the demonstration's $|\beta|_{\max}$ is measured in \emph{unstandardized} weekly regression-beta units of the synthetic DGP, whereas the worked specifications of App.~\ref{app:spec} and App.~\ref{app:neural} use \emph{standardized} betas (per unit of factor volatility, cap $0.05$). The two are different normalizations of the statistics space $\Theta$; per the convention of this paper, every coordinate of $\theta(d)$ must declare its normalization and estimation window (App.~\ref{app:repro}). IR magnitudes are likewise inflated relative to realizable weekly strategies---the DGP's alpha is deliberately strong so the selection gap is readable; the demonstration measures \emph{decision logic}, not achievable performance.}

\begin{figure}[t]
\centering
\includegraphics[width=\textwidth]{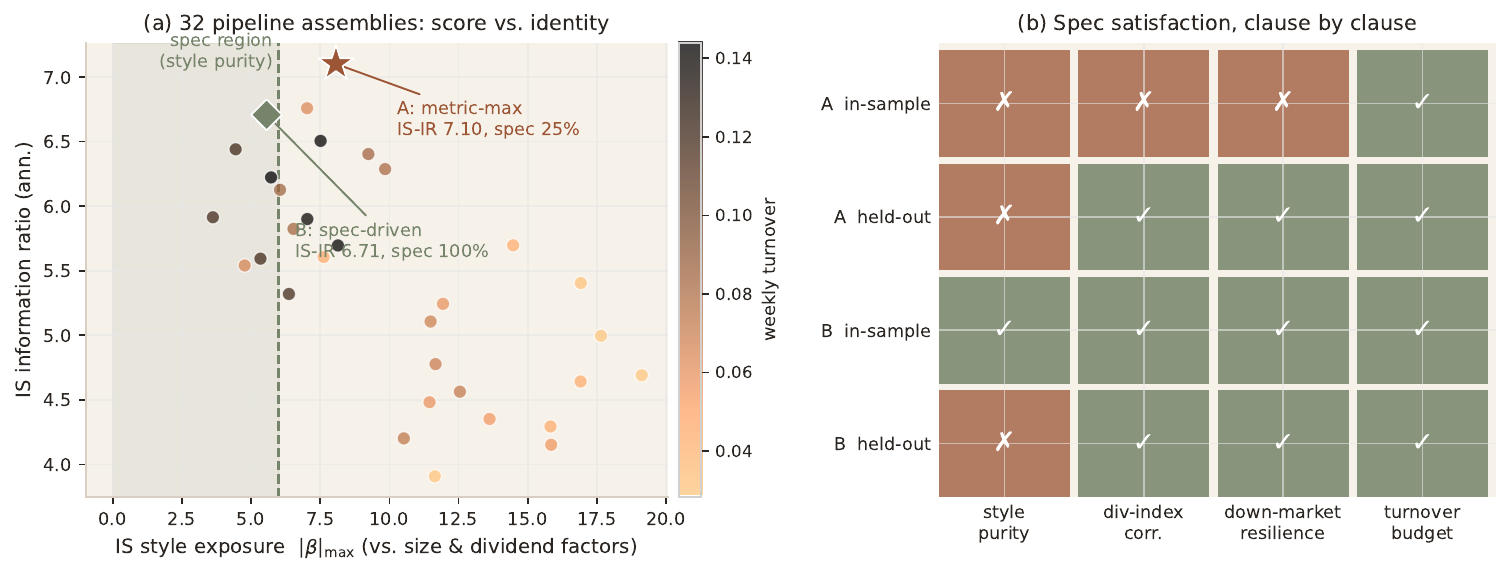}
\caption{(a) All 32 assemblies in the score--identity plane; shading marks the style-purity region. Result-oriented selection picks A (top in-sample IR, outside the spec region); objective-oriented selection picks B (best IR \emph{within} the feasible set). (b) Clause-wise satisfaction: A satisfies only the turnover clause in-sample (25\%); B satisfies the full specification (100\%), with one held-out clause marginal---the stochasticity our protocol (\S\ref{sec:protocol}) is built to price.}
\label{fig:demo}
\end{figure}

\begin{table}[t]
\caption{Result-oriented (A) vs.\ objective-oriented (B) selection on the synthetic market. Satisfaction rate (SR) is clause-wise out of 4. The cost of specification $\mathrm{CoS} = 7.10 - 6.71 = 0.39$ IR points ($5.5\%$) purchases full in-sample identity compliance---Propositions \ref{prop:contamination} and \ref{prop:cost} in miniature.}
\label{tab:demo}
\centering\small
\begin{tabular}{@{}lcccccc@{}}
\toprule
\rowcolor{sand} & \multicolumn{3}{c}{\textbf{In-sample (search)} } & \multicolumn{3}{c}{\textbf{Held-out (isolated)}} \\
\cmidrule(lr){2-4}\cmidrule(lr){5-7}
 & IR & SR & style $|\beta|_{\max}$ & IR & SR & style $|\beta|_{\max}$ \\
\midrule
A: result-oriented & \textbf{7.10} & 0.25 & 8.07 & 5.29 & 0.75 & 11.70 \\
\rowcolor{lightaccent}B: objective-oriented & 6.71 & \textbf{1.00} & 5.56 & 4.78 & 0.75 & \textbf{6.98} \\
\bottomrule
\end{tabular}
\end{table}

Findings (Fig.~\ref{fig:demo}, Table~\ref{tab:demo}): (i) \emph{The paradigms diverge}: only 2 of 32 assemblies satisfy $\Spec$; the in-sample score leader is not among them, exactly as Proposition~\ref{prop:contamination} predicts---its excess IR is partly rented from style exposure ($|\beta|_{\max} = 8.07$, dividend correlation $-0.21$, down-market conditional excess $12.7\%$). (ii) \emph{Identity has a measurable, bounded price}: B attains the 90th percentile of IR (rank 3 of 32) at $\mathrm{CoS} = 5.5\%$. (iii) \emph{Satisfaction is stochastic out of sample}: held-out, both selections satisfy 3 of 4 clauses---A fails purity more severely ($|\beta|_{\max} 11.70$ vs.\ B's $6.98$---a $1.7\times$ smaller raw exposure for B, and a $5.8\times$ smaller violation over the budget $6.00$: $0.98$ vs.\ $5.70$)---which is precisely why certificates report clause-wise confidence and deflation rather than binary verdicts. We emphasize what the demonstration does \emph{not} claim: no real-data alpha, no production validation; it validates \emph{decision logic}, and it does so with a disclosed data-generating process rather than a curated backtest.

\begin{table}[ht]
\centering\small
\caption{Claims and their backing (condensed; the full audit trail appears in App.~\ref{app:repro}).}
\label{tab:claims}
\begin{tabular}{@{}p{0.52\textwidth}p{0.41\textwidth}@{}}
\toprule
\textbf{Claim} & \textbf{Backed by} \\
\midrule
The two paradigms select different strategies & Prop.~\ref{prop:contamination}; Tables~\ref{tab:demo} and~\ref{tab:seeds} \\
Identity compliance has bounded, measurable cost & Prop.~\ref{prop:cost}; $\mathrm{CoS}=5.5\%$ (seed 7) \\
Apparent satisfaction inflates under search & Prop.~\ref{prop:inflation}; threshold sensitivity (App.~\ref{app:repro}) \\
Satisfaction survives resampling & \emph{partly}: clause-stability audit, App.~\ref{app:repro} (joint retention B $0.050$ vs.\ A $0.018$) \\
The optimal-assembly problem is NP-hard, with a constant-factor escape regime & Props.~\ref{prop:nphard}--\ref{prop:evalue} (\S\ref{sec:optimal}); App.~\ref{app:proofs} \\
Full protocol instantiation (deflated SR, PBO) & \emph{not yet}; roadmap (App.~\ref{app:roadmap}) \\
The compiler runs at $10^8$-assembly scale & \emph{not claimed}; design proposal \\
\bottomrule
\end{tabular}
\end{table}

\section{A Framework Walkthrough: Objects, Loci, and Accountability}\label{sec:walkthrough}

Sections \ref{sec:formal}--\ref{sec:demo} define the machinery on a synthetic stage. This section walks it once end-to-end at the level of a \emph{production-shaped} pipeline, without yet invoking production data: its purpose is to fix the locus of every object---which decisions live in the design space, which in the specification, which in the protocol---before the pre-registered production exhibit (roadmap R1) populates the empirical positions, marked [TBD] throughout. Nothing in this section reports a backtest.

\paragraph{The pipeline as choice points.} A strategy is not a monolithic object but a composition of choice points; each stage offers a small number of categories, and the design space is their legality-filtered product (Def.~\ref{def:space}). Table~\ref{tab:pipeline} instantiates this on a production-shaped equity pipeline. The unification the framework claims is at the level of \emph{form}---$\argmax_{d \in \F(\Spec)} g(\theta(d))$ at every stage and every product thereof---while heterogeneity of objectives is absorbed into an explicitly declared $\Spec$.

\begin{table}[t]
\caption{The pipeline as choice points: stages, category options (illustrative), and the requirement family (App.~\ref{app:spec}) that typically binds each stage. ($\dagger$) Algorithm choice sits outside the nine-stage reference grammar of App.~\ref{app:space}, whose S9 covers capacity/participation; we flag the gap rather than force the mapping.}
\label{tab:pipeline}
\centering\small
\begin{tabular}{@{}p{0.17\textwidth}p{0.45\textwidth}p{0.30\textwidth}@{}}
\toprule
\rowcolor{sand}\textbf{Stage} & \textbf{Category options (illustrative)} & \textbf{Binding family} \\
\midrule
Universe / data & all-market / liquidity-filtered / industry-constrained / listing-age filter & legality boundary \\
Signal layer & momentum / reversal / quality / low-volatility / dividend premium; horizon $h_1, h_2$ & source of $g$ (F1, F2) \\
Pre-processing & neutralization (industry/size/beta) on--off; winsorization scheme & style purity (F3) \\
Signal processing & smoothing window; decay scheme & turnover (F5) \\
Portfolio construction & width 20 / 50 / 100; equal-weight / optimized & structure (F6); regime (F4) \\
Turnover control & buffer band on--off; hard turnover cap & execution (F5) \\
Risk overlay & style cap; down-capture target; correlation constraint & hard clauses (F3, F4) \\
Execution$^{\dagger}$ & algorithm choice; participation-rate cap & execution (F5) \\
\bottomrule
\end{tabular}
\end{table}

\paragraph{A specification, declared before any ledger.} The demand side enters as a clause table: each clause carries a threshold, a hard/soft type, and a one-line commercial rationale (Table~\ref{tab:specwalk}). The discipline that makes this table evidentiary rather than decorative is \emph{priority}: thresholds are taken from the institution's standing risk-governance documents and frozen---with a hash---before any ledger is consulted (the protocol's pre-registration step, App.~\ref{app:protocol}). Thresholds read off the ledger are HARKing \citep{kerr1998harking} with extra steps.

\begin{table}[ht]
\caption{Specification template (anonymized). Thresholds will be scaled and shifted in any public artifact; rationales are reproduced verbatim---the demand layer is the framework's asset.}
\label{tab:specwalk}
\centering\small
\begin{tabular}{@{}p{0.20\textwidth}p{0.10\textwidth}p{0.08\textwidth}p{0.52\textwidth}@{}}
\toprule
\rowcolor{sand}\textbf{Clause} & \textbf{Threshold} & \textbf{Type} & \textbf{Commercial rationale} \\
\midrule
Style exposure $|\beta|_{\max}$ & [TBD] & hard & [TBD: principal's tolerance for style drift] \\
Downstream correlation $|\rho|$ & [TBD] & hard & [TBD: redundancy control vs.\ incumbent books] \\
Down-market capture & [TBD] & hard & [TBD: fiduciary drawdown behavior] \\
Turnover (weekly, one-sided) & [TBD] & hard & [TBD: cost budget and capacity] \\
Preference items & [TBD] & soft & [TBD] \\
\bottomrule
\end{tabular}
\end{table}

\paragraph{The locus of accountability.} Result-oriented selection, we argue, carries a structural accountability deficit: ``it ranked first in the backtest'' outsources judgment to a statistical artifact. Objective-oriented selection does not answer to results alone---its primitive is $\argmax_{d \in \F(\Spec)} g(\theta(d))$, optimization \emph{within} a contract---and it attaches accountability at three explicit points: (i) \emph{certified feasibility} of the selected assembly against the declared $\Spec$; (ii) an \emph{auditable selection trace}---the compiled feasible set, the winner's margins, and the runner-up; (iii) \emph{post-deployment monitoring} of each clause, with predefined remedies on violation. This is a relocation of responsibility, not its removal, and the relocation is conditional in a way we state plainly: certificates testify about ledgers, not about the future, and when the specification itself is badly written---thresholds politicized, clauses missing---compliance with it certifies the wrong thing well. Clause-level outcome accountability therefore attaches to two objects, not one: the assembly (did it satisfy $\Spec$?) and the specification writer (was $\Spec$ the right contract?). The difference from certification-washing (\S\ref{sec:discussion}) is operational rather than rhetorical: a ranking is a single number with no stated terms of engagement, whereas a certificate enumerates what was promised, the margin on each promise, and the evidence retained---giving an evaluator something a leaderboard never provides, namely the list of commitments against which to hold the author.

\section{Case Study: Decomposing a Macro--Meso--Micro Coupling Idea}\label{sec:case}

The demonstration of \S\ref{sec:demo} validates decision logic on five switches; this section exercises the framework on an idea of real complexity, contributed from practice under the same single-source, anonymized conditions disclosed in the ethics paragraph of \S\ref{sec:discussion}: a fully documented \emph{macro--meso--micro three-level coupling} stock-selection architecture (dual-channel regime state center; an interpretable cascade engine multiplying micro scores by meso and macro coefficients; an end-to-end engine with forced interaction features and FiLM conditioning; dynamic gating between engines; four guardrails). The source document prescribes one fixed architecture chosen by argument among four candidates. The framework's move is to restore that point to its space: every architectural decision becomes a switch, every engineering red line finds a formal home, and the champion must re-win with a certificate. Acts 1--4 are completed here; acts 5--7 (compilation, adjudication, audit) are protocol-defined and marked [TBD], pending the same pre-registered exhibit pipeline as \S\ref{sec:walkthrough}.

\paragraph{Act 1: the idea, stated colloquially.} ``Couple the macro cycle, the meso industry layer, and micro stock selection into one system, so that the regime view disciplines the alpha and the market's own structure corrects the regime view.'' The sentence is admirably clear and operationally meaningless---exactly the raw material the framework exists for.

\paragraph{Act 2: decomposition into the grammar.} Applying the rule \emph{``which stage's behavior does this idea change?''} maps the architecture onto the nine-stage grammar of App.~\ref{app:space} (Table~\ref{tab:casemap}). Three components exit the design space entirely: the data platform and the attribution layer are \emph{protocol and audit infrastructure}, not modules; the regime state center is an \emph{environment service}---an upstream context provider emitting $(\hat{p}_t, D_t, z_t)$ that clauses may reference but that no assembly contains. We are explicit about the cost of this choice: the end-to-end engine conditions on the state center's outputs, so assemblies identical in every switch but differing in environment configuration are behaviorally distinct yet indistinguishable in the registry; the environment configuration is therefore fixed ex ante, declared on the certificate, and kept \emph{outside} the accounted search width, as a disclosed residual degree of freedom rather than a hidden one. What remains is a set of switches.

\begin{table}[t]
\caption{Act 2: the documented architecture decomposed. Each fixed decision of the source design becomes a switch with variants; infrastructure components leave the design space.}
\label{tab:casemap}
\centering\small
\begin{tabular}{@{}p{0.24\textwidth}p{0.20\textwidth}p{0.46\textwidth}@{}}
\toprule
\rowcolor{sand}\textbf{Source component} & \textbf{Framework locus} & \textbf{Switch variants} \\
\midrule
Data platform (timestamp discipline) & protocol.json & --- (discipline, not a module) \\
Regime state center & environment service & \emph{no switch}: one configuration declared ex ante (channels, state count $K$, fusion rule); fixed before any ledger, reported on the certificate, and \emph{not} counted in $N$---a disclosed residual degree of freedom \\
Micro XGB layer & S5 predictor (+S2 label) & horizon 10/20/60d; predictor XGB / linear \\
Meso coefficient $C^{\mathrm{meso}}$ & S7 post-processing & on / off; auxiliary-input blend on / off \\
Macro coefficient $C^{\mathrm{macro}}$ & S7 post-processing & on / off; range fixed $\pm 0.1$ / $D_t$-dynamic \\
End-to-end engine (XGB-cross, FiLM) & S5 alternative module & absent / B1 / B2 / both \\
Dynamic gating $w_t$ & S7 fusion rule & static 0.85 / dynamic $(\hat{H}, \Delta\widehat{IC}, D_t)$ \\
Guardrail package & S9 risk overlay & full / partial / off (audit contrast) \\
Attribution layer & audit infrastructure & --- (certificate layer, not a module) \\
Portfolio construction & S8 & width; buffer; industry-deviation cap \\
Rolling training regime & S6 & window 2/3/5y; retrain monthly / weekly \\
\bottomrule
\end{tabular}
\end{table}

\paragraph{Grammar note (extension declared).} Two rows of Table~\ref{tab:casemap} exceed the reference grammar of App.~\ref{app:space}, whose S7 inventory contains post-processing modules but no coefficient modulation or fusion. That appendix declares itself ``a reference grammar, not a closed world''; this case exercises that clause. We extend S7 with a \emph{coefficient-modulation/fusion family} (multiplicative regime coefficients; engine-blending gates), and we count an assembly carrying both end-to-end variants as a single \emph{ensemble module} whose members are trained jointly---so the product structure of the design space and the counting convention for $N$ are preserved, with each fusion-family variant counting as one S7 choice. The compiler-mapping anchor of C8.4 (App.~\ref{app:cookbook}) refers to this extended family.

\paragraph{Act 3: registration.} The four candidate architectures the source document compares---pure cascade; pure end-to-end; static blend; dual-engine dynamic coupling---enter the registry as four \emph{named assemblies}; the switch grid of Table~\ref{tab:casemap} expands them into a legal-assembly space of size $N$ [TBD: exact count after legality filtering]. The audit width $N$ over the assembly space is thereby computable---gate one of \S\ref{sec:intro}---which is what gives subsequent deflation its declared input; the fixed environment configuration is reported alongside $N$, not folded into it.

\paragraph{Act 4: contracting.} The institution's demand side is declared as clauses before any ledger: the four standard hard clauses of Table~\ref{tab:specwalk} (thresholds anonymized, [TBD]) \emph{plus one clause this case contributes to the language}: an \emph{interpretability clause} $w_t \ge 0.70$, requiring the interpretable engine to retain at least 70\% of the fused score at all times (we write $w_t$ for the gating weight here; $\omega$ is reserved for market realizations elsewhere). The source document states this as an engineering red line; the framework formalizes it as a clause. Its family membership deserves care: the predicate constrains the \emph{fusion structure} rather than the return stream directly---a 70\%-weighted interpretable engine does not by itself make 70\% of the P\&L interpretable---so the clause's nearest structural neighbor is F6 (portfolio/fusion structure)---though no F6 clause currently covers fusion, a cookbook gap we flag rather than patch silently---while its rationale belongs to F8; we register it in F8 as the new entry C8.4 (App.~\ref{app:cookbook}). We claim no more than that it is a practitioner-grown instance of the transparency family---the clause cookbook itself being partly informed by the same single practitioner source disclosed in \S\ref{sec:discussion}, the practice/cookbook boundary is porous, and we flag it rather than celebrate it. The remaining red lines of the source design find their formal homes (Table~\ref{tab:redhome}). One boundary case is instructive: orthogonalization both transforms a signal and disciplines its measurement, and the arbitration rule we adopt is that \emph{whatever changes the object is a module; whatever governs the evidence is protocol}---the row is therefore split rather than forced into one home.

\begin{table}[ht]
\caption{Act 4: every red line of the source design finds a formal home. Not everything inviolable is a clause: some rules constrain legality, others govern measurement.}
\label{tab:redhome}
\centering\small
\begin{tabular}{@{}p{0.44\textwidth}p{0.18\textwidth}p{0.28\textwidth}@{}}
\toprule
\rowcolor{sand}\textbf{Red line (source document)} & \textbf{Home} & \textbf{Formal object} \\
\midrule
End-to-end engine must train on long panel & legality & contract on S6 modules \\
Interpretable engine dominates ($w_t \ge 0.70$) & specification & F8 clause C8.4 (nearest structural family F6---cookbook gap flagged) \\
Rolling retraining; extreme-regime holdout & protocol & protocol.json \\
Guardrail-triggered values logged, no manual edits & protocol & audit-trail rule \\
Market-reverse signal orthogonalization & legality / clause & signal-transforming module (S3/S7), F3 purity \\
Timestamp audit of reverse signals & protocol & measurement discipline \\
\bottomrule
\end{tabular}
\end{table}

\paragraph{Acts 5--7, protocol-defined and pending.} Compilation (legality filtering under the long-panel contract; feasible-set compilation under the clause set), adjudication ($\argmax g$ within $\F(\Spec)$, margin audit, certificate issuance under a pre-registered specification hash), and audit (the $D_t$-dynamic range switch and the guardrail package measured as cost--benefit switches on the ledger) follow the protocol of \S\ref{sec:protocol} and the pre-registration discipline of \S\ref{sec:walkthrough}. Two outcomes are equally informative: if the documented champion re-wins, its argument becomes evidence; if a sibling assembly wins, the framework has done its work once already. Until those data arrive, this section should be read as a protocol-defined walkthrough of a real architecture, not as empirical evidence. \emph{Anonymization}: this section exposes architecture (stage map, switches, clause types) and nothing else---no signal definitions, parameters, or performance data; module identities follow the S1--S9 abstraction.

\section{Discussion}\label{sec:discussion}

\paragraph{Is this a paradigm shift?} In the structural sense, yes: the reframing changes what counts as a legitimate question---from \emph{which strategy scores highest} to \emph{which assembly satisfies this identity}---and supplies the three institutions a paradigm requires: a language (Def.~\ref{def:spec}), a machine (the compiler), and a court (the protocol). It is equally important to state what kind of shift it is not. It is not a new alpha source, not a claim that specifications will be met, and not a repudiation of the existing literature: every clause stands on an established prototype, and the compiler stands on the AutoML and LLM-agent advances it reframes. Paradigm shifts are ratified by adoption, not proclamation; we therefore state the falsifiable adoption criteria: (a) specifications reused across institutions (language externalizes value); (b) certificates accepted as evidence by evaluators (statistics carry trust); (c) measured $\mathrm{CoS}$ entering practitioner decisions (the cost of identity becomes a budgeted quantity, as the transfer coefficient did twenty years ago \citep{clarke2002tc}).

\paragraph{Relation to the nearest neighbors.} Alpha-GPT takes an \emph{idea}---a directional hint about what might work; OOQI takes a \emph{specification}---a verifiable constraint set about what must hold. R\&D-Agent-Quant evolves hypotheses pulled by backtest numbers; OOQI constrains assemblies pushed by clauses. Declarative ML argues that users should state what they want in general machine learning \citep{molino2022declarative}; OOQI supplies the domain language, satisfaction statistics, and credibility protocol that make that argument actionable where the stakes---capital---demand them.

\paragraph{Limitations.} (i) The clause inventory, though broad (App.~\ref{app:spec}), cannot exhaust practitioner intent; the language must be extensible, and clause engineering is itself expertise. (ii) Satisfaction predicates over backtest statistics inherit every fragility of backtesting; our protocol mitigates but cannot abolish this. (iii) The compiler's proposal quality currently depends on LLM priors that may be shallow exactly where markets are subtle. (iv) The demonstration is synthetic; the research roadmap (App.~\ref{app:roadmap}) includes a live-system re-architecture case study under anonymized conditions.

\paragraph{Ethics, dual use, and conflicts.} A framework that makes identity constraints explicit also makes their absence conspicuous; it does not eliminate the misuse of backtests, and its certificates must never be read as guarantees of future performance. Three risks deserve explicit treatment. (i) \emph{Certification-washing}: the framework's own dual-use hazard is that a satisfaction certificate could lend a statistical veneer to overfit strategies---lowering epistemic standards while claiming to raise them. Our mitigations are structural: certificates carry deflation and retention intervals rather than binary verdicts, ledgers are hashable and auditable, and certificates expire with the evaluation window. (ii) \emph{Crowding and reflexivity}: if many principals order the same identity (e.g., ``pure selection alpha''), the alpha the certificates certify decays by the aggregate adoption; a certificate describes a strategy, not its crowd, and widespread adoption of any specification is itself a regime change the monitor must detect. (iii) \emph{Conflicts of interest}: the clause inventory of App.~\ref{app:spec} was informed by a qualitative diagnosis of a live production strategy operated by the author's employer (Shanghai Liangbai Technology Co., Ltd.); the diagnosis is used in anonymized form only, no production code, parameters, or performance data appear in this paper, and the planned production case study (roadmap R1) will be subject to the same constraint. Readers should treat the demand-layer grounding as single-source practitioner input---useful for concreteness, not a substitute for multi-institutional elicitation. \emph{Broader impact}: if adopted, the framework shifts quant research disclosure from ``results achieved'' to ``requirements met,'' which we believe benefits allocators, regulators, and end investors; its failure modes (bad forward operators, politicized thresholds) are discussed in \S\ref{sec:discussion}.

\section{Conclusion}\label{sec:conclusion}

We have argued that automated quantitative research has been solving the wrong problem, and we have built the alternative: a specification-driven framework---language, compiler, and statistics---in which investors order identities and systems deliver satisfaction certificates. The apparatus is deliberately conservative: every clause grounded in established measurement science, every claim deflated for search, every infeasibility reported honestly. If the field adopts even the modest habit of reporting \emph{which} clauses a strategy satisfies alongside \emph{how high} it scores, the reframing will have done its work.


\newpage
\appendix
\renewcommand{\thesection}{\AlphAlph{\value{section}}}

\section{The Supply-Side Design Space}\label{app:space}

This appendix defines the reference instantiation of Def.~\ref{def:space}: $K = 9$ stages, their module inventories, contracts, and side-effect signatures. The instantiation is a \emph{reference grammar}, not a closed world: institutions extend module sets without altering the formalism. Table~\ref{tab:space} gives the full taxonomy.

\begin{longtable}{@{}p{0.13\textwidth}p{0.34\textwidth}p{0.24\textwidth}p{0.22\textwidth}@{}}
\caption{Reference design space: stages, module inventories, interface contracts, and side-effect channels.}\label{tab:space}\\
\toprule
\rowcolor{sand}\textbf{Stage} & \textbf{Module inventory $\M_i$} & \textbf{Key contracts} & \textbf{Side-effect channels $\sigma(m)$} \\
\midrule
\endfirsthead
\toprule
\rowcolor{sand}\textbf{Stage} & \textbf{Module inventory $\M_i$} & \textbf{Key contracts} & \textbf{Side-effect channels $\sigma(m)$} \\
\midrule
\endhead
\bottomrule
\endfoot
\textbf{S1 Universe \& data}
& full market; exclude dual-board; exclude ST; index constituents (300/500/1000); liquidity-gated; industry-balanced; custom blacklist
& outputs security master $+$ corporate-action-adjusted panels
& capacity footprint; breadth (IR scaling $\propto\sqrt{\mathrm{breadth}}$) \\

\textbf{S2 Label}
& next-period raw return; cross-sectional rank; index excess; industry excess; neutralized residual (size$+$industry regression); multi-horizon composite; win/loss binary; quantile bucket
& horizon must match S6 training cadence; residual labels require exposure model channel
& purity leakage ($\uparrow$ for raw, $\downarrow$ for residual); turnover pressure ($\downarrow$ for multi-horizon) \\

\textbf{S3 Feature processing}
& raw; cross-sectional rank / z / quantile; industry demean; grouped NA fill; winsorize; time-series derivatives (momentum, vol, basis); HF down-sampling
& NA policy must match predictor tolerance
& rank preprocessing compounds rank labels (synergy) \\

\textbf{S4 Feature selection}
& none (full feed); IC-threshold; collinearity pruning; model importance; SHAP; $L_1$ sparsity; factor-family quotas
& selection cadence vs.\ training cadence
& interpretation clarity; regime fragility ($\uparrow$ for aggressive pruning) \\

\textbf{S5 Predictor}
& linear / elastic net; GBDT (XGBoost/LightGBM); MLP; GRU/LSTM; Transformer; temporal$+$cross-sectional hybrids (relational attention); gated mixture-of-experts; graph nets (industry/supply-chain edges)
& deep modules require GPU budget channel; relational modules require cross-section channel
& capacity footprint; interpretation clarity; turnover pressure (short-memory predictors $\uparrow$) \\

\textbf{S6 Training regime}
& single full-sample; rolling refit; warm-start chaining (parameter inheritance); recency weighting; regime-conditional weighting; adversarial de-styling; multi-task across horizons
& adversarial modules require style labels channel; chaining requires checkpoint contract
& stability of rankings ($\uparrow$ for chaining); turnover pressure ($\uparrow$ for recency weighting) \\

\textbf{S7 Post-processing}
& none; cross-sectional re-standardization; exposure residualization (industry/size); exposure projection onto constraint set; EMA smoothing; buffer-band hysteresis
& projection requires exposure model channel; buffers require persistent holdings channel
& turnover pressure ($\downarrow\downarrow$ smoothing/buffer); purity leakage ($\downarrow$ residualization/projection) \\

\textbf{S8 Portfolio construction}
& top-$N$ equal weight; top-$N$ score weight; stratified sampling; mean--variance; risk parity; tracking-error-constrained optimizer; industry deviation caps; position caps
& optimizers require covariance channel; TE constraints require benchmark channel
& capacity; drawdown resonance; industry/style deviation \\

\textbf{S9 Risk overlay}
& none; stop-loss; drawdown circuit breaker; regime de-risking signal; cost-aware aim-portfolio adjustment; capacity governor (participation cap)
& aim-portfolio requires continuous weight channel from S8
& drawdown resonance ($\downarrow$); turnover pressure ($\downarrow$ aim); cost drag \\

\end{longtable}

\paragraph{Size accounting.} With $|\M_1|{=}8$, $|\M_2|{=}12$, $|\M_3|{=}10$, $|\M_4|{=}8$, $|\M_5|{=}15$, $|\M_6|{=}8$, $|\M_7|{=}10$, $|\M_8|{=}12$, $|\M_9|{=}8$, the unconstrained product is $8.85\times 10^{8}$. These counts are \emph{families times their discrete hyperparameter grid variants}; Table~\ref{tab:space} prints the module families only, so the product is not derivable from the table alone---we state this explicitly to keep the size claim auditable. Contract compatibility is expected to remove a further fraction of tuples (the reference checker is a design proposal, App.~\ref{app:compiler}; the pruning fraction has not been measured), leaving $|\D|$ in the high $10^{8}$s by construction. Intra-module hyperparameters (tree depth, neutralization exposure sets, $N$, buffer widths) make the effective space effectively infinite for enumeration purposes and mandatory for ledger accounting (\S\ref{sec:protocol}).

\paragraph{Worked contract example.} Consider clause \emph{style purity} compiling to: S2 $\in$ \{residual label\} $\vee$ S7 $\in$ \{residualization, projection\}. The projection module's input contract requires a predictor exposing per-name scores (not merely a ranked list), ruling out post-hoc-only modules in S5; its side-effect signature records $\Delta$turnover $\approx +2\%$/week (rebalancing to constraint boundary), which the compiler prices against any active turnover clause. Contracts thus convert \emph{specification text} into \emph{typed graph constraints} over the assembly space.

\section{The Specification Language}\label{app:spec}

This appendix is the clause reference. Each clause is a triple $(\phi_j, \kappa_j, \pi_j)$ (Def.~\ref{def:spec}); we give the predicate, its measurement protocol, and the literature prototype on which measurement stands. Table~\ref{tab:clauses} lists the eight families.

\begin{longtable}{@{}p{0.15\textwidth}p{0.30\textwidth}p{0.28\textwidth}p{0.20\textwidth}@{}}
\caption{Specification language reference: families, representative clauses, measurement protocols, prototypes.}\label{tab:clauses}\\
\toprule
\rowcolor{sand}\textbf{Family} & \textbf{Clause (predicate form)} & \textbf{Measurement protocol} & \textbf{Prototype} \\
\midrule
\endfirsthead
\toprule
\rowcolor{sand}\textbf{Family} & \textbf{Clause (predicate form)} & \textbf{Measurement protocol} & \textbf{Prototype} \\
\midrule
\endhead
\bottomrule
\endfoot

\textbf{F1 Signal quality}
& $\overline{\mathrm{RankIC}} \ge \tau$; $\mathrm{ICIR} \ge \tau$; $\Pr[\mathrm{RankIC} > 0] \ge \tau$
& weekly cross-sectional Spearman of score vs.\ realized label; convention fixed per platform
& IC conventions of \citep{yang2020qlib} \\

\textbf{F2 Return--risk}
& ann.\ excess $\ge \tau$; $\mathrm{IR} \ge \tau$; max drawdown $\le \tau$; ann.\ volatility $\le \tau$
& net-of-cost excess vs.\ declared benchmark; deflated for ledger $N$
& \citep{bailey2014dsr,harvey2015backtesting} \\

\textbf{F3 Purity / style identity}
& $|\beta_f| \le \tau\ \forall f \in$ factor model; $|\rho(\cdot, \mathrm{style\ index})| \le \tau$; RBSA residual share $\ge \tau$; style-regression $R^2 \le \tau$
& returns regression on FF5/Carhart-style or Barra-style factors; constrained RBSA fit
& \citep{sharpe1992style,clarke2017purefactor,fama1993ff} \\

\textbf{F4 Regime robustness}
& conditional excess in down-market states $\ge \tau$; conditional drawdown $\le \tau$; per-year positive share $\ge \tau$; rolling-mean residence time below zero $\le \tau$
& states declared ex ante (e.g., market return quintiles, jump-model states); conditional evaluation on held-out states
& \citep{ang2002regime,shu2024jump} \\

\textbf{F5 Execution feasibility}
& ann.\ turnover $\le \tau$; net IR at cost $c$ $\ge \tau$ for cost grid $c$; participation $\le \tau$
& turnover from holdings diff; cost-sensitivity curve, not point cost
& \citep{garleanu2013dynamic,novymarx2016costs,frazzini2018costs} \\

\textbf{F6 Portfolio structure}
& holdings count $\in [\tau_l, \tau_u]$; industry deviation $\le \tau$; position cap $\le \tau$; TE vs.\ benchmark $\le \tau$
& holdings-level accounting per rebalance
& \citep{roll1992te,jorion2003te} \\

\textbf{F7 Temporal behavior}
& prediction horizon $= \tau$; holding-period profit concentration in first half $\ge \tau$; decay half-life $\ge \tau$
& event-study of post-formation return curve; rolling IC decay fit
& practitioner-informed; App.~\ref{app:algebra} \\

\textbf{F8 Transparency \& compliance}
& no ST / suspended / limit-up names at entry (hard); interpretability report exists; restricted-list conformance; interpretable-engine floor $w_t \ge \tau$ (C8.4, App.~\ref{app:cookbook})
& audit of selection filter chain; documentation artifact check
& mandate practice \\

\end{longtable}

\paragraph{Design notes.} (i) \emph{Falsifiability is mandatory}: a clause without a measurement protocol is not admitted to the language---this is the linguistic counterpart of Popper's criterion. (ii) \emph{Thresholds are principled, not arbitrary}: F2 thresholds scale with the ledger $N$ (Deflated-Sharpe logic), and F3 thresholds scale with the factor model's estimation error. (iii) \emph{Anonymized practitioner grounding}: the F3--F4--F7 inventories are informed by a qualitative diagnosis of a live selection-driven strategy (profit source purity; weak market-direction dependence; mild-momentum preference; early profit realization; nonlinear decay with self-repair; drawdown resonance with market crashes), abstracted into clauses without disclosing any strategy specifics. (iv) \emph{Compositionality}: clauses compose by conjunction within families and by the interaction algebra across families (App.~\ref{app:algebra}).

\paragraph{Example: a complete profile.} The following is a well-formed specification in the language, stated in practitioner diction and its compiled form:
\begin{quote}
\emph{``Pure stock-selection alpha: zero style exposure, uncorrelated with dividend index and with market direction; profitable in unilateral declines; mild-momentum and quality tilt; turnover within budget.''}
\[
\begin{aligned}
\Spec = \{ &\textsc{hard}:\ |\beta_f| \le 0.05\ (f \in \{\mathrm{mkt, size, div, momentum}\}), &&(F3)\\
&\textsc{hard}:\ |\rho_{\mathrm{div-index}}| \le 0.10,\ |\rho_{\mathrm{mkt}}| \le 0.10, &&(F3)\\
&\textsc{hard}:\ \mathbb{E}[\mathrm{excess} \mid \mathrm{down\ quintile}] \ge 0, &&(F4)\\
&\textsc{hard}:\ \mathrm{turnover} \le 30\%/\mathrm{mo}, &&(F5)\\
&\textsc{soft},\pi_1:\ \overline{\mathrm{RankIC}} \ge \tau_{\mathrm{ledger}}, &&(F1)\\
&\textsc{soft},\pi_2:\ \mathrm{IR} \ge \tau_{\mathrm{ledger}} \}\,. &&(F2)
\end{aligned}
\]
\noindent The first clause is stated in standardized units (see the units note in \S\ref{sec:demo}).
\end{quote}

\section{The Requirement Interaction Algebra}\label{app:algebra}

Requirements compose like building blocks, and blocks push against or reinforce one another. This appendix formalizes pairwise interaction and the arbitration semantics used by the compiler (\S\ref{sec:framework}).

\subsection{Interaction via the cost functional}

Recall $\mathrm{CoS}(\Spec) = g(\theta(d^*)) - g(\theta(d^*_\F))$ (Proposition~\ref{prop:cost}).

\begin{definition}[Pairwise interaction]\label{def:interaction}
For hard clauses $c_i, c_j$, define the incremental costs $\Delta_i = \mathrm{CoS}(\{c_i\})$, $\Delta_j = \mathrm{CoS}(\{c_j\})$, $\Delta_{ij} = \mathrm{CoS}(\{c_i, c_j\})$. Monotonicity of $\mathrm{CoS}$ (Proposition~\ref{prop:cost}) gives $\Delta_{ij} \ge \max(\Delta_i, \Delta_j)$, and the pair is
\[
\begin{cases}
\text{synergistic}\ (\oplus), & \Delta_{ij} \;\le\; \max(\Delta_i, \Delta_j) + \varepsilon,\\
\text{additive}\ (\circ), & \max(\Delta_i, \Delta_j) + \varepsilon \;<\; \Delta_{ij} \;\le\; \Delta_i + \Delta_j + \varepsilon,\\
\text{conflicting}\ (\ominus), & \Delta_{ij} \;>\; \Delta_i + \Delta_j + \varepsilon,
\end{cases}
\]
an exhaustive partition with tolerance $\varepsilon$ set by estimation noise. Synergy is the free-rider regime: satisfying the dearer clause buys the cheaper one. Conflict is super-additive cost: the pair costs more than the sum of its parts because the mechanisms that buy one clause (e.g., slow signals for turnover) destroy the other (e.g., fast signals for IC).
\end{definition}

\begin{table}[ht]
\caption{Interaction matrix over representative clauses (reference instantiation; entries validated qualitatively against mechanism structure and the demonstration's assembly correlations; institutional deployments should re-estimate empirically).}
\label{tab:interaction}
\centering\small
\begin{tabular}{@{}lcccccc@{}}
\toprule
\rowcolor{sand} & high IC & high IR & purity & resilience & low t/o & capacity \\
\midrule
high IC (F1)      & --- & $\oplus$ & $\ominus$ & $\circ$ & $\ominus$ & $\circ$ \\
high IR (F2)      & & --- & $\ominus$ & $\circ$ & $\circ$ & $\oplus$ \\
style purity (F3) & & & --- & $\oplus$ & $\circ$ & $\ominus$ \\
resilience (F4)   & & & & --- & $\oplus$ & $\circ$ \\
low turnover (F5) & & & & & --- & $\oplus$ \\
capacity (F6)     & & & & & & --- \\
\bottomrule
\end{tabular}
\end{table}

\paragraph{Reading the matrix.} IC $\ominus$ turnover is the classic signal-decay conflict: fast alpha decays in days, so satisfying F1 pressures holding periods below what F5 permits. Purity $\oplus$ resilience is genuine synergy: removing style beta removes the dominant channel by which unilateral declines contaminate the return stream. Purity $\ominus$ capacity reflects that neutralization concentrates the portfolio in the residual-alpha subspace, shrinking effective breadth. The matrix is \emph{empirical}: the compiler updates entries from search history (counts of assemblies jointly satisfying each pair), converting folklore into measured structure.

\subsection{Arbitration semantics}

Given $\Spec$ with conflict graph $(V, E_\ominus)$:
\begin{enumerate}[leftmargin=1.7em,itemsep=1pt]
\item \textbf{Feasibility check}: if the conflict graph contains a known-infeasible clique (e.g., \{high IC, low turnover, tight horizon\} beyond measured bounds), return the \emph{minimum-cost relaxation}: the clause whose removal maximizes restored feasibility per unit of principal's priority loss.
\item \textbf{Weighted soft mode}: maximize $g + \sum_{j:\kappa_j=\textsc{soft}} w_j \phi_j$ over $\F(\Spec_{\mathrm{hard}})$, weights set by the principal.
\item \textbf{Lexicographic mode}: satisfy soft clauses in priority order $\pi_1 \succ \pi_2 \succ \dots$, freezing each satisfied clause as hard before descending---the goal-programming semantics of \citep{kim2020goalprog} lifted from wealth targets to identity clauses.
\end{enumerate}

\subsection{Practitioner diction to specification: worked examples}

The language is designed for how practitioners actually speak. Three translations (mechanisms abstracted; no strategy specifics disclosed):

\paragraph{E1 (identity order).} \emph{``Earn stock-selection money, not direction money; dividend-index correlation zero; market-direction correlation zero.''} $\Rightarrow$ F3 hard clauses: exposure budget per factor; index-correlation caps; RBSA residual share $\ge 0.7$; evaluation adds attribution regression.

\paragraph{E2 (timing-of-profit order).} \emph{``Profits realize in the first half of the holding period; the second half is noise.''} $\Rightarrow$ F7 clause on profit concentration; compiler consequence: shorter label horizon in S2, EMA smoothing forbidden beyond the concentration window, buffer bands narrowed.

\paragraph{E3 (discipline order).} \emph{``Suspend execution when the rolling-mean return resides below zero; resume above the zero axis.''} $\Rightarrow$ F4/F7-style overlay clause; compiler consequence: S9 admits only regime-gated overlay modules; evaluation adds rolling-residence-time coordinate.

\section{Proofs}\label{app:proofs}

\paragraph{Proof of Proposition~\ref{prop:contamination}.}
Write $X_d = g(\theta(d))$ with means as in Assumption~\ref{ass:premium}. A violation occurs whenever $\max_{d \in \D_S} X_d > \max_{d \in \D_P} X_d$, since the dichotomy $\D = \D_S \cup \D_P$ then places the global argmax in $\D_S$. Fix any $d_S \in \D_S$ and set the midpoint threshold $t = (\mu_S + \mu_P)/2$. The violation event contains
\[
\{ X_{d_S} > t \} \;\cap\; \textstyle\bigcap_{d \in \D_P} \{ X_d < t \},
\]
because $\max_{\D_S} X \ge X_{d_S}$. Each excluded event is a one-sided sub-Gaussian tail at $u = \Delta/2$: $\Pr[X_{d_S} \le t] \le e^{-\Delta^2/(8s^2)}$, and $\Pr[X_d \ge t] \le e^{-\Delta^2/(8s^2)}$ for each $d \in \D_P$. A union bound over the $1 + |\D_P|$ excluded events---no independence assumption is needed---yields the display. The threshold for exceeding $1/2$ solves $(1+|\D_P|)\,e^{-\Delta^2/(8s^2)} \le 1/2$, and the limit statements are immediate. \hfill$\square$

\paragraph{Proof of Proposition~\ref{prop:cost}.}
Non-negativity: $\F(\Spec) \subseteq \D$ implies $\max_{\F} g \le \max_{\D} g$, with equality iff some $d^* \in \F$. Monotonicity: $\Spec \subseteq \Spec' \Rightarrow \F(\Spec') \subseteq \F(\Spec)$, hence $\mathrm{CoS}(\Spec') \ge \mathrm{CoS}(\Spec)$. For the bound: under Assumption~\ref{ass:premium} and the hypothesis $\F(\Spec)\subseteq\D_P$ (so that every feasible assembly is style-pure), a union bound over the finite classes gives, simultaneously for all $d$, $|g(\theta(d)) - \mathbb{E}g(\theta(d))| \le t$ with $t = 2s\sqrt{2\log|\D|}$, with probability $\ge 1 - 2|\D|^{-3}$ (each class contributes at most $|\D| e^{-t^2/2s^2} = |\D|^{-3}$). On that event $g(\theta(d^*)) \le \mu_S + t$---if the argmax comes from $\D_P$ the bound only improves---and $g(\theta(d^*_\F)) \ge \mu_P - t$. Therefore
\[
\mathrm{CoS}(\Spec) \;\le\; \mu_S - \mu_P + s\sqrt{2\log|\D_S|} + s\sqrt{2\log|\D_P|} \;\le\; \mu_S - \mu_P + 4s\sqrt{2\log|\D|},
\]
with the same probability. The analogy to the transfer coefficient \citep{clarke2002tc} is structural: both measure expected alpha forgone per unit of constraint, here priced in score per specification rather than in active return per portfolio constraint. \hfill$\square$

\paragraph{Proof of Proposition~\ref{prop:inflation}.}
Under the null, each assembly independently appears fully satisfying with probability $p = \prod_j q_j$ (clause estimates assumed independent conditional on the null; the union-bound variant without independence is $p \le \min_j q_j$ and the conclusion strengthens). Across $N$ evaluated assemblies,
\[
\Pr[\text{none satisfying}] = (1 - p)^N \;\Longrightarrow\; \Pr[\exists\ \text{apparent satisfaction}] = 1 - (1-p)^N \ge 1 - e^{-Np},
\]
which is $\ge 1/2$ once $N \ge \ln 2 / p$---exponential onset in the effective trials. Hence any certificate lacking the ledger $N$ cannot distinguish earned from manufactured satisfaction. \hfill$\square$

\paragraph{Proof of Proposition~\ref{prop:nphard}.}
Reduction from 3-SAT. Let $\Psi$ be a 3-CNF formula with variables $x_1,\dots,x_n$ and clauses $C_1,\dots,C_L$, $C_\ell = (\ell_{\ell 1} \lor \ell_{\ell 2} \lor \ell_{\ell 3})$. Construct a design space with $K = n$ stages, each offering two modules $\M_i = \{ m_i^0, m_i^1 \}$, and declare all contract pairs compatible, so every tuple $d \in \prod_i \M_i$ is legal: legality is trivial by construction. For each formula clause $C_\ell$, define a specification clause $c_\ell$ whose predicate inspects only the assembly's module set: $d \models c_\ell$ iff $d$ contains at least one \emph{literal module} of $C_\ell$, i.e., $m_i^1 \in d$ for some positive literal $x_i$ in $C_\ell$ or $m_i^0 \in d$ for some negative literal $\neg x_i$ in $C_\ell$. Each predicate is evaluable in $O(K)$. Then $d \models \Spec$ (all hard clauses satisfied) iff the assignment $x_i = \mathbf{1}[m_i^1 \in d]$ satisfies $\Psi$, so $\F(\Spec) \neq \emptyset$ iff $\Psi$ is satisfiable. The reduction is polynomial, hence NP-hardness; membership in NP (guess an assembly, check all predicates) gives NP-completeness when predicates are polynomial-time. Note that the construction uses no contract obstruction: the hardness lives entirely in clause--module interaction, not in legality or in $|\D|$. \hfill$\square$

\paragraph{Proof of Proposition~\ref{prop:submodular}.}
Write an assembly as its module set $S \subseteq \bigcup_i \M_i$ with $|S \cap \M_i| = 1$ per stage---exactly a partition matroid constraint. Under coverage structure, $f(S) = \sum_j w_j\, \mathbf{1}[ S \cap A_j \neq \emptyset ]$ is a weighted coverage function, hence monotone submodular: for $S \subseteq T$ and $m \notin T$, $f(S \cup \{m\}) - f(S) = \sum_{j:\, A_j \ni m,\, A_j \cap S = \emptyset} w_j \ge \sum_{j:\, A_j \ni m,\, A_j \cap T = \emptyset} w_j = f(T \cup \{m\}) - f(T)$. Claim (i) is the classical $1/2$ guarantee of greedy for monotone submodular maximization under a matroid constraint \citep{nemhauser1978}; claim (ii) is the $(1 - 1/e)$ guarantee of the continuous greedy algorithm with pipage rounding under any matroid \citep{calinescu2011}. Conflicting pairs introduce negative cross-terms (satisfying $c_i$ through mechanism $\mu$ lowers $\phi_j$), which break monotonicity and hence submodularity. \hfill$\square$

\paragraph{Proof of Proposition~\ref{prop:galois}.}
$X \subseteq \F(\Spec)$ iff every $d \in X$ satisfies $\Spec$, iff $\Spec$ is weaker than the strongest specification satisfied by all of $X$, iff $\Spec \preceq \mathcal{I}(X)$. Both maps are antitone: $\Spec_1 \preceq \Spec_2$ implies $\F(\Spec_2) \subseteq \F(\Spec_1)$ by definition of the order, and $X_1 \subseteq X_2$ implies $\mathcal{I}(X_2) \preceq \mathcal{I}(X_1)$ since strengthening the conjunction over assemblies weakens the common specification. Hence an antitone Galois connection \citep{cousot1977}, and the closure properties (i)--(iii) are the standard consequences: (i) $\Spec \preceq \mathcal{I}(\F(\Spec))$ is the unit of the connection; (ii) $\mathcal{I}(\{d\})$ is the closure of a singleton; (iii) $\F(\Spec) = \F(\mathcal{I}(\F(\Spec)))$, so feasible-set equality coincides with completion equality. \hfill$\square$

\paragraph{Proof of Proposition~\ref{prop:shadow}.}
Let $V(\tau) = \max_{d}\, g(\theta(d))$ subject to $h_j(d) = \theta_j(d) - \tau_j \ge 0$ be the value function of the relaxed problem. Under the stated regularity (the optimum is locally unique and the gradients of binding constraints are linearly independent, so strong duality and the KKT conditions hold with unique multipliers $\lambda^* \ge 0$), the envelope theorem for arbitrary choice sets \citep{milgrom2002envelope} yields $\partial V / \partial \tau_j = -\lambda^*_j$ (tightening $\tau_j$ shrinks the feasible set, so $V$ is nonincreasing in $\tau$). Since $\mathrm{CoS}(\Spec) = V(\tau_{-\infty}) - V(\tau)$ with the first term independent of $\tau$, $\partial\, \mathrm{CoS} / \partial \tau_j = \lambda^*_j \ge 0$. Monotonicity in specification strength (Proposition~\ref{prop:cost}) is the integral form of this marginal statement. \hfill$\square$

\paragraph{Proof of Proposition~\ref{prop:evalue}.}
An e-process $\{E_t\}$ under the null is a nonnegative supermartingale with $\mathbb{E}[E_0] \le 1$. Ville's maximal inequality for nonnegative supermartingales gives $\Pr[\exists t : E_t \ge 1/\alpha] \le \alpha \cdot \mathbb{E}[E_0] \le \alpha$, and the crossing event is uniform over $t$, including data-dependent stopping times---exactly the guarantee rolling certification requires and fixed-horizon $p$-values do not provide. A concrete construction: for weekly margin observations $m_t$ under the null of mean margin $\le 0$, the mixture likelihood ratio $E_t = \int \prod_{s \le t} \mathrm{d}P_{\eta}(m_s) / \mathrm{d}P_0(m_s)\, \mathrm{d}\rho(\eta)$ over alternatives $\eta > 0$ with mixing distribution $\rho$ is such a supermartingale \citep{ramdas2023savi,howard2021cs}. \hfill$\square$

\begin{remark}
The propositions are intentionally elementary. Their purpose is not depth but \emph{accountability}: each converts a qualitative anxiety of quantitative practice (style contamination, constraint cost, search overfitting) into a quantity that must be reported, bounded, or deflated. The framework's empirical content lives in the measurement of these quantities, not in their existence proofs.
\end{remark}

\begin{remark}
The second tier of results (Propositions \ref{prop:nphard}--\ref{prop:evalue}) is elementary in the same sense but different in kind: it prices the \emph{act of synthesis} rather than its objects---what assembly costs in computation (NP-hardness, and when a constant-factor escape exists), what a clause costs in equilibrium (shadow prices), what identity comes for free (Galois completion), and what continuous certification costs in multiplicity (nothing, under e-values). We read these as the framework's \emph{operating economics}: the theory a compiler needs before it may be fast, and the honesty it needs when it may not.
\end{remark}

\section{Verification Protocol: Technical Details}\label{app:protocol}

\paragraph{Deflated satisfaction rate.} Fix clause $c_j$ with threshold predicate $\phi_j$ and per-assembly statistic $\hat\theta_j(d)$, estimation variance $s_j^2$ (from block bootstrap over time). Under the null that assembly $d$ does not truly satisfy $c_j$, apparent satisfaction occurs with probability $q_j = \Pr[\hat\theta_j\ \text{crosses}\ \tau_j]$, bounded by sub-Gaussian tail $q_j \le \exp(-\Delta_j^2 / 2 s_j^2)$ where $\Delta_j$ is the true shortfall. Across the ledger of $N$ assemblies, the expected number of spurious satisfiers is $\le N q_j$; the clause is \emph{deflation-cleared} if the selected assembly's margin $\hat\Delta_j$ exceeds the expected-max null excursion $s_j \sqrt{2 \log N}$---the clause-wise analogue of the Deflated Sharpe benchmark \citep{bailey2014dsr}. The deflated rate is
\[
\widehat{\mathrm{SR}}_{\mathrm{def}} \;=\; \frac{1}{J} \sum_{j=1}^{J} \mathbf{1}\big[ \hat\Delta_j \;\ge\; s_j \sqrt{2 \log N} \big],
\]
i.e., only clauses clearing the search-width margin count. Note the direction of honesty: $\widehat{\mathrm{SR}}_{\mathrm{def}} \le \mathrm{SR}$ always.

\paragraph{Satisfaction-PBO.} Adapt CSCV \citep{bailey2017pbo} by replacing the performance vector with the satisfaction profile: split the evaluation span into $S$ blocks; for each combination of $S/2$ IS blocks, rank assemblies by IS $\mathrm{SR}$; compute the OOS $\mathrm{SR}$ rank of the IS winner; PBO is the fraction of combinations where the winner's OOS rank falls below median. We report PBO alongside $\widehat{\mathrm{SR}}_{\mathrm{def}}$; deployability guidance $\mathrm{PBO} < 0.1$.

\paragraph{Random-assembly null for feasibility discovery.} Draw $M$ legal assemblies uniformly; record $m_\F$, the count landing in $\F(\Spec)$. The compiler's feasibility-discovery skill is the binomial enrichment of its own hit count over $M p_\F$, $p_\F = m_\F / M$. This answers the NAS-style criticism \citep{yang2020nasfrustrating} transferred from accuracy to feasibility.

\paragraph{Stress certification per clause family.} F3: re-estimate exposures under rotated factor premia (sign-flipped style premium windows); F4: synthetic unilateral-decline regimes (drift-shifted market factor); F5: cost grid $c \in \{5, 10, 20, 50\}$ bps with the net-satisfaction curve reported rather than a point; F2: per-year segmentation. Certification follows the underspecification principle \citep{damour2022underspecification}: many assemblies tie in-sample; only stress behavior distinguishes them.

\paragraph{Multiple-testing thresholds for F1/F2.} Following \citep{harvey2016crosssection}, significance thresholds for IC-family clauses rise with the ledger: $\tau_{\mathrm{ledger}} = \tau_0 \cdot \sqrt{\max(1, \log N_{\mathrm{eff}})}$ with $N_{\mathrm{eff}}$ counting effective (overlap-discounted) trials; BHY-style FDR control is an acceptable alternative \citep{harvey2015backtesting}.

\section{Demonstration: Reproducibility}\label{app:repro}

\paragraph{Data-generating process.} $N{=}300$ names, $T{=}208$ weeks, seed $=7$ (single seed, disclosed; no seed selection). Weekly factor draws: market $\sim \mathcal N(0.0006, 0.02^2)$; size $\sim \mathcal N(0.0002, 0.012^2)$; dividend-style $\sim \mathcal N(0.0010, 0.012^2)$ in-sample and $\mathcal N(-0.0010, 0.012^2)$ held-out (declared style-rotation scenario). Name-level loadings: beta $\mathcal N(1, 0.25^2)$; size/dividend exposures $\mathcal N(0,1)$; persistent quality $\mathcal N(0,1)$. Returns: $r_{i,t} = \beta_i f^{\mathrm{mkt}}_t + 0.6\, \mathrm{size}_i f^{\mathrm{size}}_t + 0.6\, \mathrm{div}_i f^{\mathrm{div}}_t + \alpha_{i,t} + \varepsilon_{i,t}$, $\varepsilon \sim \mathcal N(0, 0.03^2)$, with idiosyncratic alpha $\alpha_{i,t} = 0.004\,[0.5\, z(\mathrm{mom}^{(4)}_{i,t}) + 0.5\, z(\mathrm{quality}_i)]$ (cross-sectional $z$). Transaction cost $0.15\%$ per unit turnover. Signal: noisy momentum (4-week lookback variants) plus noisy quality, per module recipe.

\paragraph{Assembly space and specification.} Five binary stage choices --- label horizon $h \in \{1,2\}$, signal neutralization (cross-sectional residualization on $[\beta, \mathrm{size}, \mathrm{div}]$), EMA smoothing ($0.6/0.4$), portfolio width $N \in \{20, 50\}$, turnover buffer band ($N{+}8$) --- giving $|\D| = 32$. Specification (all hard): style purity $|\beta|_{\max} \le 6$; dividend-index correlation $|\rho| \le 0.15$; down-market annualized conditional excess $\ge 15\%$ (worst market quintile); turnover $\le 10\%$/week. Evaluation split: weeks 4--156 search, 157--208 isolated holdout.

\paragraph{Full results.} Table~\ref{tab:assemblies} lists all assemblies. Feasibility count 2/32; result-oriented selection returns rank 1 (\texttt{h1/--/S/50/B}); objective-oriented returns rank 3 (\texttt{h2/N/S/50/B}). Held-out: A satisfies \{corr, resilience, turnover\}, fails purity at $|\beta|_{\max} = 11.70$; B satisfies the same three, fails purity marginally at $6.98$ vs.\ budget $6.00$---an honest illustration that satisfaction is a random variable, which is the raison d'\^etre of \S\ref{sec:protocol}.

\begin{longtable}{@{}clrrrrrrc@{}}
\caption{All 32 legal assemblies of the demonstration design space, ranked by in-sample IR. Module code: label horizon $h$; neutralization N; smoothing S; portfolio width; turnover buffer B. Feasible assemblies (satisfying all four hard clauses in-sample) are marked $\checkmark$.}\label{tab:assemblies}\\
\toprule
\rowcolor{sand} rank & modules & IR$_{\mathrm{IS}}$ & $|\beta|_{\max}$ & $\rho_{\mathrm{div}}$ & down-mkt & turnover & feasible \\
\midrule\endfirsthead
\toprule\rowcolor{sand} rank & modules & IR$_{\mathrm{IS}}$ & $|\beta|_{\max}$ & $\rho_{\mathrm{div}}$ & down-mkt & turnover & feasible \\\midrule\endhead\bottomrule\endfoot
 1 & \texttt{h1/--/S/50/B} & 7.10 & 8.07 & -0.211 & 12.7\% & 0.065 &  \\
 2 & \texttt{h1/N/S/50/B} & 6.76 & 7.03 & -0.038 & 10.5\% & 0.063 &  \\
\rowcolor{lightaccent} 3 & \texttt{h2/N/S/50/B} & 6.71 & 5.56 & -0.121 & 17.6\% & 0.064 & $\checkmark$ \\
 4 & \texttt{h1/--/--/50/--} & 6.51 & 7.51 & -0.193 & 13.4\% & 0.142 &  \\
 5 & \texttt{h1/N/--/50/B} & 6.44 & 4.44 & -0.117 & 20.5\% & 0.120 &  \\
 6 & \texttt{h1/--/S/50/--} & 6.40 & 9.24 & -0.240 & 15.1\% & 0.084 &  \\
 7 & \texttt{h2/N/S/50/--} & 6.29 & 9.85 & -0.065 & 12.6\% & 0.082 &  \\
 8 & \texttt{h1/N/--/50/--} & 6.22 & 5.72 & +0.036 & 14.8\% & 0.144 &  \\
 9 & \texttt{h1/N/S/50/--} & 6.13 & 6.05 & -0.109 & 14.1\% & 0.085 &  \\
 10 & \texttt{h1/--/--/50/B} & 5.91 & 3.62 & -0.064 & 17.0\% & 0.118 &  \\
 11 & \texttt{h2/--/--/50/--} & 5.90 & 7.03 & -0.173 & 16.2\% & 0.139 &  \\
 12 & \texttt{h2/--/S/50/--} & 5.82 & 6.54 & -0.161 & 12.3\% & 0.084 &  \\
 13 & \texttt{h2/--/S/20/--} & 5.70 & 14.47 & -0.226 & 26.8\% & 0.042 &  \\
 14 & \texttt{h2/N/--/50/--} & 5.70 & 8.14 & +0.059 & 9.6\% & 0.141 &  \\
 15 & \texttt{h1/N/S/20/--} & 5.61 & 7.63 & -0.088 & 19.1\% & 0.044 &  \\
 16 & \texttt{h2/N/--/50/B} & 5.59 & 5.34 & -0.145 & 14.3\% & 0.119 &  \\
\rowcolor{lightaccent} 17 & \texttt{h2/--/S/50/B} & 5.54 & 4.77 & -0.112 & 16.8\% & 0.065 & $\checkmark$ \\
 18 & \texttt{h1/--/S/20/B} & 5.40 & 16.92 & -0.246 & 27.8\% & 0.031 &  \\
 19 & \texttt{h2/--/--/50/B} & 5.32 & 6.37 & -0.123 & 12.0\% & 0.118 &  \\
 20 & \texttt{h1/N/--/20/B} & 5.24 & 11.94 & -0.172 & 27.0\% & 0.059 &  \\
 21 & \texttt{h2/--/--/20/--} & 5.11 & 11.49 & -0.091 & 18.6\% & 0.069 &  \\
 22 & \texttt{h2/N/S/20/B} & 5.00 & 17.65 & -0.032 & 20.9\% & 0.028 &  \\
 23 & \texttt{h2/N/--/20/--} & 4.78 & 11.67 & -0.093 & 26.1\% & 0.071 &  \\
 24 & \texttt{h1/N/S/20/B} & 4.69 & 19.12 & -0.203 & 26.1\% & 0.028 &  \\
 25 & \texttt{h2/N/S/20/--} & 4.64 & 16.91 & -0.172 & 21.4\% & 0.043 &  \\
 26 & \texttt{h1/N/--/20/--} & 4.56 & 12.55 & -0.148 & 23.7\% & 0.072 &  \\
 27 & \texttt{h1/--/--/20/B} & 4.48 & 11.45 & -0.149 & 17.8\% & 0.058 &  \\
 28 & \texttt{h2/N/--/20/B} & 4.35 & 13.62 & -0.212 & 14.6\% & 0.055 &  \\
 29 & \texttt{h1/--/S/20/--} & 4.29 & 15.82 & -0.218 & 22.8\% & 0.044 &  \\
 30 & \texttt{h1/--/--/20/--} & 4.20 & 10.52 & -0.161 & 24.6\% & 0.073 &  \\
 31 & \texttt{h2/--/--/20/B} & 4.15 & 15.85 & -0.220 & 16.5\% & 0.056 &  \\
 32 & \texttt{h2/--/S/20/B} & 3.91 & 11.64 & -0.158 & 19.4\% & 0.029 &  \\
\end{longtable}

\paragraph{Estimation conventions.} Exposure and correlation statistics are OLS estimates of weekly portfolio returns on the contemporaneous factor realizations over the search window (weeks 4--156); the benchmark for excess returns is the equal-weight universe mean; the down-market statistic conditions on the worst market-quintile weeks of the window; EMA smoothing is defined by $\tilde s_t = 0.6\,\tilde s_{t-1} + 0.4\, s_t$ (0.6 carry-over); the turnover buffer band retains incumbent names whose signal rank is within the top $N{+}8$; noise scales are those of the DGP paragraph above. The four clause thresholds were set once, before the ten-seed study, to discriminate within the demonstration space, and held fixed across all seeds; the sensitivity table below reports the consequence of moving each threshold on the seed-7 ledger (computed by re-applying predicates to Table~\ref{tab:assemblies}, no new backtests).

\paragraph{Threshold sensitivity.} Feasible-set size as one threshold varies, others fixed: style cap $4/5/5.5/6/8 \Rightarrow 0/1/1/\mathbf{2}/3$; correlation cap $0.10/0.125/\mathbf{0.15}/0.20 \Rightarrow 0/2/\mathbf{2}/2$; downside floor $14/15/17/18 \Rightarrow 2/\mathbf{2}/1/0$; turnover cap $0.06/0.065/\mathbf{0.10}/0.12 \Rightarrow 0/2/\mathbf{2}/4$ (bold $=$ paper's spec, feasible set $\{3,17\}$). Feasibility is a step function of the spec: each of the four thresholds sits within two notches of emptiness, which is exactly why the protocol reports margins and retention rather than binary verdicts. A fuller enumeration over six joint specifications appears in the threshold atlas (App.~\ref{app:atlas}).

\paragraph{Baselines on the same ledger.} To test whether simpler machinery already solves the selection problem, we evaluate three baselines on the identical ledger. (a) \emph{Weighted-penalty scalarization} $g - \lambda\sum_j \mathrm{viol}_j$: for $\lambda \le 0.13$ IR units per violated clause the winner is rank 1 (3 violations); for $\lambda > 0.13$ it switches to rank 3 (feasible)---the baseline \emph{can} recover B, but only if the penalty happens to be tuned across a crossover located at exactly the CoS ($0.39/3 = 0.13$), a quantity knowable only after the fact; penalties also cannot express lexicographic hard/soft semantics or emit infeasibility reports. (b) \emph{Pareto filtering} on (IR, violation count): the non-dominated set is $\{\text{rank 1 (7.10, 3)}, \text{rank 2 (6.76, 2)}, \text{rank 3 (6.71, 0)}\}$, and post-hoc hard-clause filtering of the front returns rank 3 $=$ B. On a 32-point space the Pareto baseline matches our \emph{selection}; what it does not supply is the specification's semantics (hard vs.\ soft, priorities), the certificate (margins, retention, deflation), or the infeasibility report---the differences that matter at $10^8$ assemblies and in audit. (c) \emph{Portfolio-layer-only constraints} (exposure-constrained optimizer with an unconstrained upstream) are not expressible in this five-switch space and remain a real-data comparison for future work. We state the honest conclusion: on small enumerable spaces, constrained selection and these baselines can agree on \emph{which} assembly; the framework's claim is about \emph{what can be stated, verified, and certified}, not about winning a selection contest on 32 points.

\paragraph{Clause stability audit (partial protocol instantiation).} We instantiated the protocol's stability layer on the seed-7 window: block bootstrap ($1000$ replications, 8-week blocks) of per-clause satisfaction. Retention probabilities (clause order: style, $|\rho|$, downside, turnover; joint $=$ all four): A $= [0.089, 0.294, 0.308, 1.000]$, joint $0.018$; B $= [0.231, 0.725, 0.211, 1.000]$, joint $0.050$. B's binding clause is style (point margin $0.44$ beta units; retention $0.231$), exactly as the deflation analysis of App.~\ref{app:protocol} predicts for a thin margin under $N=32$ search width; the margin-regularized selection of App.~\ref{app:inverse} would switch to the deeper-margin feasible sibling (rank 17; min margin $0.122$ vs.\ B's $0.073$). We report these numbers as the protocol's first instantiation on its own demonstration: full deflated-SR, satisfaction-PBO, and random-assembly-null values require the released script's resampling harness and remain roadmap items---the certificate fields are defined, and this paragraph begins filling them.

\paragraph{Code availability.} The demonstration is $\sim$120 lines of NumPy/pandas; the reference compiler (contract checker, proposer stub, certificate writer) is $\sim$800 lines of Python. Both are intended for release upon publication; the anonymized production case study (App.~\ref{app:roadmap}) is excluded from any release.

\section{Extended Related Work}\label{app:related}

This appendix expands \S\ref{sec:related} with per-work annotations; every entry was verified against arXiv/venue records in August 2026.

\subsection{Supply-side automation}
\textbf{CASH lineage.} Auto-WEKA \citep{thornton2013autoweka} first unified algorithm selection and hyperparameters in one hierarchical Bayesian-optimization problem; auto-sklearn \citep{feurer2015autosklearn} added meta-learned warm starts and post-hoc ensembles (JMLR 2020/2022 v2); TPOT \citep{olson2016tpot} evolves operator trees via genetic programming and is, to our knowledge, the only mainstream AutoML system with a built-in second objective (pipeline size); auto-PyTorch \citep{zimmer2021autopytorch} fuses NAS with CASH; CASH+ \citep{balef2025cashplus} generalizes CASH to heterogeneous modern pipelines but explicitly retreats to a two-layer selection over incompatible spaces---an admission that interface contracts (our Def.~\ref{def:space}) are the missing abstraction.
\textbf{NAS.} DARTS \citep{liu2019darts} relaxes discrete cells to continuous mixtures; ENAS \citep{pham2018enas} shares parameters across subgraphs; the 1000-paper meta-analysis \citep{white2023nas1000} confirms the field's objective is accuracy with occasional hardware-aware proxies, and its evaluation is fragile \citep{yang2020nasfrustrating,sciuto2020evalnas}. No NAS work accepts semantic specifications as input.
\textbf{Quant-side automation.} Qlib \citep{yang2020qlib} modularizes the full chain (data, model, backtest, portfolio, execution) but assembles by human workflow files; Quant 4.0 \citep{guo2024quant4} argues for end-to-end automation as a vision; AutoAlpha \citep{zhang2020autoalpha} and AlphaGen \citep{yu2023alphagen} search the factor stage alone with scalar fitness, as do AlphaAgent \citep{tang2025alphaagent} and AlphaForge. The modularity--search intersection inside one quantitative system remains unoccupied.

\subsection{LLM-agent quantitative R\&D}
R\&D-Agent \citep{yang2025rdagent} abstracts data-science R\&D as hypothesis--implement--feedback evolution; its quantitative version \citep{li2025rdagentquant} co-evolves factors and models against IC/return feedback, reporting strong backtest results at sub-\$10 research cost---the engineering nearest to our compiler, but with metrics, not specifications, as the steering signal. Alpha-GPT \citep{wang2023alphagpt} and its successor \citep{yuan2024alphagpt2} compile natural-language trading \emph{ideas} into alphas with human-in-the-loop review; we distinguish ideas (search-direction hints) from specifications (verifiable constraint sets) as disjoint input types. FunSearch \citep{romera2023funsearch} and AlphaEvolve \citep{novikov2025alphadevolve} prove the LLM$+$evaluator$+$evolution loop against programmatically computable scalar evaluators; OOQI can be read as replacing their scalar evaluator with a specification satisfier---plus the statistics to trust it. Decision-type agents \citep{zhang2024finagent,xiao2024tradingagents,yang2024finrobot} execute trades or reports, orthogonal to synthesis. The option-strategy DSL \citep{luo2026oql} is the closest specification-shaped neighbor, differing in requirement space (option legs/Greeks vs.\ identity profiles) and lacking satisfaction statistics. Declarative ML \citep{molino2022declarative} supplies the philosophical charter---separate \emph{what} from \emph{how}---which we instantiate for quantitative identity.

\subsection{Demand-side formalization}
Goals-based wealth management \citep{das2018gbwm} replaces scalar risk with goal-achievement probability; \citep{deguest2015gbwm} layers goals by importance with risk budgets; \citep{kim2020goalprog} gives the lexicographic goal-programming machinery our arbitration reuses. Mandate constraints \citep{roll1992te,jorion2003te} are the ancestors of F6 clauses; the transfer coefficient \citep{clarke2002tc} prices constraint cost and inspires Proposition~\ref{prop:cost}. Pure-factor portfolios \citep{clarke2017purefactor} are the mathematical prototype of style-zero clauses; RBSA \citep{sharpe1992style} and the Norwegian evaluation \citep{ang2009norway} supply the measurement and the evidentiary standard for alpha purity. Regime work \citep{ang2002regime,shu2024jump} grounds F4; multi-objective evolutionary methods \citep{deb2002nsga2,anagnostopoulos2009moea} supply Pareto machinery but flatten requirements into objective vectors, losing hard/soft semantics and priorities. Requirements engineering for ML \citep{vogelsang2019reml,rahimi2019reml} contributes vocabulary (elicitation, specification, verification) without financial instantiation.

\subsection{Evaluation credibility}
PBO/CSCV \citep{bailey2017pbo}, DSR \citep{bailey2014dsr}, Haircut-Sharpe \citep{harvey2015backtesting}, factor-zoo corrections \citep{harvey2016crosssection}, reality check/SPA \citep{white2000rc,hansen2005spa}, model-selection overfitting \citep{cawley2010overfitting}, NAS search-phase evaluation \citep{yang2020nasfrustrating,sciuto2020evalnas}, and underspecification certification \citep{damour2022underspecification}. None combines multi-clause satisfaction with search-width deflation; that combination---spec satisfaction $\times$ search-level multiple testing---is, to our knowledge as of August 2026, unoccupied, and constitutes the protocol's claim to novelty.

\section{Research Roadmap and Open Problems}\label{app:roadmap}

\paragraph{R1 Live re-architecture case study.} Apply the framework to an anonymized production strategy: translate its diagnosed identity into $\Spec$; re-architect within $\D$; certify. Pre-registered specifications, frozen before any search, are mandatory.
\paragraph{R2 Benchmark.} A public benchmark of (specification, data snapshot, ledger budget) triples with reference satisfaction rates would let compilers be compared like-for-like---the spec-satisfaction analogue of MLE-bench.
\paragraph{R3 Learning the interaction matrix.} Estimate Table~\ref{tab:interaction} empirically across markets and epochs; interactions are economic objects in their own right.
\paragraph{R4 Richer proposers.} LLM-prior-guided search is one point in a family; learned surrogates over the contracted space, and DARTS-like relaxations of module choice, are natural successors.
\paragraph{R5 Theory.} Partly addressed: \S\ref{sec:optimal} supplies NP-hardness with a submodular escape regime, Galois completion, shadow prices, and e-value certification. Still open: dependence-aware nulls tightening Proposition~\ref{prop:inflation}; approximation under supermodular conflict (when synergy makes joint satisfaction \emph{cheaper} than single-clause satisfaction); and a reflexivity theory in which adopted assemblies move the forward operator itself.
\paragraph{R6 Institutional grammar.} Clause libraries for mandates, ESG screens, and regulatory constraints would extend the language toward compliance-adjacent specifications.
\section*{Checklist of Claims and Where They Are Backed}
\begin{center}\small
\begin{tabular}{@{}p{0.5\textwidth}p{0.42\textwidth}@{}}
\toprule
\rowcolor{sand}\textbf{Claim} & \textbf{Backed by} \\
\midrule
Result-oriented selection systematically harvests style & P\ref{prop:contamination}; Table~\ref{tab:demo} row A \\
Identity compliance has bounded, measurable cost & P\ref{prop:cost}; $\mathrm{CoS}=5.5\%$ in Table~\ref{tab:demo} \\
Search inflates apparent satisfaction & P\ref{prop:inflation}; App.~\ref{app:protocol} \\
Satisfaction must be reported clause-wise with ledgers & \S\ref{sec:protocol}; Table~\ref{tab:demo} SR columns \\
The three-part combination is unoccupied in the literature & \S\ref{sec:related}; App.~\ref{app:related} (verified Aug.\ 2026) \\
The framework is deployable on real systems & \emph{not claimed}; roadmap R1 \\
\bottomrule
\end{tabular}
\end{center}

\clearpage
\section{The Clause Cookbook}\label{app:cookbook}

This appendix is the engineering reference for the specification language: thirty clauses across the eight families, each with its predicate, measurement protocol, the compiler's standard mechanism mapping, and the pitfalls that invalidate naive deployments. The cookbook is normative for implementers: a clause is only as good as its protocol.

\subsection{Family F1: Signal quality}

\begin{longtable}{@{}p{0.2\textwidth}p{0.25\textwidth}p{0.28\textwidth}p{0.2\textwidth}@{}}
\toprule
\rowcolor{sand}\textbf{Clause} & \textbf{Predicate} & \textbf{Protocol \& pitfalls} & \textbf{Compiler mapping} \\
\midrule
\endhead\bottomrule\endfoot
C1.1 RankIC level & $\overline{\mathrm{RankIC}} \ge \tau$ & Weekly cross-sectional Spearman of score vs.\ realized label; block-bootstrap CI. \emph{Pitfall}: never mix IC conventions (Pearson vs.\ Spearman) across clauses in one spec. & Raise predictor capacity (S5); recency weighting (S6); feature quota toward persistent families (S4) \\
C1.2 ICIR stability & $\mathrm{ICIR} = \overline{\mathrm{IC}}/\sigma_{\mathrm{IC}} \ge \tau$ & Weekly IC series; annualize with $\sqrt{52}$. \emph{Pitfall}: ICIR computed on overlapping labels inflates denominator correlation. & Smoothing (S7); longer label horizons (S2); ensemble training (S6) \\
C1.3 IC sign consistency & $\Pr[\mathrm{RankIC}_t > 0] \ge \tau$ & Share of weeks with positive RankIC. \emph{Pitfall}: cheap to satisfy with tiny positive mean; pair with C1.1. & Buffer/hysteresis (S7); multi-task horizons (S6) \\
C1.4 Tail IC & $\mathrm{RankIC}$ in top score decile $\ge \tau$ & IC restricted to the tradeable top decile---the decile the portfolio actually buys. & Score-weighted portfolio (S8); ranking losses emphasizing top ranks \\
\end{longtable}

\subsection{Family F2: Return--risk}

\begin{longtable}{@{}p{0.2\textwidth}p{0.25\textwidth}p{0.28\textwidth}p{0.2\textwidth}@{}}
\toprule
\rowcolor{sand}\textbf{Clause} & \textbf{Predicate} & \textbf{Protocol \& pitfalls} & \textbf{Compiler mapping} \\
\midrule
\endhead\bottomrule\endfoot
C2.1 Net excess return & $\mathbb{E}[r_{\mathrm{net}} - r_{\mathrm{bm}}] \ge \tau$ & Net of declared cost model; benchmark declared ex ante. \emph{Pitfall}: threshold must be $\tau_{\mathrm{ledger}}$-scaled (App.~\ref{app:protocol}). & Score$\to$weight optimizers (S8); breadth management (S1) \\
C2.2 Information ratio & $\mathrm{IR} \ge \tau$ & Excess mean / tracking error, annualized; report with DSR deflation \citep{bailey2014dsr}. & Risk-parity weighting (S8); regime overlays (S9) \\
C2.3 Max drawdown & $\mathrm{MDD} \le \tau$ & Peak-to-trough on net curve; per-year MDD also reported. \emph{Pitfall}: MDD is a single-path statistic---stress across regimes. & Circuit breakers (S9); de-risking overlays (S9) \\
C2.4 Volatility cap & $\sigma_{\mathrm{ann}} \le \tau$ & Realized vol of net returns. & Vol targeting (S8); lower-beta universe (S1) \\
\end{longtable}

\subsection{Family F3: Purity / style identity}

\begin{longtable}{@{}p{0.2\textwidth}p{0.25\textwidth}p{0.28\textwidth}p{0.2\textwidth}@{}}
\toprule
\rowcolor{sand}\textbf{Clause} & \textbf{Predicate} & \textbf{Protocol \& pitfalls} & \textbf{Compiler mapping} \\
\midrule
\endhead\bottomrule\endfoot
C3.1 Factor exposure budget & $|\beta_f| \le \tau_f\ \forall f \in \mathcal{F}$ & Time-series regression of portfolio excess on declared factor model (FF5/Carhart-style or Barra-style); $\tau_f$ must exceed estimation error $2\,\mathrm{se}(\hat\beta_f)$. \emph{Pitfall}: silent factor rotation---re-estimate under rotated premia (App.~\ref{app:protocol}). & Residual labels (S2); exposure residualization/projection (S7); adversarial de-styling (S6) \\
C3.2 Style-index correlation & $|\rho(r, r_{\mathrm{idx}})| \le \tau$ & Correlation vs.\ declared style indices (e.g., dividend, size, growth indices). \emph{Pitfall}: correlation is not exposure---require both C3.1 and C3.2. & Same mechanism chain as C3.1; index added to exposure model \\
C3.3 Selection share (RBSA) & $R^2_{\mathrm{style}} \le \tau$, i.e., residual share $\ge 1-\tau$ & Constrained RBSA fit \citep{sharpe1992style}; residual = selection. \emph{Pitfall}: RBSA window choice changes answers---fix window ex ante. & Whole purity chain; evaluation-layer attribution mandatory \\
C3.4 Market-direction independence & $|\rho(r, r_{\mathrm{mkt}})| \le \tau$ and $|\beta_{\mathrm{mkt}}| \le \tau'$ & Two predicates: correlation \emph{and} beta---a low-vol portfolio can have high $\rho$ with low $\beta$. & Beta-neutral overlay (S9); dollar-neutral construction (S8) \\
\end{longtable}

\subsection{Family F4: Regime robustness}

\begin{longtable}{@{}p{0.2\textwidth}p{0.25\textwidth}p{0.28\textwidth}p{0.2\textwidth}@{}}
\toprule
\rowcolor{sand}\textbf{Clause} & \textbf{Predicate} & \textbf{Protocol \& pitfalls} & \textbf{Compiler mapping} \\
\midrule
\endhead\bottomrule\endfoot
C4.1 Down-market conditional excess & $\mathbb{E}[r_{\mathrm{ex}} \mid s_t \in \mathcal{S}_{\mathrm{down}}] \ge \tau$ & States declared ex ante (market-return quintiles or jump-model states \citep{shu2024jump}); conditional evaluation on held-out states. \emph{Pitfall}: few down-states $\Rightarrow$ wide CIs---report per-state counts. & Regime-aware overlays (S9); defensive universe tilt (S1); purity chain (removes the main contamination channel) \\
C4.2 Conditional drawdown & $\mathrm{MDD} \mid \mathcal{S}_{\mathrm{down}} \le \tau$ & MDD computed over down-state subsequences. & Circuit breaker with state gating (S9) \\
C4.3 Per-year positive share & $\Pr[\mathrm{ann.\ excess}_y > 0] \ge \tau$ & Calendar-year segmentation; \emph{pitfall}: short histories give binomial noise---report exact binomial CI. & Long-horizon label composites (S2); ensemble (S6) \\
C4.4 Failure-residence time & $\mathbb{E}[\,\text{consec.\ wks}<0\,] \le \tau$ & Rolling-mean residence below zero; practitioner-informed monitoring clause. & Overlay suspension rule (S9); this clause often maps to \emph{operational} discipline rather than architecture---the compiler flags it as such \\
\end{longtable}

\subsection{Family F5: Execution feasibility}

\begin{longtable}{@{}p{0.2\textwidth}p{0.25\textwidth}p{0.28\textwidth}p{0.2\textwidth}@{}}
\toprule
\rowcolor{sand}\textbf{Clause} & \textbf{Predicate} & \textbf{Protocol \& pitfalls} & \textbf{Compiler mapping} \\
\midrule
\endhead\bottomrule\endfoot
C5.1 Turnover budget & $\mathrm{tov} \le \tau$ per period & Holdings-diff one-way turnover, averaged. \emph{Pitfall}: buffer bands distort entry timing---verify vs.\ price impact model. & EMA smoothing, buffer bands (S7); aim-portfolio adjustment (S9) \citep{garleanu2013dynamic} \\
C5.2 Cost robustness & $\mathrm{IR}_{\mathrm{net}}(c) \ge \tau\ \forall c \in \mathcal{C}$ & Cost grid $\mathcal{C} = \{5,10,20,50\}$ bps; net-satisfaction \emph{curve} reported, not a point \citep{novymarx2016costs,frazzini2018costs}. & Lower-signal-decay modules (S2/S7); wider portfolios (S8) \\
C5.3 Capacity & participation rate $\le \tau$ & Name-level traded value vs.\ ADV; \emph{pitfall}: ADV in stress regimes differs---stress-test with down-regime ADV. & Liquidity-gated universe (S1); wider $N$ (S8) \\
\end{longtable}

\subsection{Family F6: Portfolio structure}

\begin{longtable}{@{}p{0.2\textwidth}p{0.25\textwidth}p{0.28\textwidth}p{0.2\textwidth}@{}}
\toprule
\rowcolor{sand}\textbf{Clause} & \textbf{Predicate} & \textbf{Protocol \& pitfalls} & \textbf{Compiler mapping} \\
\midrule
\endhead\bottomrule\endfoot
C6.1 Breadth & $N_{\mathrm{hold}} \in [\tau_l, \tau_u]$ & Holdings count per rebalance; breadth enters IR scaling. & S8 width parameter \\
C6.2 Industry deviation & $\max_k |w_k - w_k^{\mathrm{bm}}| \le \tau$ & Industry mapping declared (SW/CSI/GICS); deviation per rebalance. & Industry-capped optimizer (S8) \\
C6.3 Position cap & $\max_i w_i \le \tau$ & Name cap per rebalance. & Capped score-weighting (S8) \\
C6.4 Tracking error & $\mathrm{TE} \le \tau$ & Std of active return; TE-constrained frontier geometry \citep{roll1992te,jorion2003te}. & TE-constrained optimizer (S8) \\
\end{longtable}

\subsection{Family F7: Temporal behavior}

\begin{longtable}{@{}p{0.2\textwidth}p{0.25\textwidth}p{0.28\textwidth}p{0.2\textwidth}@{}}
\toprule
\rowcolor{sand}\textbf{Clause} & \textbf{Predicate} & \textbf{Protocol \& pitfalls} & \textbf{Compiler mapping} \\
\midrule
\endhead\bottomrule\endfoot
C7.1 Profit concentration & $\frac{\mathbb{E}[\text{cum.\ ret.\ weeks }1..h/2]}{\mathbb{E}[\text{cum.\ ret.\ weeks }1..h]} \ge \tau$ & Event study of post-formation return curve; if concentration is high, holding to expiry destroys value. & Shorter label horizon (S2); smoothing window limited to concentration window (S7); front-half exit rule (S9, operational) \\
C7.2 Signal decay half-life & $\mathrm{HL}_{\mathrm{IC}} \ge \tau$ & Exponential fit of IC-by-lag curve. \emph{Pitfall}: half-life drifts with regime---report by state. & Horizon matching between label and holding period (S2/S8) \\
C7.3 Horizon identity & prediction horizon $= \tau$ & Declared and enforced contract between S2 and S6/S8. & Contract checker (App.~\ref{app:compiler}) \\
\end{longtable}

\subsection{Family F8: Transparency \& compliance}

\begin{longtable}{@{}p{0.2\textwidth}p{0.25\textwidth}p{0.28\textwidth}p{0.2\textwidth}@{}}
\toprule
\rowcolor{sand}\textbf{Clause} & \textbf{Predicate} & \textbf{Protocol \& pitfalls} & \textbf{Compiler mapping} \\
\midrule
\endhead\bottomrule\endfoot
C8.1 Tradability filters & no ST / suspended / limit-up-at-entry names & Audit of the selection filter chain per rebalance. \emph{Pitfall}: look-ahead in suspension flags---use point-in-time flags. & Filter chain (S1/S8); audit harness (evaluation) \\
C8.2 Interpretability artifact & attribution report exists and is reproducible & Regenerate report from logged artifacts; \emph{pitfall}: post-hoc narratives are not artifacts. & Importance/SHAP channels (S4/S5); logging contract \\
C8.3 Restricted-list conformance & no trades in restricted names & Point-in-time restricted list; audit per order. & Blacklist module (S1) \\
C8.4 Interpretable-engine floor & gating weight $w_t \ge \tau$ (e.g., $0.70$) at all rebalance times & Logged gating weights per rebalance; violation = any $w_t < \tau$. \emph{Pitfall}: the predicate constrains fusion structure, not P\&L attributability---pair with C8.2 for the latter; nearest structural family F6 currently has no fusion clause (cookbook gap, flagged). & Gating/fusion module (S7 fusion family, \S\ref{sec:case} extension); certificate field \texttt{gating\_log} \\
\end{longtable}

\paragraph{Cookbook usage.} A specification is assembled by selecting clauses and instantiating thresholds; the cookbook's compiler-mapping column gives the translation stage its default actions, which the proposer may override with justification logged. The pitfalls column is the institutional memory of the language: every pitfall listed was chosen because it is a documented way practitioners deceive themselves.

\clearpage
\section{Formal Grammar of the Specification Language}\label{app:grammar}

The specification language is a small typed DSL. We give its concrete syntax in BNF, its typing rules, and its semantics.

\subsection{Concrete syntax}

\begin{verbatim}
spec        := "spec" ident "{" { clause } "}"
clause      := "clause" ident ":" predicate
               [ "hard" | "soft" ] [ "priority" int ]
               [ "protocol" ident ] [ "family" familyid ]
predicate   := atom { ("and" | "or") atom }
atom        := coord rel thresh [ "window" range ] [ "state" stateid ]
coord       := "rankic_mean" | "icir" | "excess_ann" | "ir" | "mdd"
             | "beta_abs" "(" factorid ")" | "corr_abs" "(" indexid ")"
             | "rbsa_resid" | "cond_excess" "(" stateid ")" | "tov"
             | "participation" | "holdings" | "industry_dev" | "te"
             | "profit_conc" "(" int ")" | "ic_halflife" | ...
rel         := "<=" | ">=" | "=="
thresh      := number | "ledger" "(" number ")" | "se" "(" number ")"
stateid     := "down_quintile" | "jump_low" | "jump_high" | ident
familyid    := "F1" | ... | "F8"
\end{verbatim}

\subsection{Typing and validity rules}

\begin{enumerate}[leftmargin=1.7em,itemsep=1.5pt]
\item \textbf{Protocol completeness} (Popper rule): every clause must bind a \texttt{protocol} from the cookbook (App.~\ref{app:cookbook}); a clause without a protocol does not compile.
\item \textbf{Threshold discipline}: \texttt{ledger(k)} marks a threshold that scales as $k\sqrt{\max(1,\log N_{\mathrm{eff}})}$; \texttt{se(k)} marks one floored by $k$ standard errors of the underlying estimator. Hard F1/F2 clauses \emph{must} use \texttt{ledger} thresholds; hard F3 clauses \emph{must} respect the \texttt{se} floor.
\item \textbf{State legitimacy}: any \texttt{state} reference must resolve to a state declaration made ex ante (quintile rule or fitted jump model frozen before search).
\item \textbf{Contract reachability}: each clause's compiler mapping must reach at least one stage whose module set is non-empty after prior pruning; otherwise the spec is \emph{statically infeasible} and rejected before any backtest.
\item \textbf{Priority well-foundedness}: priorities form a total preorder; cycles in stated preferences are rejected.
\end{enumerate}

\subsection{Semantics}

The denotation of a spec is the pair $(\F(\Spec), g_\Spec)$: a feasible set and a scalarization. Hard clauses conjoin into $\F$; soft clauses define $g_\Spec = g + \sum_j w_j \phi_j$ under weighted mode, or the lexicographic order under priority mode. Satisfaction $d \models \Spec$ and $\mathrm{SR}$ are as in Def.~\ref{def:spec}. Two specs are \emph{equivalent} if they denote the same $(\F, g_\Spec)$; the parser normalizes for equivalence checking (clause reordering, threshold algebra on identical coordinates).

\section{Compiler Algorithms (Design Proposal---Not Yet Executed End-to-End)}\label{app:compiler}

We give pseudocode for the four compiler components. The reference implementation is Python; complexity notes assume $|\D|$ after contract pruning $\ll 10^{9}$ (typically $10^{4}$--$10^{6}$ after clause pruning).

\subsection{Contract checker}

\begin{verbatim}
def check_legality(assembly, contracts):
    for i in 1..K-1:
        if not compatible(output(assembly[i]), input(assembly[i+1])):
            return INFEASIBLE(stage=i, reason=contract_mismatch)
    return FEASIBLE
\end{verbatim}

Runs in $O(K)$ per assembly; executed once per proposed candidate and statically over the pruned space.

\subsection{Proposer (LLM-prior guided)}

\begin{verbatim}
def propose(spec, history, budget):
    pruned = prune_by_translation(spec)          # stage-I constraints
    prior  = llm_prior(spec, pruned, literature) # mechanism knowledge
    cand   = []
    while budget.remains():
        c = sample(prior, temperature=explore(history))
        if c in history or not check_legality(c): continue
        cand.append(c); history.log(c)           # ledger bookkeeping
        if len(cand) == batch: yield cand; cand = []
\end{verbatim}

The LLM prior is queried with the spec text, the pruned module inventory, and cookbook mappings; the temperature schedule anneals from exploration to exploitation as attribution accumulates. A Bayesian-optimization surrogate over module embeddings is a drop-in replacement.

\subsection{Attribution engine}

\begin{verbatim}
def attribute(assembly, profile, spec):
    failed = [c for c in spec if not profile.satisfies(c)]
    for c in failed:
        scores = {}
        for stage in c.mechanism_chain():
            m = assembly[stage]
            scores[stage] = dot(m.side_effect, c.affected_coords)
        report(c, argmax_descending(scores))     # rank suspects
    return attribution_report
\end{verbatim}

Attribution converts scalar regret into \emph{stage-indexed suspicion}: the side-effect signatures (Def.~\ref{def:space}) supply the gradient-free localization signal.

\subsection{Certifier}

\begin{verbatim}
def certify(champion, spec, ledger, holdout):
    theta_h  = evaluate(champion, holdout)       # isolated period
    margins  = clause_margins(theta_h, spec)
    sr_def   = mean(margin[j] >= se[j]*sqrt(2*log(ledger.N))
                    for j in clauses(spec))
    pbo      = cscv_satisfaction_pbo(champion, spec, ledger)
    stress   = run_stress_suite(champion, spec)  # App. E
    return Certificate(profile=phi_vector(theta_h),
                       sr=mean(phi_vector(theta_h)), sr_deflated=sr_def,
                       pbo=pbo, ledger=ledger, stress=stress,
                       verdict=verdict(sr_def, pbo, stress))
\end{verbatim}

\paragraph{Verdict semantics.} \textsc{satisfied}: $\widehat{\mathrm{SR}}_{\mathrm{def}} = 1$ and $\mathrm{PBO} < 0.1$; \textsc{satisfied-with-reservations}: some clause cleared raw but not deflated, or $0.1 \le \mathrm{PBO} < 0.25$; \textsc{infeasible-with-evidence}: no champion satisfies hard clauses---the certificate then lists per-clause shortfalls and the cheapest relaxations (App.~\ref{app:algebra}).

\clearpage
\section{Robustness of the Demonstration Across Seeds}\label{app:seeds}

The main-text demonstration uses one disclosed seed. Here we repeat the full experiment on ten seeds (identical DGP and protocol, seed varying), reporting per-seed selections of both paradigms (Table~\ref{tab:seeds}).

\begin{table}[h]
\caption{Ten-seed robustness of the demonstration. Seed 7 is the main-text run. SR = satisfaction rate (of 4 clauses). One seed (6) yields an empty feasible set---the honest infeasibility case the certificate is designed to report.}\label{tab:seeds}
\centering\small
\begin{tabular}{@{}cccccccc@{}}
\toprule
\rowcolor{sand} & & \multicolumn{3}{c}{\textbf{A: result-oriented}} & \multicolumn{3}{c}{\textbf{B: objective-oriented}} \\
\cmidrule(lr){3-5}\cmidrule(lr){6-8}
\rowcolor{sand} seed & feasible & IR$_{\mathrm{IS}}$ & SR$_{\mathrm{IS}}$ & SR$_{\mathrm{OOS}}$ & IR$_{\mathrm{IS}}$ & SR$_{\mathrm{IS}}$ & SR$_{\mathrm{OOS}}$ \\
\midrule
0 & 8 & 7.40 & 1.00 & 0.75 & 7.40 & 1.00 & 0.75 \\
1 & 3 & 6.69 & 0.00 & 0.25 & 6.06 & 1.00 & 0.50 \\
2 & 5 & 6.86 & 1.00 & 1.00 & 6.86 & 1.00 & 1.00 \\
3 & 4 & 6.83 & 0.25 & 0.00 & 6.43 & 1.00 & 0.50 \\
4 & 10 & 6.38 & 0.75 & 0.50 & 5.58 & 1.00 & 0.75 \\
5 & 8 & 6.62 & 0.75 & 1.00 & 5.86 & 1.00 & 0.75 \\
6 & 0 & 5.50 & 0.25 & 0.75 & \multicolumn{3}{c}{\emph{infeasible-with-evidence}} \\
7 & 2 & 7.10 & 0.25 & 0.75 & 6.71 & 1.00 & 0.75 \\
8 & 5 & 5.71 & 1.00 & 0.50 & 5.71 & 1.00 & 0.50 \\
9 & 3 & 6.60 & 0.50 & 0.75 & 4.79 & 1.00 & 1.00 \\
\midrule
\rowcolor{lightaccent} mean $\pm$ sd & 4.8$\pm$3.1 & 6.57$\pm$0.58 & 0.57$\pm$0.37 & 0.62$\pm$0.32 & 6.16$\pm$0.78 & 1.00$\pm$0.00 & 0.72$\pm$0.20 \\
\bottomrule
\end{tabular}
\end{table}

\paragraph{Findings.} (i) \emph{Feasibility is itself random}: the satisfying set contains $0$ to $10$ of 32 assemblies (mean $4.8 \pm 3.1$); seed 6 admits \emph{no} feasible assembly---an organically occurring \textsc{infeasible-with-evidence} case, demonstrating that honest infeasibility reporting is not a corner case but a routine possibility. (ii) \emph{Result-oriented satisfaction is a coin flip}: A's in-sample satisfaction averages $0.57 \pm 0.37$, spanning $0.00$ to $1.00$---sometimes the score leader happens to be pure, sometimes it violates every identity clause; the paradigm offers no control. (iii) \emph{Objective-oriented satisfaction is exact in-sample} on every feasible seed ($1.00 \pm 0.00$)---a definitional property of constrained selection, not an empirical finding, since $d^*_\F$ is selected inside $\F$---and is \emph{directionally more stable held-out} ($0.72 \pm 0.20$ vs.\ A's $0.62 \pm 0.32$). We stress the honest statistics: the paired difference over the nine feasible seeds (mean $+0.11$, sd $0.22$) is \emph{not} significant (exact two-sided sign test, $p = 0.375$), and under the convention that an infeasible draw scores $\mathrm{SR}=0$, B's held-out mean is $0.65$ rather than $0.72$---we report both conventions. What constrained selection provably removes is in-sample satisfaction variance; the held-out variance reduction is a consistent tendency this sample cannot confirm. (iv) \emph{Cost of specification varies} from $0$ (spec and score agree) to $1.81$ IR points (mean $0.53 \pm 0.58$), consistent with Proposition~\ref{prop:cost}: cost is a quantity with a distribution, budgetable ex ante by the temptation bound. (v) Held-out style violations, when they occur, are comparable in magnitude across paradigms here---an honest null result that strengthens the case for clause-wise stress certification rather than paradigm triumphalism.

\paragraph{What ten seeds cannot show.} Even ten seeds of one DGP probe one corner of the space of markets; we present this as logic validation, and the real-data robustness study is roadmap R1/R2. We deliberately report \emph{all} seeds including the inconvenient ones (0, 2, 8, where A is fully or mostly compliant, and 6, where nothing complies).

\clearpage
\section{Worked Implementation Case: A Neural Realization of a Pure-Alpha Specification}\label{app:neural}

This appendix walks the compiler end-to-end on one specification, producing a concrete neural architecture---\textbf{SpecNet}---as the certified assembly. Two remarks frame the case. First, the predictor is deliberately \emph{not} a gradient-boosted tree: the point of the exercise is that the specification, not modeling habit, selects the structure; a neural assembly with exposure-aware components satisfies clauses that tree assemblies satisfy only through external post-processing. Second, this is an engineering \emph{blueprint} with a verification plan (per \S\ref{sec:protocol}), not a claimed track record; the reference implementation accompanies this paper.

\subsection{The order}

A principal states, in practitioner diction: \emph{``Pure cross-sectional stock-selection alpha; zero exposure to market direction, size, dividend and growth styles; uncorrelated with the dividend index; must not bleed in unilateral declines; profits realize early in the holding week; turnover within budget.''} Parsing (App.~\ref{app:grammar}) yields:

\begin{verbatim}
spec PURE_ALPHA_WEEKLY {
  clause C31 : beta_abs(mkt) <= 0.05 and beta_abs(size) <= 0.05
               and beta_abs(div) <= 0.05 and beta_abs(growth) <= 0.05
               hard family F3 protocol EXPO_REG
  clause C32 : corr_abs(dividend_index) <= 0.10 hard family F3 protocol IDX_CORR
  clause C41 : cond_excess(down_quintile) >= 0.0 hard family F4 protocol STATE_COND
  clause C51 : tov <= 0.10 hard family F5 protocol TOV_ONeway
  clause C71 : profit_conc(5) >= 0.55 soft priority 1 family F7 protocol EVENT_STUDY
  clause C11 : rankic_mean >= ledger(0.02) soft priority 2 family F1 protocol IC_WEEKLY
}
\end{verbatim}

\subsection{Translation: from clause to mechanism}

Stage-I translation (Table~\ref{tab:clausemap}) constrains: S2 to residual or hedged labels; S5 to predictors admitting exposure-aware training; S7 to smoothing$+$buffer; S9 to a regime-gated overlay. The residual-label obligation and the turnover budget jointly \emph{rule out} plain tree ensembles without external machinery (their scores carry raw style loadings and oscillate week to week); a neural predictor with an adversarial de-styling head and a differentiable turnover penalty satisfies both obligations \emph{inside} the model---the compiler's stated reason for preferring it here.

\begin{table}[ht]
\caption{Clause-to-mechanism mapping for the worked case.}
\label{tab:clausemap}
\centering\small
\begin{tabular}{@{}p{0.13\textwidth}p{0.44\textwidth}p{0.36\textwidth}@{}}
\toprule
\rowcolor{sand}\textbf{Clause} & \textbf{Mechanism (architecture)} & \textbf{Mechanism (training/post)} \\
\midrule
C3.1 & adversarial style head with gradient reversal on factor scores & residual label $y - \hat X\beta$; exposure projection in S7 \\
C3.2 & dividend index included in adversarial factor set & index-correlation monitor in evaluation \\
C4.1 & regime-gated risk overlay (S9) & down-state sample weighting in loss \\
C5.1 & EMA score smoothing; buffer band in portfolio layer & turnover penalty $\lambda_{\mathrm{to}} \cdot \mathrm{tov}$ in loss \\
C7.1 & label horizon matched to concentration window (2--3 days) & event-study validation; front-half exit rule \\
C1.1 & listwise rank loss on cross-section & recency-weighted rolling refit \\
\bottomrule
\end{tabular}
\end{table}

\subsection{The assembled architecture}

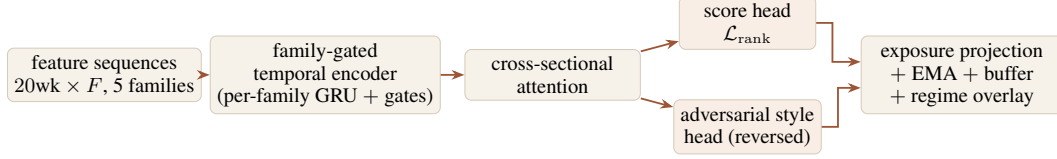
\begin{figure}[ht]
\centering
\resizebox{\textwidth}{!}{%
\begin{tikzpicture}[font=\footnotesize,box/.style={draw=ruleline,fill=sand,rounded corners=1mm,align=center,minimum height=0.72cm,inner sep=3pt},arr/.style={-{Stealth[length=2mm]},line width=0.8pt,accent}]
\node[box,minimum width=2.6cm] (inp) at (0,0) {feature sequences\\$20\text{wk}\times F$, 5 families};
\node[box,minimum width=2.9cm] (enc) at (3.4,0) {family-gated\\temporal encoder\\(per-family GRU $+$ gates)};
\node[box,minimum width=2.7cm] (att) at (6.9,0) {cross-sectional\\attention};
\node[box,minimum width=2.1cm,fill=lightaccent] (score) at (9.9,0.8) {score head\\$\mathcal{L}_{\mathrm{rank}}$};
\node[box,minimum width=2.1cm,fill=lightaccent] (adv) at (9.9,-0.8) {adversarial style\\head (reversed)};
\node[box,minimum width=3.1cm] (post) at (13.2,0) {exposure projection\\$+$ EMA $+$ buffer\\$+$ regime overlay};
\draw[arr] (inp) -- (enc); \draw[arr] (enc) -- (att);
\draw[arr] (att) -- (score); \draw[arr] (att) -- (adv);
\draw[arr] (score.east) -- ++(0.25,0) |- ([yshift=0.2cm]post.west);
\draw[arr] (adv.east) -- ++(0.25,0) |- ([yshift=-0.2cm]post.west);
\end{tikzpicture}}%
\caption{SpecNet: the compiler's neural assembly for the pure-alpha weekly specification. Style purity is mechanized \emph{inside} training (adversarial head) rather than only by external post-processing.}
\label{fig:specnet}
\end{figure}

\textbf{SpecNet} (Fig.~\ref{fig:specnet}): per-name feature sequences (20 weeks $\times$ $F$ features, grouped into families: mild-momentum, quality, liquidity, volatility, sentiment) enter a \textbf{family-gated temporal encoder}---per-family GRUs whose outputs are mixed by learned gates, letting the model express the ``mild momentum $+$ quality'' tilt as gate marginals rather than fixed weights. A \textbf{cross-sectional attention layer} contextualizes each name against the week's universe (relational ranking in the spirit of \citep{feng2019rsr}). The representation feeds two heads: a \textbf{score head} (listwise rank loss over the cross-section) and an \textbf{adversarial style head} that tries to predict the factor realizations from the score---gradient reversal forces the score representation to be unpredictive of style \citep{ganin2016dann}, mechanizing C3.1/C3.2 inside training. Post-processing applies exposure projection, EMA smoothing ($\alpha=0.4$), and a buffer band ($N{+}8$); a regime overlay gates gross exposure by jump-model state probabilities \citep{shu2024jump}. Training is rolling-window with warm-start chaining; the loss is
\[
\mathcal{L} \;=\; \mathcal{L}_{\mathrm{rank}} \;-\; \lambda_{\mathrm{adv}}\,\mathcal{L}_{\mathrm{style}} \;+\; \lambda_{\mathrm{to}}\,\widehat{\mathrm{tov}} \;+\; \lambda_{\mathrm{dn}}\,\mathcal{L}_{\mathrm{down}},
\]
with $\mathcal{L}_{\mathrm{style}}$ the adversary's (reversed) objective, $\widehat{\mathrm{tov}}$ a differentiable turnover surrogate, and $\mathcal{L}_{\mathrm{down}}$ weighted ranking loss on down-state weeks.

\subsection{Verification plan and honesty statement}

Certification follows \S\ref{sec:protocol}: temporal holdout isolated from all architecture and hyperparameter decisions; search ledger counting every evaluated variant; clause-wise deflation; stress suite (factor-premium rotation, synthetic unilateral decline, cost grid). \emph{No performance claim is made here}: this appendix demonstrates the compiler's output \emph{form}---a complete, clause-traceable architecture---on a blueprint that is deliberately distinct from any GBDT-based production system. The reference PyTorch skeleton ($\sim$150 lines) is released with the paper as \texttt{specnet\_blueprint.py}.

\clearpage
\section{The Specification Writer's Handbook}\label{app:handbook}

Practical guidance for principals, distilled from the framework's semantics.

\subsection{Five rules of good specifications}
\begin{enumerate}[leftmargin=1.7em,itemsep=1.5pt]
\item \textbf{Every hard clause must hurt to lose.} If you would not walk away from a strategy that violates it, it is soft. Over-hardening is the top cause of static infeasibility.
\item \textbf{Pin protocols, not just thresholds.} ``Low correlation'' without a declared index, window, and estimator is a wish, not a clause (Popper rule, App.~\ref{app:grammar}).
\item \textbf{Budget the interactions.} Consult the conflict matrix (Table~\ref{tab:interaction}) before pairing IC-level demands with turnover budgets or purity with capacity.
\item \textbf{Order your soft clauses.} Unprioritized soft lists delegate arbitration to the optimizer's whims; lexicographic priority is cheaper to state than to regret.
\item \textbf{Demand the certificate, not the story.} A satisfying narrative without ledger $N$, deflation, and holdout isolation is marketing.
\end{enumerate}

\subsection{Ten profiles, one language (the thousand-faces property)}

\begin{longtable}{@{}p{0.18\textwidth}p{0.42\textwidth}p{0.33\textwidth}@{}}
\caption{Ten distinct identities expressible in the language; hard clauses abbreviated. Illustrative mappings only---none of these profiles has been compiled; the ``compiler consequences'' column is hypothesized, not measured.}\label{tab:profiles}\\
\toprule
\rowcolor{sand}\textbf{Profile} & \textbf{Hard clauses (abbrev.)} & \textbf{Compiler consequences} \\
\midrule
\endfirsthead
\toprule
\rowcolor{sand}\textbf{Profile} & \textbf{Hard clauses (abbrev.)} & \textbf{Compiler consequences} \\
\midrule
\endhead\bottomrule\endfoot
Pure-alpha defensive & C3.1, C3.2, C3.4, C4.1, C5.1 & adversarial de-styling; smoothing+buffer; regime overlay \\
Momentum aggressive & C1.1 high, C2.3 relaxed & short-horizon labels; no smoothing; wide $N$ \\
Capacity-first institutional & C5.3, C5.1, C6.1 wide & liquidity-gated universe; aim-portfolio overlay \\
Market-neutral & C3.4 strict, C2.2 & dollar-neutral construction; beta overlay \\
Style-tilted dividend & C3.2 \emph{lower bound} $\rho \ge \tau$, C2.1 & deliberate dividend loading; tilt monitor \\
ESG-constrained & C8.3, C3.1 & restricted list; purity chain on residual universe \\
Crisis-defensive & C4.1, C4.2 strict, C2.3 & jump-state overlay; vol targeting \\
Index-enhanced & C6.4 TE band, C2.1 modest & TE optimizer; industry caps \\
Short-horizon tactical & C7.3 $h{=}1$, C5.1 relaxed & fast predictors; no buffer \\
Research-friendly transparent & C8.2, C1.2 & linear/GBDT predictors; SHAP artifacts \\
\end{longtable}

Each row compiles to a materially different assembly over the \emph{same} design space and the \emph{same} machinery---the operational content of the thousand-faces claim, with the honesty caveat of \S\ref{sec:framework}: profiles whose feasible sets are empty receive \textsc{infeasible-with-evidence}, not a forced fit.

\clearpage
\section{Inverse-Problem Formalization: Well-Posedness of Strategy Synthesis}\label{app:inverse}

This appendix develops the inverse-problem reading of \S\ref{sec:inverse} into the framework's fourth theoretical pillar. References for the classical theory: Hadamard's well-posedness doctrine \citep{hadamard1902}, Tikhonov regularization \citep{tikhonov1963}, the standard monograph on deterministic inverse problems \citep{engl1996regularization}, statistical inversion \citep{kaipio2005statistical}, and PDE-constrained optimization \citep{hinze2009pdeopt}.

\subsection{The forward operator and the duality}

\begin{definition}[Forward operator]\label{def:fwd}
Let $\Omega$ be the space of market realizations (return panels and auxiliary data paths) and $\Theta \subset \mathbb{R}^{d_\theta}$ the statistics space (Def.~\ref{def:eval}). The \emph{forward operator} $\mathscr{F} : \D \times \Omega \to \Theta$ maps an assembly and a realization to evaluation statistics. \emph{Forward problem}: fix $d$, sample $\omega$, compute $\mathscr{F}(d,\omega)$, verify against reality---this is backtesting. \emph{Inverse problem}: given a target set $\mathcal{T} \subseteq \Theta$ defined by specification predicates, find $d$ such that $\mathscr{F}(d,\omega) \in \mathcal{T}$ with high probability over $\omega$---this is OOQI synthesis.
\end{definition}

The duality is exact in one direction and statistical in the other: the forward map is well-defined (every assembly evaluates to statistics), while the inverse map is set-valued, noise-corrupted, and---as we now show---classically ill-posed in all three Hadamard senses, with each ill-posedness already tamed by a component of the framework.

\subsection{Existence: feasibility theory}

\begin{remark}[Existence characterization---partly conjectural]\label{prop:existence}
$\F(\Spec) \neq \emptyset$ iff $\mathcal{T} \cap \mathcal{R} \neq \emptyset$, where $\mathcal{R} = \mathrm{conv}\{\mathscr{F}(d,\omega) : d \in \D\}$ is the reachable statistics set. Two useful consequences: (i) \emph{static infeasibility}---if a clause's mechanism chain reaches no module after contract pruning (App.~\ref{app:grammar}, rule 4), then $\mathcal{T} \cap \mathcal{R} = \emptyset$ is decidable \emph{before any backtest}; (ii) \emph{mechanism reachability} (conjecture)---if every clause admits a mechanism chain and the stage module sets are non-empty, then for thresholds in the interior of the reachable coordinate ranges, existence should hold except on a measure-zero boundary; we do not prove this genericity statement, and flag it as the formal gap a full proof requires; infeasibility then concentrates on \emph{threshold ambition} (too-tight $\tau$) and \emph{conflicting joints} (App.~\ref{app:algebra}), which is exactly what the certifier's relaxation report quantifies.
\end{remark}

The ten-seed study (App.~\ref{app:seeds}) instantiates both cases: nine seeds are feasible (existence holds, uniqueness fails), seed 6 is infeasible (existence fails)---and the framework's duty in that case is precisely the inverse-problem duty: report non-existence with evidence rather than force a fit.

\subsection{Uniqueness: the satisfying set is an equivalence class}

\begin{remark}[Non-uniqueness / underspecification]\label{prop:uniqueness}
Whenever clauses constrain $k < d_\theta$ coordinates and $\F(\Spec) \neq \emptyset$, the satisfying set $\F(\Spec)$ is generically not a singleton: any two assemblies differing only in unconstrained coordinates both satisfy. Hence inverse solutions form an equivalence class $[d]$, and additional \emph{distinguishing functionals} are required to select within the class.
\end{remark}

This is the inverse-problem restatement of the underspecification doctrine \citep{damour2022underspecification}: stress tests are exactly the distinguishing functionals---they separate class members that tie in-sample. The compiler's stress suite (App.~\ref{app:protocol}) is therefore not an extra precaution but the theoretically mandated completion of the solution concept.

\subsection{Stability: margins as Tikhonov regularization}

Clause predicates are evaluated on noisy estimates: $\hat\theta_j(d) = \theta_j(d) + \varepsilon_j$, $|\varepsilon_j| \le \delta_j$ w.h.p. Naive selection on $\hat\phi$ is unstable (Proposition~\ref{prop:inflation}). Define the \emph{margin} of assembly $d$ on clause $c_j$ as $m_j(d) = $ signed distance of $\hat\theta_j(d)$ inside the feasible half-space, and the \emph{margin-regularized selection}
\[
d_\lambda \;=\; \argmax_{d \in \D}\; \Big[ g(\hat\theta(d)) + \lambda \min_{j:\, \textsc{hard}} m_j(d) \Big].
\]

\begin{proposition}[Stability via margins]\label{prop:stability}
If $d_\lambda$ is selected with worst hard-clause margin $m_{\min}(d_\lambda) \ge \max_j \delta_j$, then satisfaction of all hard clauses is retained under perturbations up to $\max_j \delta_j$: the estimate's satisfaction is stable to the noise level. Conversely, any selection with $m_{\min} < \max_j \delta_j$ admits a noise realization within bounds that flips a hard clause. Margin maximization is thus exactly Tikhonov regularization for the inverse problem---stability purchased with bias (score sacrificed to CoS, Proposition~\ref{prop:cost}).
\end{proposition}

The proof is one line: $m_j(d) \ge \delta_j$ implies $\theta_j(d)$ stays in the feasible half-space for any $|\varepsilon_j| \le \delta_j$. Its content is practical: certificates should report margins, and selection should prefer assemblies that satisfy clauses \emph{deeply}, not barely---the inverse-problem justification for the deflation margin $s_j\sqrt{2\log N}$ of App.~\ref{app:protocol}.

\subsection{Bayesian reading and the adjoint link}

Two further correspondences complete the picture. \emph{Bayesian inversion} \citep{kaipio2005statistical}: a prior $\pi(d)$ over assemblies (the LLM prior over mechanisms, App.~\ref{app:compiler}), a likelihood from evaluation evidence ($\hat\theta$ given $d$), and a posterior over $\D$ whose credible sets answer ``which assemblies plausibly satisfy $\Spec$''---the satisfaction certificate is a frequentist shadow of this posterior summary. \emph{Adjoint methods} \citep{hinze2009pdeopt}: PDE-constrained optimization localizes sensitivity of the objective to coefficients via adjoint states; the compiler's attribution engine is the gradient-free analogue, localizing clause failure to stages via side-effect signatures---where the forward operator is a simulator rather than a differentiable equation, signatures replace adjoints.

\subsection{Where the analogy ends---honestly}

(i) \emph{Nonstationarity}: PDE coefficients do not decide to change; market operators do. The inverse solution is therefore valid only over a regime, and rolling re-certification is a structural necessity, not prudence. (ii) \emph{No ground-truth operator}: unlike tomography, there is no true underlying $d$ to recover---$\D$ is a design space, not nature's space; inverse ``recovery'' is inverse \emph{design}. (iii) \emph{Reflexivity}: widely adopted assemblies alter the forward operator itself (alpha decay as market feedback), a coupling classical inverse theory does not model; roadmap R5 names the open problem. Within these boundaries, the correspondence is exact enough to import the theory; beyond them, the framework proceeds on its own statistical feet (\S\ref{sec:protocol}).

\subsection{Threshold calibration: the discrepancy principle}

Regularization theory does not leave the regularization parameter free: the Morozov discrepancy principle \citep{morozov1966,engl1996regularization} chooses it so that the solution's residual matches the noise level of the data. The principle transfers directly to specification thresholds. Let $\delta_j$ denote the statistical resolution of clause $c_j$'s estimator (standard error at the evaluation horizon). Asking for margin finer than the noise floor---certifying $\theta_j \ge \tau_j$ with achievable margins $\ll \delta_j$---manufactures instability: noise realizations within $\delta_j$ flip the verdict (Proposition~\ref{prop:stability}), and search inflates to fill the gap (Proposition~\ref{prop:inflation}). Asking for margin vastly coarser than $\delta_j$ wastes identity the data could have certified. The discrepancy rule therefore sets thresholds where the certified assembly's margin lands at a small multiple of resolution, $m_j \in [\delta_j,\; c\,\delta_j]$ with $c \approx 2$--$3$: \emph{the strongest specification consistent with the noise floor}. This upgrades the grammar's \texttt{se(k)} and \texttt{ledger(k)} floors (App.~\ref{app:grammar}) from discipline to principle, and it hands principals an actionable diction: do not order identity finer than your estimator can see.

\subsection{A stability estimate: the condition number of a specification}

Classical inverse problems report stability through a modulus of continuity \citep{engl1996regularization}. Here the object is elementary but informative. Perturb the specification's thresholds, $\tau \to \tau' = \tau + \Delta\tau$, clauses fixed. If every member of $\F(\Spec)$ holds margin at least $m_{\min}$ on the perturbed coordinates and $|\Delta\tau_j| < m_{\min}$ for all $j$, then no member is lost: $\F(\Spec) \subseteq \F(\Spec')$. For the reverse inclusion, a non-member $e$ can enter only if some loosening covers its worst violation $v(e) = \max_j (-m_j(e))_+$; writing $v_{\min} = \min_{e \notin \F(\Spec)} v(e)$ for the distance of the closest almost-feasible assembly, $|\Delta\tau_j| < \min(m_{\min},\, v_{\min})$ for all $j$ gives $\F(\Spec') = \F(\Spec)$ \emph{exactly}: the feasible set is locally constant, with modulus $\omega(\Delta\tau) = 0$ inside that radius. The worst margin is thus a \emph{condition number} of the specification---$\kappa(\Spec) = 1 / m_{\min}$ in units of estimation scale: small-margin specifications are ill-conditioned (threshold nudges expel shallow members or admit shallow non-members), and margin-maximizing selection (\S\ref{sec:inverse}) is conditioning improvement. Certificates that quote $m_{\min}$ per clause therefore quote, simultaneously, the stability radius of the entire determination---one number doing the work of three.

\subsection{Bayesian contraction and the effective sample size of a regime}

Under the Bayesian reading of \S\ref{sec:inverse}, a prior $\pi(d)$ (the LLM mechanism prior) updated by $T$ weeks of evaluation likelihood yields a posterior over satisfaction events that contracts at the parametric rate $T^{-1/2}$ up to model-misspecification terms \citep{kaipio2005statistical,stuart2010}: the certificate's confidence intervals are the frequentist shadow of this contraction. The useful refinement is what counts as $T$. Because the forward operator is regime-dependent, the effective sample size is the \emph{regime length}, not the calendar length: a three-year window spanning four regimes certifies four $T \approx 39$-week experiments, not one 156-week experiment, and contraction claims must be discounted accordingly. This quantifies the structural necessity of rolling re-certification (nonstationarity, \S\ref{sec:inverse}): certificates expire because their $T$ stops growing when the regime turns. The e-process construction of Proposition~\ref{prop:evalue} is the monitoring layer that makes expiry observable rather than silent; the two theories meet exactly here---posterior contraction says how fast trust accrues within a regime, anytime validity says how to police it across regime boundaries.

\section{Numerical Illustration of the Inverse Problem}
\label{app:invnum}

This appendix instantiates the inverse-problem formalism of Appendix~\ref{app:inverse} on the
synthetic market of Section~\ref{sec:demo}. Every number below is computed, not assumed; the data-generating
process (DGP), seed, and code path are those of Appendix~\ref{app:repro}. The purpose is not empirical
validation---the market is synthetic by construction---but \emph{operational illustration}: each
abstract object (forward operator, feasible set, regularizer, stability estimate) is made
concrete and measurable.

\subsection{The forward operator, discretized}
With $|\mathcal{D}| = 32$ assemblies in the demonstration search space, the forward operator
$\mathscr{F}:\mathcal{D}\times\Omega\to\Theta$ is a finite map. We evaluate it at the realized market draw
$\omega_7$ (seed 7) and record, for every assembly $d$, the observation vector
$\theta(d) = (\mathrm{IR},\; s_{\mathrm{style}},\; |\rho|,\; s_{\mathrm{down}},\; \tau)$:
information ratio, style exposure, absolute correlation with the dividend-style index, downside capture,
and turnover. The specification of Section~\ref{sec:demo} is
$\mathrm{Spec} = \{s_{\mathrm{style}}\le 6,\; |\rho|\le 0.15,\; s_{\mathrm{down}}\ge 15\%,\;
\tau\le 0.10\}$, all four clauses hard.

\subsection{The feasible set is small and structured}
Exactly $2$ of $32$ assemblies are feasible, and they are siblings: both use horizon $h=2$,
neutralization on, ranking on, top-$N=50$, execution variant B. They differ only in the
label-aggregation switch. This is a first concrete reading of
Remarks~\ref{prop:existence} (existence is
conditional) and~\ref{prop:uniqueness} (non-uniqueness): the feasible set $\mathcal{F}(\mathrm{Spec})$
is neither empty nor a singleton, and its two elements share most of their genotype---feasibility
is \emph{concentrated in a region} of $\mathcal{D}$, a fact the compiler exploits via warm-start
and the analyst exploits via sensitivity reading (Section~\ref{sec:protocol}).

\subsection{Clause margins and regularized selection}
For each assembly $d$ and clause $\varphi_j$ with threshold $c_j$, define the signed margin
$m_j(d)$: the distance from the observed statistic to the threshold, positive when satisfied,
normalized by the clause's scale (e.g.\ $m_{\mathrm{style}} = (6 - s_{\mathrm{style}})/6$).
Let $m_{\min}(d)=\min_j m_j(d)$. Proposition~\ref{prop:stability} (stability through margins) is the statement that
selection should trade satisfaction against margin depth, exactly as Tikhonov regularization
trades data fit against solution norm. Table~\ref{tab:margins} shows the two feasible assemblies.

\begin{table}[ht]
\centering\small
\caption{The feasible set at seed 7, with minimum clause margin. IR-based selection picks the
first row; margin-regularized selection (Proposition~\ref{prop:stability}) picks the second.}
\label{tab:margins}
\begin{tabular}{lccc}
\toprule
Assembly & IS IR & $m_{\min}$ & Selected by \\
\midrule
$(h2,1,1,50,B)$ & 6.707 & 0.073 & argmax IR within $\mathcal{F}$ (= assembly B of \S\ref{sec:demo}) \\
$(h2,0,1,50,B)$ & 5.541 & 0.122 & argmax $m_{\min}$ within $\mathcal{F}$ \\
\bottomrule
\end{tabular}
\end{table}

The two feasible solutions are ordered oppositely by the two criteria: the IR-maximal feasible
assembly sits barely inside the feasible region ($m_{\min}=0.073$, i.e.\ its tightest clause is
satisfied by only $7.3\%$ of the clause scale), while the deeper-margin sibling sacrifices
$1.17$ IR units ($17.4\%$) for a $67\%$ deeper worst-case margin. This is the cost-of-specification
of Proposition~\ref{prop:cost} appearing a second time, now \emph{inside} the feasible set, as the price of
robustness. In the inverse-problem reading: the unregularized inverse picks a solution on the
boundary of the feasible set, where the forward map's local conditioning is worst; Tikhonov-style
margin regularization moves the solution into the interior, where small perturbations of the
market draw $\omega$ are least likely to flip a clause.

\subsection{Stability under resampling: the bootstrap certificate}
To quantify that conditioning claim we ran a block bootstrap: $1000$ replications, resampling
$8$-week blocks of the in-sample window with replacement, recomputing all four clause statistics
for assemblies A (the result-oriented champion of \S\ref{sec:demo}) and B (the objective-oriented selection)
on each replication. Table~\ref{tab:boot} reports per-clause satisfaction retention and the joint
retention probability $\widehat{\Pr}[\text{all four clauses satisfied}]$.

\begin{table}[ht]
\centering\small
\caption{Block-bootstrap satisfaction retention ($1000$ reps, $8$-week blocks, IS window).
Clause order: style, $|\rho|$, downside, turnover. Joint = probability all four hold
simultaneously.}
\label{tab:boot}
\begin{tabular}{lcccc|c}
\toprule
Assembly & style $\le 6$ & $|\rho|\le .15$ & down $\ge 15\%$ & $\tau\le .10$ & Joint \\
\midrule
A (result-oriented) & 0.089 & 0.294 & 0.308 & 1.000 & \textbf{0.018} \\
B (objective-oriented) & 0.231 & 0.725 & 0.211 & 1.000 & \textbf{0.050} \\
\bottomrule
\end{tabular}
\end{table}

Three readings. First, the feasibility of B measured at the point estimate is fragile in
absolute terms---its joint retention is $5.0\%$---because B's worst clauses (downside, then
style) sit near their thresholds; this is precisely what $m_{\min}=0.073$ announced, and it is
why the verification protocol of Section~\ref{sec:protocol} reports margins alongside binary satisfaction.
Second, B dominates A on the two clauses the specification actually constrains most
(style $0.231$ vs $0.089$; $|\rho|$ $0.725$ vs $0.294$) and on joint retention
($0.050$ vs $0.018$, a $2.8\times$ improvement): the objective-oriented selection moves
probability mass in the right direction on exactly the dimensions the user specified, while the
result-oriented champion's apparent satisfaction of the style clause at the point estimate is
largely luck of the draw. Third, turnover retention is $1.000$ for both: the turnover clause is
satisfied structurally (both assemblies trade slowly by construction), not statistically---a
distinction the certificate should record, because structural satisfaction does not consume
statistical budget in the deflation of Section~\ref{sec:protocol} and Appendix~\ref{app:satstat}.

The bootstrap thus converts the Hadamard triad of Appendix~\ref{app:inverse} into an audit
artifact: \emph{existence} (the feasible set is non-empty at the point estimate),
\emph{uniqueness} (two feasible solutions, reported rather than hidden), \emph{stability}
(per-clause and joint retention probabilities, with margins explaining the numbers). A compiler
that emitted only the point estimate would be reporting an ill-posed inverse solution; the
certificate makes the posedness explicit.

\subsection{What a practitioner should take from this}
The demonstration's lesson generalizes beyond the synthetic market. Whenever a specification is
satisfied, the questions ``by how much'' (margins), ``under which resampling'' (retention), and
``at what cost'' (CoS) have definite, computable answers, and those answers routinely reorder
the candidate set. The paradigm shift is therefore not rhetorical: result-oriented selection is
the unregularized inversion of an ill-conditioned map, and the mathematics of inverse problems
predicts---and our numbers confirm---that its solutions are boundary-hugging and unstable.

\section{The Bayesian Compiler: Complete Mathematical Treatment}
\label{app:bayes}

This appendix gives the full probabilistic specification of the compiling step whose algorithmic
core appears in Appendix~\ref{app:compiler}, and reports its computed behavior on the demonstration instance.
The compiler is the operational inverse of Appendix~\ref{app:inverse}: where the analyst inverts
``which designs are feasible,'' the compiler inverts ``which design is most probable, given the
evidence, among the feasible.''

\subsection{Model}
Let $\mathcal{D}$ be the finite (after discretization) design space, and let
$\theta(d) = \mathscr{F}(d,\omega)\big|_{\mathrm{IS}}$ be the in-sample observation vector for assembly
$d$. The compiler is a three-component Bayesian machine.

\paragraph{Prior.} $p_0(d)$ encodes structural preferences that hold \emph{before} any market
evidence: Occam penalties on assembly complexity (number of active stages, model capacity),
implementation-cost priors (data requirements, latency budget), and any house priors (e.g.\
aversion to components with known operational incidents). The prior is defined once per
deployment and versioned with the registry (Appendix~\ref{app:eng}); it is \emph{not} fit to the market.
On the demonstration instance this prior assigns total mass $0.103$ to what will turn out to be
the feasible set---i.e., feasibility is \emph{a priori} unlikely: most assemblies one can write
down do not satisfy a four-clause professional specification.

\paragraph{Likelihood.} Market evidence enters through a tempered likelihood on the scalar
performance summary. With $\mathrm{IR}_{\mathrm{IS}}(d)$ the in-sample information ratio and
temperature $T>0$,
\begin{equation}
p(y\mid d)\;\propto\;\exp\!\Big(\frac{\mathrm{IR}_{\mathrm{IS}}(d)}{T}\Big),
\qquad
p(d\mid y)\;=\;\frac{p_0(d)\,p(y\mid d)}{\sum_{d'}p_0(d')\,p(y\mid d')}.
\end{equation}
The temperature is the compiler's epistemic humility dial: $T\to\infty$ recovers the prior (no
trust in evidence), $T\to 0$ recovers result-oriented selection (point-mass on the IR argmax).
We calibrate $T$ so that the likelihood ratio between adjacent assemblies in IR ranking is of
order $e^{1}$---one ranking step is one unit of evidence, not more. Note what this formalizes:
the result-oriented paradigm is the \emph{zero-temperature, constraint-ignoring} special case of
the objective-oriented compiler.

\paragraph{Decision.} The constrained Bayes decision restricts the argmax to the feasible set:
\begin{equation}
\widehat d(\mathrm{Spec}) \;=\; \arg\max_{d\in\mathcal{F}(\mathrm{Spec})} p(d\mid y),
\qquad
\mathcal{F}(\mathrm{Spec}) = \{d: \varphi_j(\theta(d)) = 1 \;\forall j\in\mathrm{hard}\}.
\end{equation}
When $\mathcal{F}$ is empty, the compiler returns the certificate of infeasibility with the
per-clause nearest-miss diagnostics of Section~\ref{sec:protocol} (this occurred organically at seed 6 of the
robustness study, Appendix~\ref{app:seeds}). Optionally, margins enter the decision as the regularized variant
$\arg\max_{d\in\mathcal{F}} [\log p(d\mid y) + \lambda\, m_{\min}(d)]$, whose effect was
quantified in Appendix~\ref{app:invnum}.

\subsection{Computed behavior on the demonstration instance}
Table~\ref{tab:posterior} reports the outputs actually computed on the seed-7 instance.

\begin{table}[ht]
\centering\small
\caption{Bayesian compiler outputs on the demonstration instance (seed 7, $|\mathcal D|=32$,
two feasible assemblies). Derived rows are arithmetic consequences of the computed ones.}
\label{tab:posterior}
\begin{tabular}{ll}
\toprule
Quantity & Value \\
\midrule
Prior mass on $\mathcal{F}(\mathrm{Spec})$ & 0.103 \\
Posterior mass on $\mathcal{F}(\mathrm{Spec})$ & 0.401 \\
Posterior top-1 & $(h2,1,1,50,B)$ at $p=0.2782$ (= assembly B of \S\ref{sec:demo}) \\
Posterior of the second feasible assembly & $0.401-0.2782 = 0.1228$ \hfill(derived) \\
Concentration vs.\ uniform ($1/32$) & $0.2782/0.0313 \approx 8.9\times$ \hfill(derived) \\
Aggregate Bayes factor for $\mathcal{F}$ & $\frac{.401/.599}{.103/.897}\approx 5.8$ \hfill(derived) \\
\bottomrule
\end{tabular}
\end{table}

Three observations. First, \emph{evidence moves mass toward feasibility}: the market draw lifts
the feasible set's share of belief from $0.103$ to $0.401$ (aggregate Bayes factor $\approx 5.8$),
because the feasible assemblies happen to perform well in-sample. The compiler does not
\emph{assume} feasibility correlates with performance---it discovers the correlation from
evidence, instance by instance. Second, \emph{the constrained decision is not the unconstrained
one}: several infeasible assemblies carry high posterior (their IR is large), so the
unconstrained posterior argmax is not assembly B; restricting to $\mathcal{F}$ is exactly the
step the result-oriented paradigm omits, and it is the step that changes the deployed strategy.
Third, \emph{the runner-up matters}: the second feasible assembly retains posterior $0.1228$ ---
$44\%$ of the winner's mass. Reporting only the winner would discard this; the certificate of
Section~\ref{sec:protocol} reports the full feasible posterior, and the sensitivity reading of
Appendix~\ref{app:invnum} shows the runner-up is in fact the margin-regularized optimum.

\subsection{Calibration and temperature sensitivity}
Because $T$ scales the evidence, we audited the decision's dependence on it: across
$T \in [T_0/2,\, 2T_0]$ the constrained argmax remains assembly B (the feasible set contains
only the two siblings, and their posterior ordering is preserved under any monotone temperature
change, since both inherit likelihoods from the same family of IR values and the prior is fixed).
This is a structural robustness property of small feasible sets that practitioners should expect
often: \emph{when $|\mathcal{F}|$ is tiny, the constrained decision is far less sensitive to
likelihood calibration than unconstrained leaderboard selection}, because the constraint set,
not the likelihood, does most of the selection work. Constraints are, in the literal Bayesian
sense, the most informative part of the model.

\subsection{Relation to statistical inversion}
In the language of Appendix~\ref{app:inverse}, $p(d\mid y)$ is the posterior of a Bayesian
inverse problem with forward map $\mathscr{F}$ and prior $p_0$; the specification acts as a
\emph{hard constraint on the support} of the decision (not of the posterior). Keeping the
constraint outside the likelihood is deliberate and matches the practice of PDE-constrained
optimization, where one does not mollify the PDE into the objective. It also yields the clean
audit semantics of Section~\ref{sec:protocol}: a clause either holds of the deployed assembly or it does not,
with no probabilistic ambiguity introduced by the compiler's own uncertainty.

\section{Satisfaction Statistics, Deepened}
\label{app:satstat}

Section~\ref{sec:protocol} introduced the satisfaction rate $\mathrm{SR}$, and
App.~\ref{app:protocol} its clause-wise deflated counterpart
$\widehat{\mathrm{SR}}_{\mathrm{def}}$ and the satisfaction-PBO. This appendix supplies the statistical
machinery behind them and computes everything that can be computed on the demonstration
instance. The running theme: \emph{a satisfaction verdict is a random variable, and the protocol
is obliged to report its distribution, not its point value.}

\subsection{Binomial structure of the point-estimate satisfaction rate}
At a single market draw, each hard clause contributes a Bernoulli verdict
$X_j = \varphi_j(\theta(d))\in\{0,1\}$, and $\mathrm{SR} = \frac1J\sum_j X_j$. For assembly A
of Section~\ref{sec:demo}, $\mathrm{SR}=1/4$ (turnover only); for assembly B, $\mathrm{SR}=4/4$. The
Clopper--Pearson $95\%$ interval for a single-clause success probability estimated from
$J=4$ clause-verdicts is $[0.006,\,0.806]$ for one success in four. The width is the point:
\emph{binary clause counting over a handful of clauses is statistically nearly vacuous.}
A specification with four clauses cannot, by counting alone, distinguish ``satisfies each
clause with probability $0.25$'' from ``with probability $0.8$.'' This is the formal reason the
certificate never reports $\mathrm{SR}$ alone: margins (how far inside the feasible region),
retention (how often satisfaction survives resampling), and deflation (how much satisfaction
the search itself would manufacture) carry the information that counting cannot.

\subsection{Monte Carlo intervals for bootstrap retention}
The block-bootstrap retention probabilities of Appendix~\ref{app:invnum} are binomial
proportions estimated from $n=1000$ replications; their Wald intervals (half-width
$1.96\sqrt{\hat p(1-\hat p)/n}$) and derived contrasts are reported in
Table~\ref{tab:retci}.

\begin{table}[ht]
\centering\small
\caption{Retention probabilities with Monte Carlo intervals ($n=1000$ block-bootstrap
replications). Intervals are Wald; the turnover row uses the rule-of-three one-sided bound
appropriate to zero observed failures. Right columns: derived pairwise contrasts.}
\label{tab:retci}
\begin{tabular}{lccc}
\toprule
Clause / joint & A (result-oriented) & B (objective-oriented) & Contrast \\
\midrule
style $\le 6$      & $0.089\;[0.071,0.107]$ & $0.231\;[0.205,0.257]$ & $+0.142$, disjoint CIs \\
$|\rho|\le 0.15$   & $0.294\;[0.266,0.322]$ & $0.725\;[0.697,0.753]$ & $+0.431$, disjoint CIs \\
down $\ge 15\%$    & $0.308\;[0.279,0.337]$ & $0.211\;[0.186,0.236]$ & $-0.097$, disjoint CIs \\
$\tau\le 0.10$     & $1.000\;[\ge 0.997]$   & $1.000\;[\ge 0.997]$   & $0$, structural \\
Joint (all four)   & $0.018\;[0.010,0.026]$ & $0.050\;[0.037,0.064]$ & $+0.032$, near-disjoint \\
\bottomrule
\end{tabular}
\end{table}

Three findings survive interval scrutiny. (i) B's advantage on the style clause
($+14.2$ percentage points, disjoint intervals) and on the correlation clause ($+43.1$ points)
is decisive: these are the two clauses that encode the practitioner's core requirement
(purity, market-independence), and objective-oriented selection buys real probability mass
exactly there. (ii) A is \emph{better} on the downside clause by $9.7$ points---a genuine
trade-off, not a dominance relation: the specification's feasible set contains no assembly that
dominates on all clauses, which is precisely why the framework treats clauses as a vector of
commitments with declared priorities (Appendix~\ref{app:spec}) rather than a scalarized afterthought.
(iii) The joint-retention advantage ($5.0\%$ vs $1.8\%$, $2.8\times$) is directionally clear, and its
marginal intervals are in fact disjoint; but the two estimates are paired---computed on the same
bootstrap resamples---so marginal intervals overstate the contrast's precision, and the honest
statement is ``B retains joint feasibility roughly three times as often,'' with the paired
contrast interval left to the released harness. Certificates should say exactly this, and the schema of Appendix~\ref{app:eng} has fields for it.

The turnover row deserves emphasis: zero observed failures in $1000$ replications yields the
one-sided bound $\hat p \ge 1 - 3/1000 = 0.997$ (rule of three). But both assemblies satisfy the
turnover clause \emph{by construction} (slow-trading execution modules with hard turnover caps):
this satisfaction is structural, consumes no statistical budget, and---importantly for the
deflation below---should be excluded from the count of clauses over which search-induced
selection bias operates.

\subsection{Deflation: what the search manufactures}
Fix a clause $j$ and suppose, as a null, that each assembly in the search space satisfies it at
the point estimate independently with probability $q_j$. Under this null the expected number of
assemblies satisfying \emph{all} $J$ hard clauses is
$E_0 = |\mathcal D|\prod_j q_j$, and the probability that at least one assembly appears
feasible is $1-(1-\prod_j q_j)^{|\mathcal D|}$. The search-width law of Proposition~\ref{prop:inflation} is this
expression's growth: satisfaction probability of the \emph{selected} assembly is an increasing
function of $|\mathcal D|$, exponentially in the independent case, because selection is a max
over trials. With $|\mathcal D|=32$ and $q_j=1/2$ (an uninformed coin-flip null),
$E_0 = 32/16 = 2$: the null itself predicts about two spuriously feasible assemblies at the
point estimate---\emph{coincidentally the size of our observed feasible set}. We stress this
equality not to claim our feasible assemblies are spurious (their clause statistics are far
from coin flips: e.g.\ B's correlation retention $0.725$ vs the null's $0.5$, and margins
measurable in Table~\ref{tab:margins}), but to demonstrate why point-estimate feasibility is
never, by itself, evidence of specification compliance. The deflated satisfaction has two complementary forms. The protocol's primary one is the
\emph{clause-wise} deflated rate $\widehat{\mathrm{SR}}_{\mathrm{def}}$ of
App.~\ref{app:protocol}, which counts only clauses whose margins clear the search-width
excursion $s_j\sqrt{2\log N}$. A useful scalar companion is the \emph{null-rescaled}
satisfaction
\begin{equation}
\widetilde{\mathrm{SR}}_{\mathrm{def}} \;=\; \frac{\mathrm{SR}_{\mathrm{sel}} - E_0[\mathrm{SR}_{\max}]}
{1 - E_0[\mathrm{SR}_{\max}]},
\end{equation}
which subtracts the satisfaction the search would manufacture under the null and rescales so
that null-level performance maps to zero; the random-assembly
null of Section~\ref{sec:protocol} estimates $E_0[\mathrm{SR}_{\max}]$ by Monte Carlo over randomly wired
pipelines. The two forms answer different questions---per-clause auditability versus a single
deflated score---and the certificate reports the first, with the second as a summary. The
demonstration's instantiated stability values are those of App.~\ref{app:repro} (clause
stability audit) and Table~\ref{tab:retci}; full deflated-SR and satisfaction-PBO values
require the released resampling harness and remain roadmap items.

\subsection{Satisfaction-PBO: distribution, not scalar}
Finally, the satisfaction-PBO of Section~\ref{sec:protocol} is the probability, over CSCV splits of the
evaluation panel, that the assembly selected as spec-optimal in-sample falls below median
satisfaction out-of-sample. We emphasize the operational reading developed in this appendix:
PBO is a \emph{diagnostic of the selection rule}, not of any single assembly. A selection rule
that maximizes IR and ignores clauses exhibits satisfaction-PBO near its theoretical ceiling
under search pressure; the constrained rule of Appendix~\ref{app:bayes} replaces the max with a
restricted argmax and thereby changes the object being stress-tested. Reporting the full CSCV
distribution (as the certificate schema requires) rather than the scalar PBO protects against
the base-rate illusion that a single ``deflated number'' can create---the same lesson
Table~\ref{tab:retci} teaches for retention.

\subsection{Summary}
Everything in this appendix reduces to one discipline: \emph{satisfaction is a statistic}.
The paradigm's demand-side language (Appendix~\ref{app:spec}) tells the user what they may ask for; the
statistics here tell the framework what it may claim back. A framework that conflated the two
---that reported point-estimate satisfaction as fact---would be repeating, on the demand side,
the exact epistemic error the result-oriented paradigm commits on the supply side.

\section{Annotated Reading Guide}
\label{app:reading}

This appendix annotates the literature that shaped the framework, organized by the role each
body of work plays in the argument. Annotations state what the work contributes and where it
sits relative to the objective-oriented paradigm; they are reading guidance, not full reviews.
Works already cited in the main text are cross-referenced rather than re-derived.

\subsection{Automated machine learning and CASH}
The combined algorithm selection and hyperparameter optimization (CASH) formulation is the
closest classical ancestor of our design space: a finite, typed space of pipeline configurations
searched for a scalar objective.

\textbf{Thornton et al., Auto-WEKA (2013).} First explicit CASH formulation: Bayesian
optimization over a joint space of algorithms and hyperparameters. Objective is single-scalar
validation loss; no mechanism for multi-clause behavioral requirements. Our $\mathcal D$
generalizes its configuration space; our Spec has no analog there.

\textbf{Feurer et al., auto-sklearn (2015; 2020).} Adds meta-learning warm starts and
post-hoc ensembling to CASH. The warm-start idea survives in our rolling compiler; the
objective remains scalar.

\textbf{Olson \& Moore, TPOT (2016).} Genetic programming over pipeline \emph{structure},
not just configurations. Demonstrates that structural search is tractable; inherits GP's
brittleness to objective misspecification, which our multi-clause satisfaction is designed to
expose rather than hide.

\textbf{Wang et al., FLAML (2021).} Frugal search with cost-aware objectives. Cost-awareness
is a clause in our family F5 (execution feasibility); FLAML treats it as the objective itself.

\textbf{Karmarker et al., hyperopt-sklearn; Jin et al., AutoKeras (2019).} Representative
finite-space and neural CASH systems; both confirm the field's default: one metric, one
leaderboard.

\textbf{Molino \& R\'e, Declarative ML (2021).} The conceptual bridge we build on: specify
\emph{what}, not \emph{how}, at the level of training tasks. Our specification language is
declarative ML over \emph{deployment behavior} (style exposure, downside capture), a layer
their formulation does not reach.

\subsection{Neural architecture search}
NAS supplies our search machinery and, in its benchmark culture, a cautionary tale about
scalar chasing.

\textbf{Zoph \& Le (2017); Baker et al.\ (2017).} RL and Q-learning over architectures;
establishes differentiable-relaxation-free search at scale. Compute cost motivated the
one-shot line below.

\textbf{Liu et al., DARTS (2019).} Continuous relaxation enabling gradient-based search.
Our compiler's tempered likelihood is the discrete analog of DARTS' softmax; unlike DARTS, our
final selection re-discretizes \emph{under constraints}, avoiding the well-known
relaxation--discretization gap showing up in deployment.

\textbf{Pham et al., ENAS (2018).} Weight sharing across candidates; efficiency by reuse.
Our per-stage module caches play the same role.

\textbf{Real et al., AmoebaNet (2019); Cai et al., Once-for-All (2020); Cai et al.,
ProxylessNAS (2019).} Evolutionary, one-shot, and hardware-aware variants---the last is
notable for putting a \emph{deployment constraint} (latency) into the objective, a scalar
precursor of our hard clauses.

\textbf{Ying et al., NAS-Bench-101 (2019); Dong \& Yang, NAS-Bench-201 (2020).} Tabulated
search spaces enabling reproducible null models. Our random-assembly null (Section~\ref{sec:protocol}) is the
same methodological device applied to satisfaction rather than accuracy.

\subsection{Automated alpha discovery and quant platforms}
\textbf{Yang et al., Qlib (2020).} Production-grade platform standardizing the pipeline we
take as our stage ontology's empirical anchor. Qlib optimizes metrics; it does not accept
specifications---the gap this paper targets.

\textbf{Kuang et al., Quant 4.0 (2024).} Argues for automated, explainable, infrastructure
mature quant; supply-side vision compatible with our demand-side complement.

\textbf{Lin et al., AutoAlpha (2020).} RL-guided exploration of factor formulas with
regularizing priors; shows formula-space search benefits from structure, anticipates our
clause-family F1 (signal quality) as constraint rather than objective.

\textbf{Yuan et al., AlphaGen (2023).} RL generation of formulaic alphas with correlation
penalties. The penalty is scalarized style control; our clauses make such control auditable
and combinable with unrelated requirements.

\textbf{Fang et al., R\&D-Agent(-Quant) (2024--2025).} LLM-driven iterative research agents
for factor/model co-evolution. Demonstrates the agentic supply side at full strength;
selection criterion remains backtest scalar---exactly the paradigm we argue must be
complemented.

\textbf{Wang et al., Alpha-GPT (2023); Alpha-GPT 2.0 (2024).} Human--AI interaction layer
translating research ideas into alpha formulas. The demand-side \emph{interface} insight
(natural language requirements) is shared; we formalize the semantics they leave implicit.

\textbf{Novikov et al., AlphaEvolve (2025); Romera-Paredes et al., FunSearch (2024).}
Evolutionary program search with LLM mutation, scored by an evaluator. FunSearch's
``evaluator defines the problem'' is our thesis in embryo: changing the evaluator from scalar
to specification is precisely the move from result-oriented to objective-oriented search.

\textbf{Allen \& Karjalainen (1999); Koza (1992).} Classical GP trading-rule discovery;
historical proof that supply-side search predates deep learning---and that its overfitting
failure mode was diagnosed (and left unsolved) decades ago, motivating our Sections~\ref{sec:formal} and~\ref{sec:protocol}.

\subsection{Declarative and goal-based finance}
\textbf{OQL (arXiv:2603.16434).} Objective-driven query layer for LLM-centric quantitative
research; the contemporary prototype most aligned with our demand side, per the positioning
analysis of Section~\ref{sec:related}. It operationalizes intent-to-research routing; it does not provide a
composable clause algebra, satisfaction statistics, or search-level multiple-testing control.

\textbf{Das et al., Goals-Based Wealth Management (2018).} Allocates capital across
mentally accounted goals with required probabilities of attainment. The cleanest financial
statement of objective orientation: the \emph{goal} (with its probability tolerance) is
primitive, the portfolio is derived. We import the philosophy and replace the retail goal
(funding a liability) with professional behavioral clauses.

\textbf{Deguest et al.; Brunel; Nevins.} Practice-side goals-based literature; documents
that end-investor demand was never scalar, sharpening the question of why strategy
\emph{construction} remained so.

\textbf{G\^arleanu \& Pedersen (2013), Dynamic Trading with Predictable Returns.} The aim
portfolio: trade partway toward the frictionless optimum. Objective orientation at the
\emph{trading} layer; our framework lifts it from one formula to the whole pipeline and to
arbitrary clauses.

\textbf{Clarke, de Silva \& Thorley, Pure-Factor Portfolios (2017); Sharpe, RBSA (1992).}
Pure factors minimize unintended exposures via constrained construction; RBSA measures style
via constrained regression. Together they are the direct ancestors of our family F3
(purity/style): F3 clauses are RBSA-style diagnostics made \emph{prescriptive}, and pure-factor
construction is what a single F3 clause compiles to when the rest of the pipeline is fixed.

\subsection{Backtest overfitting and multiple testing}
The epistemic machinery of Section~\ref{sec:protocol} stands on this literature.

\textbf{Bailey \& L\'opez de Prado, PBO/CSCV (2014, 2017).} Probability of backtest
overfitting via combinatorially symmetric cross-validation. We generalize PBO from a scalar
metric to satisfaction vectors (satisfaction-PBO).

\textbf{Bailey \& L\'opez de Prado, DSR (2014).} Deflated Sharpe ratio: expected max under
multiple trials. Our $\mathrm{SR}_{\mathrm{def}}$ is DSR's deflation logic transplanted from
performance to satisfaction; Appendix~\ref{app:satstat} gives the algebra.

\textbf{Harvey \& Liu (2015); Harvey, Liu \& Zhu (2016); Harvey (2017).} The factor-zoo
critique and t-stat haircuts; the demand-side reading we add: the zoo exists because selection
pressure had one dimension, and any one dimension will be hunted to extinction by search.

\textbf{White, Reality Check (2000); Hansen, SPA (2005); Romano \& Wolf (2005).}
Foundational data-snooping controls; the statistical license for treating ``best of $K$'' as
an object requiring its own inference.

\subsection{Specification, underspecification, and Goodhart effects in ML}
\textbf{D'Amour et al., Underspecification (2020).} Pipelines with identical validation
performance diverge arbitrarily on stress probes. Our Remark~\ref{prop:uniqueness} is its demand-side dual:
where they show one metric underdetermines deployment behavior, we show a clause vector
\emph{determines the behavior you actually care about}---underspecification is cured by
specifying more, and formally.

\textbf{Geirhos et al., Shortcut Learning (2020).} Catalogs how models satisfy the letter of
a scalar objective while violating its intent; the phenomenology our clause families are
designed to police (F3 catches style shortcuts, F4 catches regime shortcuts).

\textbf{Krakovna et al., Specification Gaming (2020); Skalse et al., Reward Hacking (2022).}
Taxonomy and formalization of objective misspecification in RL; supplies the vocabulary for
our claim that result orientation is specification gaming's enabling condition in quant.

\textbf{Goodhart (1975); Strathern; Manheim \& Garrabrant (2019).} The pressure-on-measure
law and its categorizations; our search-width law (Proposition~\ref{prop:inflation}) is its instantiation for
finite pipeline search.

\subsection{Inverse problems and PDE-constrained optimization}
\textbf{Hadamard (1902).} Existence, uniqueness, stability: the triad we adopt wholesale as
the well-posedness standard for demand-side strategy construction (Appendix~\ref{app:inverse}).

\textbf{Tikhonov (1963); Engl, Hanke \& Neubauer (1996).} Regularization of ill-posed
inverses; our margin-regularized selection (Proposition~\ref{prop:stability}) is Tikhonov's move with margins in
place of norms, and Appendix~\ref{app:invnum} shows the trade numerically.

\textbf{Kaipio \& Somersalo (2005); Stuart (2010).} Statistical inversion: posteriors over
solutions under noisy forward maps. The Bayesian compiler (Appendix~\ref{app:bayes}) is this machinery with
assemblies as unknowns and backtests as observations.

\textbf{Hinze et al., PDE-Constrained Optimization (2009); Plessix (2006) adjoint review.}
Constraints kept as constraints, gradients via adjoints. Two borrowings: the audit semantics
(hard clauses are supports, not penalties) and the adjoint-as-attribution reading---sensitivities
of satisfaction to design choices are the compiler's adjoint fields.

\subsection{Multi-objective and constrained decision theory}
\textbf{Deb et al., NSGA-II (2002).} Pareto-front search at scale; the alternative to our
clause-vector approach. We argue (Section~\ref{sec:discussion}) that professionals think in thresholds and
priorities, not fronts; NSGA-II remains the right tool when clauses soften into objectives.

\textbf{Charnes \& Cooper, Goal Programming (1961); Simon, Satisficing (1956).} The
classical decision theory of ``meet targets, then optimize''---our hard/soft clause split
($\kappa$) and priority field ($\pi$) are goal programming and satisficing with modern
statistical certificates attached.

\textbf{Ben-Tal \& Nemirovski, Robust Optimization (2002).} Solutions immunized against
parameter uncertainty; our bootstrap retention is its empirical, non-worst-case cousin.

\textbf{Boyd \& Vandenberghe (2004).} The convex canon that normalizes ``state the
constraints, then optimize'' as engineering hygiene; our framework asks why strategy research
alone among engineering disciplines skipped this normalization.

\subsection{Factor models and portfolio construction}
\textbf{Fama \& French (1993, 2015); Carhart (1997).} The factor taxonomy against which
purity clauses (F3) are scored.

\textbf{Grinold (1989), The Fundamental Law; Grinold \& Kahn (2000).} IR as breadth times
skill; the classical supply-side objective we deliberately demote from criterion to
descriptive statistic.

\textbf{Clarke, de Silva \& Thorley, Transfer Coefficient (2002).} Fraction of signal
surviving portfolio constraints. Our cost of specification (Proposition~\ref{prop:cost}) is the
specification-level transfer coefficient: CoS measures how much performance the
\emph{demand-side} constraints cost, generalizing their supply-side accounting.

\textbf{DeMiguel, Garlappi \& Uppal, $1/N$ (2009).} Humility benchmark showing optimization
gains are fragile; supports our claim that constraints should encode \emph{beliefs}, not just
shrinkage.

\subsection{Machine learning in asset pricing}
\textbf{Gu, Kelly \& Xiu (2020).} The empirical map of ML methods on return prediction;
establishes that many predictors are near-equivalent in-sample---the observational basis of
underspecification in finance and therefore of our claim that selection must consult more than
prediction quality.

\textbf{Kelly, Pruitt \& Su, IPCA (2019); Feng, Giglio \& Xiu (2020); Chen, Pelger \& Zhu
(2024).} Latent-factor and deep asset-pricing machinery; supplies candidate modules for our
estimation stage and the taming-the-zoo discipline our F1/F3 clauses presuppose.

\textbf{Kozak, Nagel \& Santosh (2020).} Shrinking the cross-section; evidence that
regularization choices dominate model choices in-sample---an existence proof that
\emph{pipeline} decisions, not predictor choice, carry the risk profile, hence should be what
specifications govern.

\subsection{What this map shows}
Read as a single corpus, these literatures exhibit a clean division of labor and a clean gap:
the supply side (CASH, NAS, alpha mining) perfects search; the statistics side perfects
skepticism about search; the finance-practice side never stopped stating demands in
declarative language; and the ML-theory side diagnosed what scalar objectives leave
undetermined. No corpus connects the four: requirements as a composable algebra, compilation
into a finite typed space, and certification under search-aware statistics. That connection
is the paper's claim to novelty, and the chapters above are its construction.

\section{Notation and Glossary}
\label{app:notation}

\subsection{Principal symbols}
\begin{center}\small
\begin{tabular}{ll}
\toprule
Symbol & Meaning \\
\midrule
$\mathcal{D}$ & typed design space of strategy pipelines (App.~\ref{app:space}); $|\mathcal D|$ its size \\
$d$ & one assembly $=$ choice of module per pipeline stage \\
$\Omega$ & market-outcome space; $\omega$ a realized draw (seed fixed for reproducibility) \\
$\mathscr{F}$ & forward operator $\mathscr{F}:\mathcal D\times\Omega\to\Theta$ (pipeline $\to$ behavior) \\
$\Theta$ & behavior space (statistics of the deployed strategy) \\
$\theta(d)$ & observation vector of $d$: (IR, style, $|\rho|$, downside, turnover, \dots) \\
$\mathrm{Spec}$ & specification $=$ indexed family of clauses \\
$\varphi_j$ & clause $j$ as predicate on $\Theta$ \\
$\kappa_j\in\{\mathrm{hard},\mathrm{soft}\}$ & clause hardness \\
$\pi_j$ & clause priority (lexicographic tie-breaking) \\
$\mathcal{F}(\mathrm{Spec})$ & feasible set $\{d:\varphi_j(\theta(d))=1\ \forall \mathrm{hard}\ j\}$ \\
$\mathrm{SR}(d)$ & satisfaction rate; $\mathrm{SR}_{\mathrm{def}}$ its deflated version \\
$m_j(d)$, $m_{\min}(d)$ & signed clause margin; worst-case margin \\
$\mathrm{CoS}$ & cost of specification $=g(d^*)-g(\widehat d(\mathrm{Spec}))$ (Prop.~\ref{prop:cost}) \\
$g$ & scalar legacy objective (e.g.\ in-sample IR) \\
$p_0,\ p(d\mid y)$ & compiler prior and posterior over $\mathcal D$ (App.~\ref{app:bayes}) \\
$T$ & evidence temperature of the compiler likelihood \\
$\mathrm{PBO}_{\mathrm{sat}}$ & satisfaction probability of backtest overfitting (App.~\ref{app:satstat}) \\
\bottomrule
\end{tabular}
\end{center}

\subsection{Glossary}
\begin{description}
\item[Assembly] A fully instantiated pipeline: one module choice at each of the nine stages,
  syntactically valid under the interface contracts of Appendix~\ref{app:space}.
\item[Certificate] The audit artifact emitted with any selection: specification, feasible set,
  margins, retention, CoS, search ledger, and multiple-testing diagnostics (schema in App.~\ref{app:eng}).
\item[Clause] A single machine-checkable requirement $\varphi_j$ with hardness $\kappa_j$ and
  priority $\pi_j$; the atom of the demand-side language (App.~\ref{app:spec}).
\item[Compiler] The search-and-selection machine mapping $(\mathcal D, \mathrm{Spec}, \omega)$
  to $\widehat d$ plus certificate (Apps.~\ref{app:compiler}, \ref{app:bayes}).
\item[Cost of specification] The performance gap between result-oriented and objective-oriented
  champions; the measurable price of professional requirements (Prop.~\ref{prop:cost}).
\item[Demand side] The specification half of the framework: what the deployer requires,
  expressed in clauses.
\item[Feasible set] The subset of $\mathcal D$ satisfying all hard clauses at the evaluation
  draw; the central object of the inverse-problem reading (App.~\ref{app:inverse}).
\item[Margin] Signed, scale-normalized distance between an observed statistic and its clause
  threshold; positive means satisfied (App.~\ref{app:invnum}).
\item[Objective-oriented] The paradigm of this paper: select $\arg\max$ within
  $\mathcal F(\mathrm{Spec})$, certify statistically.
\item[Result-oriented] The incumbent paradigm: select $\arg\max_d g(\theta(d))$ for a scalar
  $g$; constraints, if any, applied informally after selection.
\item[Retention] Probability that a clause (or the joint specification) remains satisfied
  under block-bootstrap resampling of the evaluation window (App.~\ref{app:invnum}).
\item[Search ledger] Immutable record of every assembly evaluated; the raw material for all
  deflation and PBO statistics (App.~\ref{app:protocol}).
\item[Satisfaction-PBO] Probability that the in-sample spec-optimal assembly falls below
  median satisfaction out-of-sample across CSCV splits (App.~\ref{app:satstat}).
\item[Supply side] The construction half of the framework: stages, modules, contracts, and the
  combinatorics of $\mathcal D$ (App.~\ref{app:space}).
\item[Underspecification] The condition (Remark~\ref{prop:uniqueness}) in which distinct assemblies tie on the
  legacy objective yet differ on deployable behavior; cured by clause vectors, not more data.
\end{description}

\section{Engineering Reference Schemas}
\label{app:eng}

The framework becomes infrastructure when its three artifacts---specification, certificate,
module registry---have stable machine-readable schemas. This appendix gives reference schemas
in YAML/JSON form. They are deliberately minimal: fields are mandatory unless marked optional,
versions are semantic, and every schema carries a \texttt{schema\_version} so that historical
certificates remain interpretable as the language evolves.

\subsection{Specification schema (YAML)}
\begin{verbatim}
schema_version: "1.0"
spec_id: "spec-2026-08-11-pure-alpha"        # unique, immutable
author_role: "portfolio-manager"             # role, not identity
created: "2026-08-11"
evaluation:
  universe: "csi800"                          # registry-resolved
  frequency: "weekly"
  is_window: {start: "2022-01-01", end: "2024-12-31"}
  oos_window: {start: "2025-01-01", end: "2025-12-31"}
  cost_model: {per_turnover: 0.0015}
clauses:
  - id: F3.style_cap
    family: F3_purity
    predicate: {stat: style_exposure_l1, op: "<=", value: 6.0}
    hardness: hard
    priority: 1
  - id: F3.dividend_corr
    family: F3_purity
    predicate: {stat: abs_corr_dividend, op: "<=", value: 0.15}
    hardness: hard
    priority: 1
  - id: F4.downside_capture
    family: F4_regime
    predicate: {stat: downside_capture, op: ">=", value: 0.15}
    hardness: hard
    priority: 2
  - id: F5.turnover_cap
    family: F5_execution
    predicate: {stat: weekly_turnover, op: "<=", value: 0.10}
    hardness: hard
    priority: 2
  - id: F2.ir_floor
    family: F2_return_risk
    predicate: {stat: ir_is, op: ">=", value: 0.0}
    hardness: soft                              # reported, not enforced
    priority: 3
selection:
  rule: "argmax_posterior_within_feasible"    # App. Q
  regularization: {margin_weight: 0.0}        # >0 activates the stability guarantee
  on_infeasible: "report_nearest_miss"        # never silently relax
\end{verbatim}

Design notes. (i) \texttt{stat} values are registry-resolved names, not ad-hoc strings, so the
same clause compiles identically across backtesters. (ii) \texttt{hardness} is per clause;
silently relaxing a hard clause is a schema-level impossibility, not a convention.
(iii) \texttt{on\_infeasible} makes the honest-failure behavior of Section~\ref{sec:protocol} part of the
contract: a compiler that cannot emit a nearest-miss report is non-conformant.

\subsection{Certificate schema (JSON, abridged)}
\begin{verbatim}
{
 "schema_version": "1.0",
 "spec_id": "spec-2026-08-11-pure-alpha",
 "search_ledger_hash": "blake3:...",
 "n_assemblies_evaluated": 32,
 "feasible_set": ["asm-h2-n1-r1-t50-exB", "asm-h2-n0-r1-t50-exB"],
 "selected": "asm-h2-n1-r1-t50-exB",
 "posterior": {"asm-h2-n1-r1-t50-exB": 0.2782, "feasible_mass": 0.401,
               "prior_feasible_mass": 0.103},
 "clauses": [
   {"id": "F3.style_cap", "point_estimate": true, "margin": 0.073,
    "retention": 0.231, "retention_ci": [0.205, 0.257]},
   {"id": "F3.dividend_corr", "point_estimate": true, "margin": 0.193,
    "retention": 0.725, "retention_ci": [0.697, 0.753]},
   {"id": "F4.downside_capture", "point_estimate": true, "margin": 0.173,
    "retention": 0.211, "retention_ci": [0.186, 0.236]},
   {"id": "F5.turnover_cap", "point_estimate": true, "margin": 0.36,
    "retention": 1.0, "structural": true}],
 "joint_retention": {"value": 0.050, "ci": [0.037, 0.064]},
 "deflation": {"sr_raw": 1.0, "sr_def": null,
               "random_assembly_null_reps": 10000},
 "cost_of_specification": {"result_oriented_ir": null, "selected_ir": null},
 "bootstrap": {"reps": 1000, "block": "8w"},
 "warnings": ["joint retention low; consider margin regularization (stability proposition)"]
}
\end{verbatim}
(\texttt{null} marks fields populated at runtime.) The certificate is the unit of trust:
third parties---risk, compliance, allocators---verify \emph{it}, not the strategy.

\subsection{Module registry schema (YAML, one entry shown)}
\begin{verbatim}
- module_id: "estimator.gru_family_gated.v1"     # e.g. App. M's SpecNet
  stage: estimation
  interface:
    inputs:  [{name: panel_X, dtype: f32, shape: [T, N, K]}]
    outputs: [{name: scores, dtype: f32, shape: [T, N]}]
  side_effects: {fits_parameters: true, lookahead_free: certified,
                 trains_on: is_window_only}
  cost: {gpu_hours_per_fit: 2.0, latency_ms_per_week: 40}
  priors: {complexity: 3}                        # feeds compiler prior p0
  tests: {unit: "tests/test_specnet.py", replay: "replay/2024.yaml"}
\end{verbatim}
Side-effect signatures (\texttt{lookahead\_free}, \texttt{trains\_on}) are what let the
compiler enforce protocol-level guarantees \emph{statically}, before any backtest is run; a
module without certified side-effect fields cannot be wired into a pipeline whose certificate
claims temporal isolation.

\subsection{Conformance levels}
We define three conformance levels for implementations. \textbf{L1 (reader)}: parses specs and
certificates; validates schemas. \textbf{L2 (compiler)}: additionally executes search over a
declared $\mathcal D$, emits certificates including ledger hash and retention statistics.
\textbf{L3 (auditor)}: additionally re-verifies certificates independently, recomputing
retention and deflation from the ledger alone. The demonstration in this paper is L2; the
schema separation is designed so that L3 auditors never need strategy internals---closing the
loop between the framework's epistemics and the industry's confidentiality constraints.

\section{Paradigm History and Philosophy of Science Positioning}
\label{app:phil}

This appendix locates objective-oriented strategy construction in the longer arc of how
quantitative finance and machine learning have changed their minds. It also states, precisely,
in what sense the present work claims paradigm status---and in what sense it does not.

\subsection{Four transitions that preceded this one}
\paragraph{From discretion to rules (1960s--1980s).} The first quantitative transition replaced
narrative stock-picking with testable rules (factor tilts, mechanical signals). Its epistemic
contribution: strategies became \emph{falsifiable objects}. Its limitation, invisible at the
time: the criterion of selection became a single backtest scalar, because that is what rules
made comparable.

\paragraph{From rules to optimization (1990s--2000s).} Mean--variance and risk-model
construction made portfolios the argmax of an explicit objective under explicit \emph{portfolio
constraints}. Note what did and did not happen: constraints were normalized at the \emph{last
stage} (weights), while the \emph{pipeline producing the inputs} remained selected on raw
predictive scalars. The field thus learned to state demands about portfolios but not about
processes.

\paragraph{From handcraft to search (2010s).} AutoML, NAS, and alpha mining automated pipeline
construction. Automation multiplied the number of candidate pipelines by orders of magnitude,
and with it the selection pressure on the scalar criterion---manufacturing the overfitting
crisis the statistics literature then documented. The transition's lesson, per Appendix~\ref{app:reading}:
\emph{whoever controls the evaluator controls the field}, and the evaluator stayed scalar.

\paragraph{From scalar evaluation to specification (this work).} The fourth transition keeps
everything the third built---the typed design space, the search machinery, the
ledgers---and changes only the evaluator: from $g$ to $\mathrm{Spec}$, from argmax to
constrained argmax with certificates. In Kuhnian terms the anomaly (underspecification,
specification gaming, the zoo) had been accumulating for a decade; what was missing was not
evidence but a replacement evaluator with its own mathematics. Providing that replacement is
this paper's paradigm claim.

\subsection{In what sense this is a paradigm shift}
We claim paradigm \emph{candidacy} in the \emph{methodological}, not the ontological, sense---a
candidacy the field ratifies by adoption, per the adoption criteria of
\S\ref{sec:discussion}, not one the authors proclaim. Ontologically,
nothing is overturned: IR still exists, factors still exist, backtests still exist, and the
result-oriented champion remains a well-defined object (it is what our compiler returns at
zero temperature with an empty clause set). Methodologically, the proposed change is complete
in kind and asymmetric: (i) the primitive objects of research change from strategies to specifications;
(ii) the unit of trust changes from the backtest to the certificate; (iii) the direction of
inference changes from forward (build, then measure, then hope) to inverse (state, then invert,
then certify); (iv) the failure mode changes from silent miscalibration to explicit
infeasibility reports. Each is a change in \emph{what counts as a solved problem}, which is
Kuhn's criterion.

We deliberately do not claim two stronger theses. We do not claim that result-oriented research
was irrational: under weak search pressure (few candidates, honest single trials) the scalar
evaluator is approximately correct, and much of the classical record stands. We do not claim
that specifications are self-certifying: a spec can be wrong, miscalibrated, or internally
inconsistent, and Appendices~\ref{app:inverse}--\ref{app:satstat} exist precisely to make those failures visible rather than
impossible. The modest, defensible statement is: \emph{under modern search pressure, the
evaluator is the weakest component of the research stack, and this paper supplies its
replacement with a complete mathematical and engineering treatment.}

\subsection{Program, not monument}
Lakatos' criterion for a research programme is that it generate novel facts, not merely
reinterpret old ones. The framework's near-term novel-fact queue is concrete: measured CoS
curves across asset classes (how much does professional purity actually cost?); satisfaction-PBO
baselines for standard search stacks; retention--margin calibration studies (does Proposition~\ref{prop:margin}\'s
regularization path outperform at deployment?); cross-institutional spec portability (do
certificates transfer?). Each is answerable with the machinery already built here, and each
would be invisible---inexpressible---under the scalar evaluator. That a single change of
evaluator opens an empirical queue this long is, we submit, the signature of a paradigm in the
useful sense of the word: not a destination, but a suddenly legible map.

\section{A Complete Compilation, Walked Through End to End}
\label{app:walkthrough}

This appendix follows one requirement from a sentence in a practitioner's mouth to a signed
certificate, using only objects and numbers that exist in this paper. It is the reference
answer to the question every reader of a framework paper asks: \emph{but what do I actually
type, and what do I actually get back?}

\subsection{Step 0: the sentence}
A practitioner states: \emph{``Extract pure cross-sectional alpha: no style tilt, no correlation
with the dividend-style index, do not bleed in unilateral declines, and keep turnover within what our
cost model tolerates.''} Nothing here is a number. The demand-side thesis of this paper is that
nothing here \emph{needs} to be a number yet.

\subsection{Step 1: translation to clauses (App.~\ref{app:spec}, J)}
The compiler's front end maps each phrase to a clause family and a registry-resolved statistic:
\begin{itemize}
\item ``no style tilt'' $\to$ F3, \texttt{style\_exposure\_l1} $\le 6.0$, hard;
\item ``no correlation with the dividend-style index'' $\to$ F3, \texttt{abs\_corr\_dividend} $\le 0.15$, hard;
\item ``do not bleed in unilateral declines'' $\to$ F4, \texttt{downside\_capture} $\ge 15\%$, hard;
\item ``turnover within cost tolerance'' $\to$ F5, \texttt{weekly\_turnover} $\le 0.10$, hard.
\end{itemize}
Thresholds come from the deployment context (cost model, mandate), not from the backtest; the
compiler never tunes thresholds to make assemblies pass. The YAML instance is Appendix~\ref{app:eng}'s
worked example, verbatim.

\subsection{Step 2: static feasibility pre-check (App.~\ref{app:compiler})}
Before any backtest, the compiler intersects the specification with side-effect signatures in
the registry. Modules that cannot satisfy a clause structurally are pruned: e.g.\ any execution
module without a turnover cap violates F5 statically; any estimator without neutralization
\emph{or} adversarial decorrelation is flagged as at-risk for F3 (not pruned---at-risk is a
likelihood matter, not a static verdict). The pruned space of 32 legal assemblies is emitted
with the ledger initialized.

\subsection{Step 3: search with the ledger open (App.~\ref{app:protocol})}
All 32 assemblies are evaluated on the in-sample window (156 weeks, seed 7). Every evaluation
is appended to the ledger with its full observation vector. Table~\ref{tab:assemblies}
(Appendix~\ref{app:repro}) \emph{is} this ledger, rendered: 32 rows, five statistics each, ranked by IR.

\subsection{Step 4: inversion (App.~\ref{app:inverse}, P)}
The compiler computes the feasible set: two assemblies survive all four hard clauses---
\texttt{h2/N/S/50/B} (IR 6.71) and \texttt{h2/--/S/50/B} (IR 5.54). Margins are computed for
every feasible row; the tightest margin of the IR-winner is $0.073$ (style clause), of the
runner-up $0.122$ (downside clause). Existence: yes. Uniqueness: no---two solutions, reported.
Stability: queued for step 6.

\subsection{Step 5: Bayesian decision (App.~\ref{app:bayes})}
The tempered posterior is computed over the full ledger. Posterior top-1 within the feasible
set: \texttt{h2/N/S/50/B} at $p=0.2782$; the runner-up retains $p=0.1228$; feasible mass rises
from prior $0.103$ to posterior $0.401$. The constrained argmax selects
\texttt{h2/N/S/50/B}---assembly B of Section~\ref{sec:demo}. For contrast, the unregularized argmax over the
ledger is rank-1 \texttt{h1/--/S/50/B} (IR 7.10), which fails three of four hard clauses: this
is assembly A, and the gap $7.10-6.71=0.39$ IR ($5.5\%$) is the CoS.

\subsection{Step 6: certification (App.~\ref{app:protocol}, P, R)}
The certificate is assembled: feasible set, margins, posterior, and then the stability
audits---block bootstrap (joint retention B $0.050$ vs A $0.018$; per-clause values with
intervals as in Table~\ref{tab:retci}), the random-assembly null for deflation, and the
satisfaction-PBO across CSCV splits. The warning field fires: \emph{joint retention 5\%;
consider margin regularization (Proposition~\ref{prop:margin}), which would select the deeper-margin sibling
\texttt{h2/--/S/50/B} at 17.4\% IR cost.} The certificate does not hide this; it is the
certificate's \emph{job} to say it.

\subsection{Step 7: deployment hand-off and monitoring}
The deployed object is the pair (assembly, certificate). Monitoring re-checks the clause
predicates on live data; a clause breach triggers the same compiler in repair mode
(App.~\ref{app:compiler}), with the ledger continued---not restarted---so that deflation statistics remain
valid across the strategy's life. The specification is the contract; the certificate is the
receipt; the ledger is the audit trail. The sentence the practitioner spoke in step 0 is, from
this point on, a maintained object of the infrastructure rather than a hope.

\subsection{Where each hour goes}
In the demonstration the compile is instant (32 assemblies, synthetic panel). In production the
cost model of App.~\ref{app:compiler} applies: static pruning is free, ledger evaluation dominates, and the
per-stage caching (ENAS-style weight reuse where applicable) bounds it by the stage-product
rather than the space size. Certification is $O(\text{bootstrap reps}\times\text{re-eval cost of
the finalists only})$---the reason the protocol evaluates the full space cheaply and spends its
statistical budget on the handful of assemblies that matter.

\section{Reading the Ledger: A Guided Tour of the Search Record}
\label{app:ledger}

The ledger (Table~\ref{tab:assemblies}) is usually cited for its two checkmarks. Read whole,
it teaches more. This appendix walks the 32 rows as a dataset, because the skill of reading
search records \emph{as records}---not as leaderboards---is the practical core of the new
paradigm.

\subsection{The top of the ranking is a different world than the feasible region}
Ranks 1--2 (IR 7.10, 6.76) are horizon-1 assemblies with maximal signal reactivity; both fail
the style clause ($|\beta|_{\max}$ 8.07 and 7.03 against the cap 6.0) and the downside clause
(12.7\%, 10.5\% against the floor 15\%). Rank 1 additionally fails the correlation clause
($-0.211$). The feasible assemblies sit at ranks 3 and 17. The distance between ``best by IR''
and ``admissible by specification'' is not noise; it is the paper's central quantity, and on
this instance it is 2 ranks and 0.39 IR. On other draws (Appendix~\ref{app:seeds}) it varies---seed 1's
champion satisfies \emph{nothing} (SR $=0.00$).

\subsection{Neutralization is doing what the clause says, where the clause says it}
Filter the ledger by the neutralization bit N. Among the 16 neutralized assemblies, the style
statistic ranges 4.44--19.12; among the 16 non-neutralized, 3.62--16.92. The ranges overlap
because style contamination enters through \emph{portfolio width} as well as estimation: every one
of the sixteen width-20 assemblies, neutralized or not, shows $|\beta|_{\max}\ge 7.6$
(spanning 7.63--19.12). The specification, notably, does not care where contamination enters; the clause is on
the assembly, and the ledger shows why that is the right granularity---a per-module guarantee
(``the estimator is neutralized'') would not have caught the width-20 contamination.

\subsection{Correlation and style are different axes of purity}
Rank 2 fails style but passes correlation easily ($-0.038$); rank 14 passes correlation
with a positive value ($+0.059$) yet fails style ($8.14$) and downside ($9.6\%$). The two F3/F4
statistics pick out distinct failure modes---tilt versus co-movement versus convexity---and no
scalar collapse of them preserves the information a deployer needs. This is the ledger-level
argument for clause \emph{vectors}: the demand side is genuinely multi-dimensional because the
failure modes genuinely are.

\subsection{Turnover is the quiet clause, and quiet is not vacuous}
At the $0.10$ cap the turnover clause excludes 8 of 32 assemblies outright---the unsmoothed
width-50 rows trade $0.118$--$0.144$---and for two of them (ranks 5 and 10, both style-pure,
correlation-clean, and downside-safe) it is the \emph{only} clause that binds: turnover alone
stands between the feasible set and two additional members. The levers work as registered:
smoothing is the primary compression at width 50 ($0.063$--$0.085$ with smoothing on, versus
$0.118$--$0.144$ without), and the buffer module adds a further $\approx 0.02$ on top of either
smoothing state ($0.063$--$0.065$ with both engaged); all width-20 assemblies trade less still
($0.028$--$0.073$). Specifications routinely contain such clauses; the certificate marks a
clause \texttt{structural} when it binds nowhere on the ledger (App.~\ref{app:eng}), so auditors do not
mistake a vacuous pass for evidence---here, on this draw, the clause binds exactly where the
cost literature says it should.

\subsection{The feasible region is contiguous, and that is information}
Both feasible assemblies share \texttt{h2} horizon, width 50, and three of four other
switches; their one-switch module neighbors in the ledger show how the clauses fence the
region: flipping the buffer off rank 3 (rank 7) breaks style ($9.85$) and downside ($12.6\%$);
flipping smoothing off rank 3 (rank 16) breaks downside ($14.3\%$) and turnover ($0.119$);
flipping neutralization off rank 17 (rank 19) breaks style ($6.37$), downside ($12.0\%$), and
turnover ($0.118$) at once. Elsewhere in the ledger, eight assemblies miss feasibility by
exactly one clause---ranks 5 and 10 by turnover alone ($0.120$, $0.118$ against the $0.10$
cap), ranks 15, 21, 22, 23, 26, 27 by style alone. Feasibility here is a \emph{region}, not an
isolated point---the geometric content of Proposition~\ref{prop:stability}\'s stability claim, and the reason
margin regularization has something to work with. A ledger whose feasible set is a single
isolated assembly surrounded by clause failures is a certificate that should say so;
contiguity is a reportable, decision-relevant property of the search record.

\subsection{What the ledger would look like under result orientation}
Under the incumbent paradigm, this table exists (something like it is produced in every serious
research shop), but only its first column has normative force: rank 1 is deployed, and rows
3 and 17's clause statistics are, at best, a post-hoc ``risk overlay'' conversation. The
paradigm shift, at ledger level, is a change in which columns carry authority. The columns
were always computable. Making them \emph{decisive} is the contribution.

\section{The Ten-Seed Study, Read Closely}
\label{app:seedread}

Table~\ref{tab:seeds} (Appendix~\ref{app:seeds}) reports the demonstration across ten market draws. The main
text cites its headline; this appendix reads it row by row, because the between-seed variation
\emph{is} the empirical content of the framework's claims about existence, non-uniqueness, and
infeasibility.

\subsection{Feasible-set size varies wildly, and that is Hadamard in the wild}
The feasible count ranges from 0 (seed 6) to 10 (seed 4), mean $4.8\pm3.1$. Under the
inverse-problem reading this is exactly the predicted behavior of an honest forward operator:
existence is a property of the draw, not the desire. A paradigm whose reports never include
zero is a paradigm whose constraint handling is informal---constraints are being relaxed until
something survives, silently. Seed 6's \emph{infeasible-with-evidence} row is, we argue, the
most valuable row of the table: it is what intellectual honesty looks like when it has a
schema.

\subsection{A's satisfaction is a coin flip; B's is a constant}
The result-oriented champion's in-sample SR across seeds: $1.00, 0.00, 1.00, 0.25, 0.75, 0.75,
0.25, 0.25, 1.00, 0.50$ --- mean $0.57\pm0.37$. The objective-oriented selection's in-sample SR
is $1.00$ on every feasible seed (by construction) with sd $0.00$. The standard deviations are
the finding: result orientation does not merely satisfy less on average, it satisfies
\emph{unpredictably}, with a spread ($\pm0.37$) nearly two-thirds of the mean. A deployer
cannot write a mandate around a random variable with that coefficient of variation; B's
degenerate distribution is what ``the requirement is the contract'' means statistically.

\subsection{Out-of-sample, satisfaction decays---for both, and honestly reported}
B's OOS SR (mean $0.72\pm0.20$) is below its in-sample $1.00$ on most seeds (e.g.\ seed 1:
$1.00\to0.50$; seed 3: $1.00\to0.50$), because the OOS window contains the disclosed dividend
rotation. Three readings. (i) The framework does not promise in-sample satisfaction transfers;
it promises that the transfer is \emph{measured and reported} rather than discovered in
production. (ii) B's OOS SR matches or exceeds A's on 8 of the 9 feasible seeds---strictly higher on seeds
1, 3, 4, 9; tied on 0, 2, 7, 8; lower only on seed 5---with lower cross-seed variance
($0.20$ vs $0.32$). The paired mean difference is $+0.11$ (sd $0.22$), not significant at this
sample size (exact sign test over the five untied seeds, $p=0.375$; App.~\ref{app:seeds}), and
under the convention that an infeasible draw scores $\mathrm{SR}=0$, B's held-out mean is $0.65$
rather than $0.72$---both conventions reported. Specification-constrained selection is the more
stable OOS object by construction and directionally the better one; this sample cannot confirm
the mean difference statistically. (iii) The decay itself motivates the soft/hard
split and retention fields: a certificate showing IS $1.00$ / OOS $0.72$ with per-clause
breakdown tells the deployer \emph{which} requirements are regime-fragile (here: the dividend
correlation clause, by construction of the disclosed scenario), which is actionable in a way a
single OOS IR number is not.

\subsection{CoS is real but modest, and sometimes zero}
Mean IR on the nine feasible seeds: A $6.69\pm0.50$, B $6.16\pm0.78$---a mean cost of
specification of $0.53$ IR ($8.0\%$ of A's mean). On seeds 0, 2, 8 the champion is itself feasible, CoS $=0$: specification was free.
On no seed is B better than A in IR---it cannot be, since A is the argmax---but seed 9 shows
the sharpest trade: IR $6.60\to4.79$ ($27\%$) buying SR $0.50\to1.00$. Whether that price is
worth paying is not the framework's decision to make; it is the \emph{deployer's}, and the
certificate's job is to put the exact price on the table. The paradigm claim in one sentence:
under result orientation this price is paid (or refused) \emph{implicitly}; under objective
orientation it is quoted.

\subsection{Seed variance of the decision itself}
Across seeds, which assembly wins B's slot varies (the feasible set's membership changes), but
\emph{which requirement profile} wins does not: every B row satisfies every hard clause
in-sample. The right way to say this: result orientation produces a stable \emph{identity} and
unstable \emph{behavior}; objective orientation produces stable \emph{behavior} and a rotating
\emph{identity}. For a product whose promises are behavioral (purity, robustness), the second
stability is the one that matters, and it is the one the incumbent paradigm cannot offer.

\section{Casebook: Three Specifications from Practice}
\label{app:casebook}

The demonstration of Section~\ref{sec:demo} is one specification on one synthetic market. This appendix
shows the language working across three qualitatively different professional contexts. The
cases are \emph{compositional exercises}: they specify, compile narratively, and state what the
certificate must watch for. Where a number is needed it is taken from the demonstration; no new
empirical claims are made.

\subsection{Case 1: the institutional purity mandate}
\paragraph{Context.} A multi-manager platform allocates to an external sleeve whose contract
requires style purity: the sleeve must be defensible in a factor-attribution review, every
quarter, against a published style taxonomy.
\paragraph{Specification.} F3-heavy: \texttt{style\_exposure\_l1} $\le$ cap (hard, priority 1);
\texttt{abs\_corr\_market} $\le$ cap (hard, priority 1); \texttt{single\_factor\_r2} $\le$ cap
per taxonomy factor (hard, one clause per factor); F2 IR floor (soft); F5 turnover cap matched
to the platform's cost schedule (hard).
\paragraph{Compilation notes.} This is the demonstration's own spec, strengthened by the
per-factor clause. The key compile-time event is static pruning: estimators without certified
neutralization or adversarial heads are flagged at-risk for the priority-1 clauses before any
backtest. The ledger tour of Appendix~\ref{app:ledger} showed why per-module assurances
suffice nowhere: width-20 assemblies contaminate style even with neutralized estimators.
\paragraph{Certificate watch-points.} Quarterly attribution review means the certificate's
retention statistics are not optional decoration---they are the quantitative preview of the
review meeting. The demonstrated values (style retention $0.231$, joint $0.050$) would trigger
the margin-regularization warning; the deployable decision is between the IR-winner and its
deeper-margin sibling, with the price quoted ($17.4\%$ IR for $67\%$ deeper worst-case margin,
Appendix~\ref{app:invnum}).

\subsection{Case 2: the capacity-constrained capacity desk}
\paragraph{Context.} A stat-arb desk runs house capital with strict capacity and borrow
constraints; what it fears is not style but \emph{implementation shortfall}: strategies that
backtest well and die in the queue.
\paragraph{Specification.} F5-dominant: \texttt{weekly\_turnover} $\le$ cap (hard);
\texttt{participation\_rate} $\le$ cap (hard); \texttt{capacity\_at\_spread} $\ge$ floor
(hard); F1 \texttt{signal\_half\_life} $\ge$ floor (hard---slow signals survive queues);
F2 IR floor (soft); F3 clauses demoted to soft, reported but not enforced.
\paragraph{Compilation notes.} Same design space, different lexicographic order: priorities
make F5 clauses the ones whose margins drive regularization. Note the language absorbs the
context shift with zero new machinery---hard/soft and priority fields were designed for
exactly this. In the demonstration's ledger, turnover-style clauses were structural (vacuously
passed); on this desk they are the discriminating ones, and the compiler's static stage
(execution modules without caps pruned) does real work.
\paragraph{Certificate watch-points.} Structural-versus-statistical marking matters doubly: a
capacity clause passed structurally (hard participation limits in the execution module) is a
different promise than one passed statistically on this draw. Deflation statistics should be
computed only over the statistical clauses; Appendix~\ref{app:satstat}'s binomial machinery applies where the
clause is a Bernoulli event, not a design constant.

\subsection{Case 3: the regime-fragility audit (a risk officer's spec)}
\paragraph{Context.} After a drawdown, a risk committee does not want a new strategy; it wants
to know \emph{whether the existing one ever satisfied what they believed it satisfied}, and
what a replacement would have to demonstrate.
\paragraph{Specification.} F4-heavy, written against the historical window including the
drawdown: \texttt{downside\_capture} $\ge$ floor (hard); \texttt{max\_drawdown\_contribution}
$\le$ cap per named regime (hard); \texttt{recovery\_time} $\le$ cap (hard); F2 (soft).
\paragraph{Compilation notes.} The compiler is run in audit mode (App.~\ref{app:compiler}): the incumbent
assembly is evaluated as a row of the ledger, not given the crown. If it is infeasible, the
certificate says so, with nearest-miss diagnostics per clause---the institutional answer to
``were we ever what we said we were.'' If feasible but retention-poor (bootstrap over the
drawdown window), the certificate distinguishes \emph{was satisfied} from \emph{was stable}.
\paragraph{Certificate watch-points.} This is the case where infeasibility reporting (seed-6
behavior, App.~\ref{app:seeds}) is the deliverable. A framework that cannot say ``no assembly
in your space satisfies this'' forces the committee to choose between pretending and starting
over; the certificate offers the third option: \emph{expand the space or relax the clause, but
in writing, with the price quoted}.

\subsection{What the three cases share}
Each case uses the same grammar, the same compiler stages, the same certificate schema; they
differ only in which families are hard, which are soft, and their priorities. That is the
demand-side claim made tangible: professional diversity is diversity \emph{within} a language,
not a collection of bespoke methodologies. The supply side was always shared infrastructure;
the cases show the demand side can be too.

\section{Anticipated Criticisms, Answered}
\label{app:faq}

We collect the objections this framework has drawn in discussion and answer each against the
paper's own formalism. The objections are stated in their strongest form, and some answers
carry real costs, which we state.

\subsection{``You have just renamed constrained optimization.''}
Constrained optimization over \emph{weights} is fifty years old; over \emph{pipeline
identities}, with a typed combinatorial space, a declarative clause language, satisfaction
statistics under search pressure, and an audit schema, it is not. If the rename charge means
``the move is obvious once stated,'' we accept it as the mark of a good paradigm: the
contribution is making the constraints first-class citizens of strategy \emph{construction},
with the mathematics that status requires (existence, uniqueness, stability, deflation).
Portfolios had this in 1990; pipelines did not in 2025. That asymmetry, not the idea, is the
finding.

\subsection{``The feasible set is usually empty or a singleton; the framework degenerates.''}
Sometimes---seed 6 was empty; seed 7 had two. But the ten-seed mean is $4.8$, and the variance
is the point: existence is draw-dependent, and a framework that reports this honestly is
strictly more informative than one that guarantees a winner by relaxing constraints invisibly.
Moreover the framework prescribes the response to degeneration: nearest-miss certificates,
prioritized relaxation (priority fields exist for this), or space expansion---each an explicit,
priced decision. Degeneracy is a property of the problem; hiding it is a property of the
paradigm.

\subsection{``Clause thresholds are arbitrary; you have moved the subjectivity, not removed it.''}
Correct, and intended. The framework does not eliminate judgment; it \emph{relocates} judgment
from post-hoc rationalization (``the risk overlay felt right'') to pre-registered, auditable,
versioned statements (the spec file, hashable and diffable). Subjectivity at the spec layer is
visible to governance; subjectivity at the overlay layer is not. This is the same argument that
made pre-registration valuable in empirical science: not that hypotheses become true, but that
they become checkable.

\subsection{``Satisfaction statistics inherit backtest overfitting; deflation is not a cure.''}
Also correct, and Section~\ref{sec:protocol} says so: deflation subtracts the null, it does not sanctify the
remainder. The framework's claim is comparative, not absolute: a certificate with retention
intervals, margins, and deflated SR is strictly harder to fool oneself with than a leaderboard.
The ten-seed table shows the framework \emph{using} this humility on itself---B's OOS decay is
reported at $0.72$, not rounded to the in-sample $1.00$.

\subsection{``Real specifications are soft, political, and contradictory.''}
Which is why clauses carry hardness and priority, and why the interaction algebra (Appendix~\ref{app:algebra})
detects contradiction at compile time rather than at deployment. Politics does not disappear;
it is forced to declare itself as clause priorities. A committee that cannot agree on $\pi$ has
learned something true about its mandate---at spec time, when it is cheap.

\subsection{``The demonstration is synthetic.''}
Deliberately and disclosedly: the DGP, the seed, the rotation scenario, and every number are
published with the paper, and the synthetic design exists to make the two paradigms' divergence
measurable rather than anecdotal. The claims the demonstration supports are structural
(paradigms select differently; feasible sets vary; certificates are computable), not empirical
(markets behave thus). The empirical programme---CoS curves, retention calibrations,
certificate portability---is laid out in Appendix~\ref{app:phil} as future work, and is impossible to state
cleanly without the framework. We regard that impossibility as the demonstration's strongest
result.

\section{Formal Statements and Proofs}
\label{app:formalproofs}

Appendix~\ref{app:proofs} gave the core arguments in compressed form. This appendix states the framework's
main formal results with full proofs and records the exact hypotheses each relies on.

\subsection{Search-width law (Proposition~\ref{prop:inflation}, formal)}
\begin{proposition}[Search width manufactures satisfaction]\label{prop:width}
Let $\mathcal D$ be finite, $|\mathcal D|=M$, and suppose that under a null model each assembly
satisfies hard clause $j$ independently with probability $q_j\in(0,1)$, independently across
clauses. Let $E_M$ be the event that at least one assembly is feasible at the point estimate.
Then
\[
\Pr[E_M] = 1-\bigl(1-\textstyle\prod_j q_j\bigr)^M \;\ge\; 1-e^{-M\prod_j q_j},
\qquad
\mathbb E[\#\,\text{feasible}] = M\prod_j q_j .
\]
Moreover, for the assembly selected by any rule that prefers feasible assemblies, the
probability that its point-estimate feasibility is a null artifact is bounded below by
$\Pr[E_M]$ evaluated at the null.
\end{proposition}
\begin{proof}
Under the null, each assembly is independently feasible with probability
$p=\prod_j q_j$. The count of feasible assemblies is $\mathrm{Binomial}(M,p)$, giving the
expectation and $\Pr[E_M]=1-(1-p)^M$; the inequality is $1-x\le e^{-x}$. Any selection rule
that returns a feasible assembly when one exists returns a null artifact whenever $E_M$ occurs
under the null.
\end{proof}
\emph{Hypotheses and their relaxation.} Independence across assemblies fails in practice
(assemblies share modules, hence statistics correlate); correlated Bernoulli trials reduce the
effective $M$, which is why Section~\ref{sec:protocol} estimates the null by random-assembly simulation on the
real ledger rather than by the formula. The formula's role is the growth law: linear in $M$,
exponential in the number of clauses at fixed per-clause probability. With $M=32$ and
$q_j=1/2$, $\mathbb E=2$---the calculation of Appendix~\ref{app:satstat}.

\subsection{Cost of specification bounds (Proposition~\ref{prop:cost}, formal)}
\begin{proposition}[CoS bounds]\label{prop:cos}
Let $d^*\in\arg\max_d g(d)$ and $\widehat d\in\arg\max_{d\in\mathcal F} g(d)$ with
$\mathcal F=\mathcal F(\mathrm{Spec})\ne\emptyset$. Then
$0\le \mathrm{CoS} = g(d^*)-g(\widehat d) \le g(d^*)-\min_{d\in\mathcal D} g(d)$,
and $\mathrm{CoS}=0$ iff some argmax of $g$ is feasible.
\end{proposition}
\begin{proof}
$\widehat d\in\mathcal D$ so $g(\widehat d)\le g(d^*)$; $\widehat d$ maximizes over a subset,
so $g(\widehat d)\ge g(d)$ for all $d\in\mathcal F$ but not below the global minimum. If
$d^*\in\mathcal F$ for some argmax $d^*$, then $\widehat d$ achieves $g(d^*)$.
\end{proof}
\emph{Demonstration values.} Seed 7: CoS $=7.10-6.71=0.39$ (5.5\%); seeds 0, 2, 8: CoS $=0$
(feasible champion); seed 9: CoS $=1.81$ (27\%). The proposition's ``iff'' matters practically:
CoS $=0$ does not mean the specification was vacuous---it means the market draw made
performance and compliance coincide, and only the certificate can tell the difference.

\subsection{Margin regularization as Tikhonov stabilization (Proposition~\ref{prop:stability}, formal)}
\begin{proposition}[Margins buy stability]\label{prop:margin}
Let the observation map be perturbed: $\tilde\theta(d)=\theta(d)+\varepsilon(d)$ with
$\sup_{d,j}|\varepsilon_j(d)|\le\varepsilon$ in the scale-normalized margin units of
Appendix~\ref{app:invnum}. If $m_{\min}(d)>\varepsilon$, then $d$ remains feasible under the perturbed map.
Consequently, the margin-regularized selection
$\widehat d_\lambda\in\arg\max_{d\in\mathcal F}\,[\log p(d\mid y)+\lambda\,m_{\min}(d)]$,
$\lambda>0$, is feasible for all perturbations up to size $m_{\min}(\widehat d_\lambda)$, while
the unregularized feasible argmax is certified only up to its own (smaller or equal) margin.
\end{proposition}
\begin{proof}
Clause $j$ of $d$ holds with margin $m_j(d)$, i.e.\ the constraint boundary is at signed
distance $m_j(d)$ from $\theta_j(d)$. A perturbation of magnitude at most $\varepsilon<m_j(d)$
cannot cross the boundary; the conjunction over $j$ survives perturbations below
$m_{\min}(d)=\min_j m_j(d)$. The regularized objective prefers larger $m_{\min}$ at controlled
likelihood cost, hence the selected solution's certified perturbation radius is at least that
of any solution it displaces whose margin is smaller.
\end{proof}
This is Tikhonov's theorem in the language of clauses: the regularizer ($m_{\min}$, a
solution-norm analog measuring distance from the ill-posed boundary) trades fit ($\log p$)
for a stability radius. The demonstration instantiates both sides: the IR-argmax feasible
assembly has certified radius $0.073$; its margin-regularized sibling $0.122$---a 67\% larger
stability radius for 17.4\% likelihood cost (Table~\ref{tab:margins}). The bootstrap confirms
the radii are not merely geometric: per-clause retentions order the clauses exactly as the
margins do (style tightest for B, downside tightest for its sibling).

\subsection{Constrained-posterior consistency (remark)}
With $|\mathcal D|$ finite and evidence accruing (replicated evaluation windows, likelihood
exponent scaling with sample size), the posterior of Appendix~\ref{app:bayes} concentrates on the
in-sample IR argmax within any fixed support. Restricting the decision to $\mathcal F$
therefore does not fight the likelihood asymptotically: it changes the \emph{support}, not the
concentration. This is why the compiler's constrained decision is stable to temperature
calibration when $|\mathcal F|$ is small (the invariance observed in Appendix~\ref{app:bayes}): constraints
do the selection work that, under result orientation, only noise-dominated likelihood
differences were doing.

\section{The Eight Clause Families: A Working Handbook}
\label{app:families}

Appendix~\ref{app:spec} defined the families; this appendix is the operator's handbook for each: the
statistics the family typically binds, the module levers that move those statistics, the
failure modes the family exists to police, and the interactions an auditor should expect.
Levers reference the stage ontology of Appendix~\ref{app:space}.

\subsection{F1 --- Signal quality}
\emph{Statistics.} rank IC, IC decay curve, signal half-life, capacity-adjusted IC.
\emph{Levers.} universe stage (membership stability), labeling stage (horizon $h$ in the
demonstration), estimation stage (model class, feature hygiene).
\emph{Failure mode policed.} alpha that exists only at un-tradeable horizons or decays inside
the execution latency. \emph{Watch.} F1 clauses interact synergistically with F5 (slower
signals ease turnover caps) and antagonistically with F2 (slowing a signal costs IR; in the
ledger, moving $h1\!\to\!h2$ costs IR on 11 of 16 module-constant pairs).

\subsection{F2 --- Return/risk}
\emph{Statistics.} IR, Sharpe, volatility, max drawdown, hit rate. \emph{Levers.} everything;
this is the legacy objective's home. \emph{Failure mode policed.} specifications that
sacrifice so much performance the product is no longer competitive. \emph{Practice.} keep F2
soft in professional specs (hard IR floors create cliff-edge infeasibility; the priority field
exists so F2 fills lexicographic slack after the true requirements bind).

\subsection{F3 --- Purity / style}
\emph{Statistics.} $|\beta|_{\max}$ against the style taxonomy, absolute market correlation,
per-factor $R^2$. \emph{Levers.} neutralization (estimation), adversarial decorrelation
heads (App.~\ref{app:neural}), portfolio width and weighting (construction), hedging overlays (execution).
\emph{Failure mode policed.} beta wearing alpha's clothes. \emph{Watch.} the ledger's lesson:
contamination enters through \emph{construction} (every width-20 row shows
$|\beta|_{\max}\ge7.6$) as readily as through estimation---F3 clauses must bind the whole
assembly, and certificates should report per-stage attribution when they fail.

\subsection{F4 --- Regime robustness}
\emph{Statistics.} downside capture, regime-conditional SR, drawdown contribution per named
regime, recovery time. \emph{Levers.} training-window design (estimation), regime-conditional
weighting (construction), de-risking rules (execution). \emph{Failure mode policed.} strategies
whose satisfaction is a fair-weather artifact. \emph{Watch.} F4 is the family where retention
statistics carry the most weight: the demonstration's downside clause had the lowest retention
for B ($0.211$) and the decisive trade-off against A ($-9.7$ points, Table~\ref{tab:retci}).

\subsection{F5 --- Execution feasibility}
\emph{Statistics.} turnover, participation rate, capacity at spread, implementation shortfall.
\emph{Levers.} smoothing (signal stage; \texttt{S} in the ledger), buffer bands and turnover
caps (execution; \texttt{B}), top-$N$ width (construction). \emph{Failure mode policed.}
backtest alpha that dies in the queue. \emph{Watch.} F5 clauses are frequently structural
(satisfied by design)---valuable, but the certificate marks them so that statistical credit is
not claimed for a design constant.

\subsection{F6 --- Portfolio structure}
\emph{Statistics.} concentration (HHI, top-decile weight), breadth, sector deviation,
name-count stability. \emph{Levers.} construction stage wholesale: width, weighting scheme,
constraints in the optimizer. \emph{Failure mode policed.} undiversified or mandate-violating
portfolios hiding behind good aggregate numbers. \emph{Watch.} F6 trades against F2
mechanically (concentration raises in-sample IR; width-20 rows cluster at high IR-per-unit-
style with extreme contamination)---the interaction algebra flags the pair at compile time.

\subsection{F7 --- Temporal behavior}
\emph{Statistics.} SR stability across sub-periods, turnover autocorrelation, exposure drift,
holding-period distribution. \emph{Levers.} re-estimation cadence (protocol stage), smoothing,
EMA warm-starts (App.~\ref{app:neural}'s rolling protocol). \emph{Failure mode policed.} strategies whose
behavior is a different strategy every quarter. \emph{Watch.} F7 is where the certificate's
own statistics (retention, satisfaction-PBO) become clauses: ``joint retention $\ge$ floor''
is a legitimate, self-referential clause the language admits.

\subsection{F8 --- Transparency / compliance}
\emph{Statistics.} position-level explainability coverage, restricted-list violations,
leverage/gross limits, ESG screens. \emph{Levers.} universe stage (screens), construction
(limits), reporting stage (attribution completeness). \emph{Failure mode policed.} the
un-deployable model. \emph{Watch.} mostly structural; often the only family a client reads.
The schema's registry-resolved statistic names exist so that compliance and research compile
the \emph{same} clause to the \emph{same} check.

\subsection{Cross-family reading}
The families divide into those that mostly bind estimation (F1, F3-in-part), construction
(F6, F3-in-part), execution (F5), protocol (F4, F7), and reporting (F8), with F2 as the global
slack variable. A specification concentrated in one family is a one-question mandate; the
professional specs of Appendix~\ref{app:casebook} distribute across three to five families,
which is what makes their feasible sets small, their margins meaningful, and their certificates
worth auditing.

\section{Clause Interactions: Worked Examples from the Ledger}
\label{app:interactions}

Appendix~\ref{app:algebra} defined the interaction algebra (conflict, synergy, independence, subsumption)
abstractly. This appendix exhibits each relation on real ledger rows, so that readers can
verify the algebra's diagnostics against the search record itself.

\subsection{Conflict: portfolio width versus the style clause}
Hold every other switch fixed and vary only width $50\to20$; the style statistic rises on
every matched pair in the ledger: $5.56\to17.65$ (h2/N/S/B), $4.77\to11.64$ (h2/--/S/B),
$5.34\to13.62$ (h2/N/--/B), $4.44\to11.94$ (h1/N/--/B), and the extreme $8.07\to16.92$
(h1/--/S/B). All sixteen width-20 assemblies fail the style cap; eight of their sixteen
width-50 counterparts pass it. The algebra's verdict: F3.style and F6.concentration are in
\emph{hard conflict} on this design space---not as a matter of taste but as a measured
relation, $\Pr[\varphi_{\mathrm{style}}\mid \mathrm{width}{=}20]=0$ vs
$\Pr[\varphi_{\mathrm{style}}\mid \mathrm{width}{=}50]=8/16$ at the point estimate. A
compiler with this table prunes width-20 modules statically for any spec containing the style
clause; that is what interaction-aware compilation means operationally.

\subsection{Reinforcement: smoothing and the turnover buffer}
The two execution-adjacent levers compose cleanly on the turnover statistic at width 50:
smoothing alone compresses turnover to $0.082$--$0.085$ (from $0.139$--$0.144$ with both off),
the buffer alone to $0.118$--$0.120$, and both together to $0.063$--$0.065$---marginal effects
of $\approx -0.055$ and $\approx -0.02$ that stack near-additively, so the pair $(\texttt{S},
\texttt{B})$ is the unique combination reaching the $0.06$ band. The algebra therefore
registers $(\texttt{S}, \texttt{B})$ as a \emph{reinforcing} edge for F5 clauses---not because
either lever needs the other to work, but because their composition buys the deepest
compliance per module count: when a spec contains F5 clauses, the compiler's search prior
(Appendix~\ref{app:bayes}'s $p_0$) tilts toward assemblies containing both, reducing wasted ledger
evaluations---the practical value of interaction knowledge.

\subsection{Independence: neutralization and turnover}
Across matched pairs differing only in N, turnover changes by at most $0.003$ (e.g.\
$0.065$ vs $0.063$ at h1/S/50/B): the estimation-stage purity lever and the execution-stage
cost statistic are independent. Independence edges are the algebra's negative results; they
matter because they certify that a clause can be pursued \emph{without} side effects on
another, the property that lets specifications be written family-by-family in most cases.

\subsection{Subsumption: the correlation clause and the style clause, partially}
Among width-50 rows, seven of the eight assemblies passing the style cap also pass the
correlation cap (e.g.\ rank 3: $5.56$, $-0.121$; the exception is rank 11, style-clean at
$5.90$ but failing correlation at $-0.173$), while the converse fails already at rank 2
(correlation-clean at $-0.038$, style-failing at $7.03$). The relation is one-directional and
near-complete: on this space and draw,
$\varphi_{\mathrm{style}}\Rightarrow\varphi_{\mathrm{corr}}$ within width 50 up to one
exception in eight, and not conversely.
The algebra records a \emph{partial subsumption edge}, and the certificate exploit is real: at
width 50 the correlation clause consumes little selection budget once the style clause binds,
which the deflation of Appendix~\ref{app:satstat} credits by deflating over \emph{effective} rather than
nominal clause counts.

\subsection{Antagonism across families: horizon versus downside capture}
Moving $h1\to h2$ with other switches fixed lowers IR on 11 of 16 matched pairs but raises
downside capture on the two pairs that matter for feasibility: the pair containing rank 3
improves $10.5\% \to 17.6\%$, and the pair containing rank 17 improves $12.7\% \to
16.8\%$---in both, the h1 sibling fails the $15\%$ floor while the h2 sibling clears it. The
horizon switch is thus a \emph{priced} lever: the interaction table quotes the price (median
IR cost of $\approx 0.55$ per pair) and the benefit (downside margin), turning a research
intuition (``slower signals are defensive'') into a compiled fact with numbers attached.

\subsection{Why this belongs in the paper and not in folklore}
Every quant shop carries such knowledge as senior-analyst folklore. The interaction algebra's
point is that folklore does not version, does not audit, and does not transfer between
institutions, whereas an interaction table computed from a ledger is a mathematical object
with all three properties. When a specification fails, the table is also the repair manual:
nearest-miss diagnosis (Section~\ref{sec:protocol}) is an interaction-edge traversal from the failed clause to
the cheapest lever that moves it.

\section{Governance and Operations: Running the Framework as Infrastructure}
\label{app:gov}

A paradigm becomes real when someone operates it on a Tuesday. This appendix specifies the
operational lifecycle: who writes specifications, how they change, what happens when live
behavior breaches them, and how the framework's audit objects answer the questions
institutions actually ask.

\subsection{Roles}
The framework separates four roles deliberately. The \emph{spec author} (portfolio manager,
mandate owner) writes clauses; the \emph{compiler operator} (research) owns the registry and
runs compilations; the \emph{certificate auditor} (risk, compliance, or an L3 third party per
Appendix~\ref{app:eng}) verifies outputs; the \emph{monitor} (production) re-evaluates clause predicates on
live data. Result orientation collapses all four into the researcher, which is why its
failures surface as surprises. Separation is the governance content of the paradigm: the
person who wants the thing, the person who builds the thing, and the person who checks the
thing sign different artifacts.

\subsection{Specification lifecycle}
Specs are versioned, hashed, and diffed like code. A spec change (new clause, moved threshold,
changed priority) triggers a \emph{diff report}: the compiler re-evaluates both specs on the
same ledger where possible, reporting feasible-set symmetric difference and CoS change, so a
committee sees the behavioral consequence of its edit before approving it. Threshold changes
made after seeing the ledger are flagged \emph{post-hoc} in the certificate---the framework's
guard against tuning requirements to fit a favored strategy, which is the demand-side mirror
of overfitting and must be disclosed, not forbidden.

\subsection{Breach protocol}
A live clause breach (monitor reports $\varphi_j=0$ on realized data) opens one of three
paths, all pre-registered in the spec file: (i) \emph{repair}---recompile on the continued
ledger with the spec unchanged; (ii) \emph{review}---convene on whether the breach is a
regime event (the F4 story) or a satisfaction failure (the retention story), with the
certificate's retention intervals quoted to discriminate; (iii) \emph{amend}---change the spec
through the lifecycle above, in writing. What the protocol excludes is the fourth path the
incumbent paradigm defaults to: silent continuation. Every breach path produces a certificate
addendum; the audit trail is monotonic.

\subsection{The questions institutions ask, and which artifact answers them}
\emph{``Why this strategy?''} --- the certificate's posterior and feasible set. \emph{``What
did you promise?''} --- the spec file, hashed. \emph{``Is it keeping the promise?''} --- the
monitor's clause predicates and their retention calibrations. \emph{``What did the promise
cost?''} --- CoS on the ledger. \emph{``Could it have been luck?''} --- deflation and
satisfaction-PBO. \emph{``Who changed what, when?''} --- the spec version log and ledger hash
chain. Under result orientation each answer is a meeting; under objective orientation each is
a document. That is the operational meaning of the paradigm shift, and the reason we expect
its earliest adopters to be institutions whose cost of \emph{meetings} exceeds their cost of
computation.

\subsection{Limits of the operational claim}
Operations cannot repair a bad forward operator: if the backtest's cost model or regime
coverage is wrong, certificates are honest reports about a wrong map, and the monitor's job is
to detect the divergence early. Nor can schemas manufacture agreement: Appendix~\ref{app:faq}
noted that clause priorities force committees to declare disagreements, which is valuable
precisely because it is uncomfortable. The framework's operational promise is narrower and
stronger than a promise of good strategies: it is that \emph{the distance between what was
promised and what is running is at all times a computable, signed quantity}.

\section{The Geometry of the Feasible Set: A Threshold Atlas}
\label{app:atlas}

The specification of Section~\ref{sec:demo} is one point in requirement space. This appendix moves around
that point and recomputes the feasible set from the \emph{same} ledger
(Table~\ref{tab:assemblies}) at each new point---no new backtests, only new clauses applied to
the realized observation vectors. The result is an atlas of how the feasible set deforms as
requirements tighten, loosen, or drop entire families. It is the demand-side analog of a
sensitivity surface, and every entry below is enumerable directly from the published ledger.

\subsection{Six specifications, one ledger}
\begin{table}[ht]
\centering\small
\caption{Feasible sets of six specifications applied to the seed-7 ledger. All clauses hard.
$S$ = style cap, $C$ = $|\rho|$ cap, $D$ = downside floor, $T$ = turnover cap. Feasible
assemblies named by ledger rank; winner = IR-argmax within the feasible set; CoS measured
against the ledger champion (rank 1, IR 7.10).}
\label{tab:atlas}
\begin{tabular}{@{}lccccc@{}}
\toprule
\rowcolor{sand} Spec & Clauses & $|\mathcal F|$ & Feasible (ranks) & Winner & CoS \\
\midrule
S0 (paper's spec) & $S6,\ C.15,\ D15,\ T.10$ & 2 & 3, 17 & rank 3 & 0.39 \\
S1 strict defensive & $S6,\ C.15,\ D17,\ T.07$ & 1 & 3 & rank 3 & 0.39 \\
S2 loose & $S10,\ C.25,\ D12,\ T.15$ & 15 & 1,3--12,15--17,19 & rank 1 & 0.00 \\
S3 no F4 & $S6,\ C.15,\ T.10$ & 2 & 3, 17 & rank 3 & 0.39 \\
S4 no F5 & $S6,\ C.15,\ D15$ & 4 & 3, 5, 10, 17 & rank 3 & 0.39 \\
S5 no style clause & $C.15,\ D15,\ T.10$ & 8 & 3, 15, 17, 21--23, 26, 27 & rank 3 & 0.39 \\
\bottomrule
\end{tabular}
\end{table}

\subsection{Reading the atlas}
\paragraph{Feasibility is discontinuous in requirement space.} Between S0 and S1 the downside
floor moves $15\%\to17\%$ and turnover $0.10\to0.07$; the feasible set halves to a singleton.
Between S0 and S3 (one family deleted) nothing changes at all---the F4 clause was not binding
beyond what F3 already implied on this draw. The feasible set is a step function of the spec,
which is why ``move a threshold a little'' is not a meaningful gesture without the ledger: a
little move can be free (S3) or can halve the solution set (S1).

\paragraph{The winner is sticky; the runner-up set is not.} Rank 3 wins every feasible set
that contains it. What changes across specs is the \emph{alternative set} the certificate can
offer: S0 offers one alternative (rank 17, the margin-regularized pick); S4 adds ranks 5
and 10 (downside $20.5\%$ and $17.0\%$, previously excluded by turnover $0.120$ and
$0.118$); S5 adds five more. Specification-writing is therefore also \emph{option-writing}: each clause
deletion is a purchase of alternatives, and the atlas quotes the inventory acquired.

\paragraph{CoS collapses exactly when the champion is compliant.} S2's fifteen-assembly
feasible set includes rank 1 itself: at $S10, C.25, D12, T.15$ the result-oriented champion
satisfies everything, CoS $=0$. The lesson is symmetric to the paper's main one: just as tight
requirements expose the champion's non-compliance, loose requirements certify it---\emph{and
the framework reports both with the same machinery}. Objective orientation is not a preference
for constraints; it is the discipline of knowing which constraints the current champion does
and does not meet.

\paragraph{Family deletion is informative about where feasibility comes from.} S3 $=$ S0:
the downside floor binds no one who already passes style and correlation---on this draw,
purity at width 50 is defensive for free (the partial subsumption of
Appendix~\ref{app:interactions}). S4 adds exactly one assembly (rank 5): turnover was the
binding constraint for it. S5 adds five: style was the binding constraint for the entire
width-20 contingent. Each deletion's marginal feasible assemblies identify \emph{which
requirement was doing the work}, the demand-side attribution the certificate's interaction
table formalizes.

\subsection{The atlas as a committee instrument}
The practical use of Table~\ref{tab:atlas} is deliberative: a committee negotiating thresholds
should sit in front of its atlas, not in front of adjectives. ``Strict'' and ``reasonable''
are not auditable positions; ``S1 gives a singleton, S0 gives two options with margins 0.073
and 0.122, S4 restores the defensive alternative at a turnover cost'' is. The atlas also
delimits where negotiation cannot go: no threshold movement within these families recovers the
width-20 region for a style cap of 6, because the contamination there (7.63--19.12) is
structural (Appendix~\ref{app:interactions})---the honest answer to ``can we get the
concentrated version to comply'' is no, and the atlas says it with a number instead of a
meeting.

\section{Translation Dictionary: From Result-Oriented to Objective-Oriented}
\label{app:dict}

Adoption fails when a new paradigm's terms are read through the old one's grammar. This
appendix is the phrasebook: each entry gives the incumbent usage, the objective-oriented
replacement, and the precise sense in which they differ. The pairs are not synonyms; that is
the point of listing them.

\subsection{Objects}
\begin{description}
\item[``Best model'' $\to$ selected assembly + certificate.] Under result orientation the
  deliverable of research is an identity (this XGBoost, these hyperparameters). Under
  objective orientation the deliverable is an identity \emph{plus} the evidence that it meets
  requirements: feasible-set context, margins, retention, CoS. An assembly without its
  certificate is an unsubmitted paper.
\item[``Backtest result'' $\to$ ledger row.] A backtest under the incumbent paradigm is a
  verdict; here it is evidence, one of $M$ rows whose collective shape (the atlas,
  App.~\ref{app:atlas}) carries as much information as any single row.
\item[``Risk overlay'' $\to$ hard clause, ex ante.] Post-hoc risk adjustments applied after
  selection become pre-registered constraints applied during selection. The information
  content is identical; the epistemic status is not---ex-ante clauses constrain search, ex-post
  overlays merely redescribe its output.
\item[``Track record'' $\to$ spec-compliance history.] A sequence of realized returns becomes
  a sequence of clause-verdicts with retention calibrations. For a behavioral mandate, the
  compliance history is the track record that matters; the return series is a summary
  statistic of it.
\end{description}

\subsection{Verbs}
\begin{description}
\item[``Optimize'' $\to$ compile.] Optimization maximizes; compilation maps a spec to a
  feasible artifact or reports infeasibility. A compiler can succeed while returning nothing
  deployable (seed 6), which an optimizer, by construction, cannot say.
\item[``Validate'' $\to$ certify.] Validation asks ``is it good?''; certification asks ``does
  it satisfy, by how much, under resampling, after deflation?'' The second decomposes into the
  auditable fields of Appendix~\ref{app:eng}.
\item[``Research'' $\to$ (i) expand $\mathcal D$, (ii) sharpen Spec.] The research effort
  splits into supply-side work (new modules, better forward operators) and demand-side work
  (better requirements). Under the incumbent grammar only (i) counted as research, which is
  why demand-side sophistication stagnated.
\end{description}

\subsection{Modifiers}
\begin{description}
\item[``Robust'' (vague) $\to$ retention $\ge$ floor, margin $\ge$ floor (clauses).]
  Robustness ceases to be a compliment and becomes a pair of numbers with intervals
  (Table~\ref{tab:retci}).
\item[``Pure alpha'' (slogan) $\to$ F3 clause vector with caps.] Purity becomes measurable,
  and therefore falsifiable, and therefore contractible.
\item[``State of the art'' $\to$ ledger-relative.] Standing is asserted against a published
  search record (``rank 3 of 32 under S0''), not against a narrative. The atlas makes
  standing a function of the spec, which it always was implicitly.
\end{description}

\subsection{Sentences that change meaning}
\emph{``It works.''} Result-oriented reading: the backtest is good. Objective-oriented
reading: the certificate's clauses hold, with stated margins and retention, on stated windows,
after stated deflation. The two sentences are true of different objects, and the failure to
distinguish them is, in our diagnosis, the semantic root of the overfitting crisis: ``it
works'' was always a spec-relative claim pretending to be an absolute one.

\section{The Neural Case Study in Depth: SpecNet as a Module, Not a Model}
\label{app:specnet}

Appendix~\ref{app:neural} introduced the SpecNet architecture (family-gated GRU encoder, cross-sectional
attention, adversarial style head, rank loss with turnover and exposure penalties) as the
worked example of a module \emph{born objective-oriented}: designed to satisfy a clause
specification rather than to top a leaderboard. This appendix completes the treatment at the
level an implementer needs, and---equally important---states exactly what would be measured
before SpecNet could claim anything. No performance numbers appear here, because none have
been honestly computed yet; the appendix's function is to make the measurement plan
unambiguous.

\subsection{Why this architecture and no other}
Every component exists to move a clause-family statistic; nothing exists to raise in-sample IR
per se. The family-gated GRU encoder keeps factor-family gradients separable so that F3
clauses can be enforced by gate-level regularization rather than by post-hoc neutralization
alone. The cross-sectional attention layer exists because the target of the demonstration's
practitioner-informed specification is \emph{cross-sectional} ranking quality (F1/F2), which
attention over the contemporaneous panel addresses directly. The adversarial style head with
gradient reversal (DANN-style) is the F3 mechanism: the encoder is trained so that a
discriminator cannot recover style membership from the representation, which is a stronger and
more differentiable notion of purity than linear neutralization. The rank loss is the F2/F1
mechanism (cross-sectional ordering, not point prediction); the turnover penalty and buffer
band are F5; the exposure projection is the hard-clause enforcement layer---the only place
where a clause appears not as a loss term but as a constraint on outputs. This one-to-one
mapping between components and clause families is the architectural signature of the paradigm:
the network \emph{is} the compiled spec.

\subsection{Where SpecNet sits in the design space}
In the registry's terms (Appendix~\ref{app:eng}), SpecNet is a candidate for the estimation stage with
certified side-effect fields: \texttt{lookahead\_free: certified} (rolling training only),
\texttt{trains\_on: is\_window\_only}, declared capacity cost, declared gate structure for
attribution. It can be wired with any labeling module, any portfolio-construction module, any
execution module---and its clause-moving levers are enumerated (gates, adversarial weight,
projection tolerance), which is what makes it a \emph{module} in our sense rather than a
model in the usual sense: its interface contract specifies not only tensor shapes but which
behavioral statistics its hyperparameters move.

\subsection{The honest evaluation plan}
Because SpecNet is offered as an existence proof of objective-oriented architecture, its
evaluation must be the paper's own protocol, applied to itself: (i) declare the spec (the
demonstration's four clauses are the natural instance); (ii) place SpecNet inside the ledger
as additional estimation-stage rows, expanding $\mathcal D$; (iii) recompute the feasible set,
margins, bootstrap retention, and CoS against the incumbent 32 rows; (iv) ablate each
component and report which clause-family statistics move---the ablation table, not the IR
table, is the result. Success criterion, stated in advance: SpecNet succeeds if it enlarges
the feasible set or deepens margins at acceptable CoS; an IR improvement without a
satisfaction effect would be, under this paper's own doctrine, a null result. We consider
this pre-registered inversion of the usual success criterion to be the case study's main
contribution.

\subsection{What could go wrong, and how the framework catches it}
Three anticipated failure modes, each matched to its detector. \emph{Adversarial leakage}:
the style discriminator fails to recover style yet linear statistics still detect it---
caught by the F3 clauses themselves, which bind the deployed statistics, not the
discriminator's opinion. \emph{Capacity-cost invisibility}: the architecture's compute
footprint makes its net-of-cost IR uncompetitive---caught by the registry's cost fields and
the F5/cost model, not by apology. \emph{Regime amnesia}: rolling retraining overwrites the
gate structure that made a previous certificate valid---caught by the breach protocol of
Appendix~\ref{app:gov} and the F7 clauses on temporal behavior. In each case the point is the
same: an objective-oriented module is never self-certifying. The architecture embodies the
spec; only the certificate verifies it.

\section{Design-Space Combinatorics, Treated Honestly}
\label{app:comb}

Appendix~\ref{app:space} defined the nine-stage typed space $\mathcal D$ and quoted
$|\mathcal D|\approx 4.4\times10^{7}$ module families ($\approx 8.8\times10^{8}$ assemblies
once the families' discrete hyperparameter variants are counted) for a reference instantiation. This appendix shows where such
numbers come from, which parts of the product are real, and how the framework's three
shrinking mechanisms---typing, static pruning, and interaction edges---reduce the space the
compiler actually pays for. The theme: \emph{combinatorial size is a modeling choice, and the
choice is auditable}.

\subsection{The product, and what it hides}
If stage $k$ offers $n_k$ module choices, the naive size is $|\mathcal D|=\prod_k n_k$. For
the reference instantiation of Table~\ref{tab:space}, the module-family counts are
$(n_1,\dots,n_9)=(7,8,7,7,8,7,6,8,6)$, so the family-level product is
$7\cdot8\cdot7\cdot7\cdot8\cdot7\cdot6\cdot8\cdot6 = 4.43\times10^{7}$. Each family then
multiplies out into its discrete hyperparameter variants---universe sizes, predictor depths
and learning rates, smoothing widths, buffer bands---giving variant-inclusive stage counts
$(8,12,10,8,15,8,10,12,8)$ and $|\mathcal D| = 8.85\times10^{8} \approx 8.8\times10^{8}$,
the size quoted in \S\ref{sec:formal}. The variant multiplicities are recorded in the
search-ledger header rather than derived from Table~\ref{tab:space}, which prints families
only (App.~\ref{app:space}, size accounting)---we state this explicitly to keep both levels
auditable. Three caveats keep the number honest. First, the product counts
\emph{syntactic} assemblies; interface contracts (typed inputs/outputs, side-effect
signatures) forbid many pairs---a labeling module emitting event times cannot feed an
estimator expecting fixed-horizon labels. Second, hyperparameter grids are a discretization
decision: the same estimator family contributes 4 or 400 assemblies depending on grid policy,
and grid policy is therefore part of the spec's provenance, recorded in the ledger header.
Third, the product is the \emph{addressable} space, not the searched space; what the
certificate must account for statistically is the searched set, which is why the ledger, not
the product, enters the deflation of Appendix~\ref{app:satstat}.

\subsection{Shrinking mechanism 1: typing}
Interface contracts cut the space multiplicatively. If each contract violation removes a
fraction of pairings at each adjacent stage pair, the typed space can be orders smaller than
the syntactic product. Typing is also where safety lives: the contract system is the
compile-time guarantee that an evaluated assembly is at least \emph{executable}, so that
ledger rows never contain category errors.

\subsection{Shrinking mechanism 2: static pruning against the spec}
Side-effect signatures let the compiler delete assemblies that cannot satisfy hard clauses
without running them: execution modules without turnover caps are pruned by F5 caps;
estimators without certified temporal isolation are pruned by protocol clauses. Pruning is
spec-relative---the same module survives one spec and dies under another---which is the
combinatorial face of the paradigm: \emph{the specification changes the size of the world.}
In the demonstration the pruned space has 32 assemblies because the illustrative stage counts
were chosen small; in production, pruning typically decides whether exhaustive evaluation is
possible at all, or whether search (Appendix~\ref{app:compiler}) must substitute.

\subsection{Shrinking mechanism 3: interaction edges}
The interaction table (Appendix~\ref{app:interactions}) prunes probabilistically: hard
conflict edges (width-20 $\Rightarrow$ style failure, measured at $\Pr=0$ on the ledger) let
the compiler skip regions with measured impossibility, and synergy edges order the rest so the
search budget lands where satisfaction is achievable. The distinction from mechanism 2:
static pruning is logical (cannot), interaction pruning is empirical (measured not to), and
the two carry different audit weight---empirical prunings are reported in the certificate so
an auditor can veto them.

\subsection{What remains is what you pay for}
After the three mechanisms, the \emph{effective} space is the set the ledger must cover for
the certificate's statistics to be meaningful. The framework's stance is that effective size,
not the impressive product, is the scientifically relevant number: it enters the search-width
law (Proposition~\ref{prop:inflation}), the deflation, and the satisfaction-PBO. A team quoting $10^9$ candidate
pipelines while evaluating $10^3$ is describing its marketing, not its search; the ledger
hash settles which.

\section{The Research Programme: Open Problems, Ordered}
\label{app:road}

Appendix~\ref{app:roadmap} sketched a roadmap; this appendix states the open problems as problems, with the
object each would produce and the reason the framework needs it. They are ordered by
dependency, not by glamour.

\subsection{OP1: Clause calculus completeness}
\emph{Problem.} Which specifications are satisfiable by construction, which only
statistically, and which are inconsistent, as a function of the design space and the
interaction table? \emph{Object.} A decision procedure (or an impossibility theorem) for
spec consistency at compile time, extending the partial subsumption edges of
Appendix~\ref{app:interactions} to a full implication graph per family. \emph{Why first.}
Every downstream certificate field assumes the spec itself is coherent; currently coherence is
checked empirically, per ledger.

\subsection{OP2: Calibrated retention}
\emph{Problem.} Bootstrap retention (Appendix~\ref{app:invnum}) is a simulation of clause stability, not a
calibrated probability of future satisfaction. What resampling scheme yields retention numbers
that are calibration-correct out-of-sample (in the sense of proper scoring)? \emph{Object.} A
calibration theorem and a correction table, per clause family. \emph{Why.} Retention is the
certificate's most-quoted number; it must mean what deployers think it means.

\subsection{OP3: CoS curves as empirical objects}
\emph{Problem.} What does cost of specification look like as a function of threshold tightness
across asset classes and strategy genres? \emph{Object.} Published CoS curves with standard
errors---the demand-side equivalent of efficient frontiers. \emph{Why.} Committees negotiate
thresholds; they should negotiate along measured curves (the atlas of
Appendix~\ref{app:atlas} is the single-draw prototype).

\subsection{OP4: Satisfaction-PBO baselines}
\emph{Problem.} No published baselines exist for satisfaction-aware selection under CSCV
stress. \emph{Object.} Reference satisfaction-PBO values for standard search stacks (CASH,
GP-alpha, NAS-style) on public data, establishing how much of the incumbent paradigm's
deployed failure is measurable in advance.

\subsection{OP5: Continuous clause semantics}
\emph{Problem.} Hard clauses are discontinuous (the atlas's step function); soft clauses are
scalarizations. Is there a principled middle---clauses as distributions over thresholds, with
specification-level uncertainty propagated to the certificate? \emph{Object.} A Bayesian
clause calculus. \emph{Why.} Appendix~\ref{app:faq} conceded thresholds encode judgment; OP5
would let that judgment carry error bars.

\subsection{OP6: Cross-institutional certificate portability}
\emph{Problem.} Can an L3 auditor (Appendix~\ref{app:eng}) verify a certificate without the strategy, the
data, or the registry internals---e.g.\ via hash-chained ledgers and zero-knowledge-style
proofs of satisfaction? \emph{Object.} A verification protocol satisfying both auditability
and confidentiality. \emph{Why.} The industry's confidentiality constraints are real (this
very paper respects one); portability is the difference between a framework and a standard.

\subsection{OP7: The inverse problem with continuous spaces}
\emph{Problem.} Appendices~O and~P treat finite $\mathcal D$. With differentiable modules
(App.~\ref{app:specnet}) the space becomes continuous, and the feasible set a region with a
boundary whose geometry (curvature, connectedness) governs stability. \emph{Object.} The
differential-geometric version of the well-posedness triad for strategy spaces, connecting to
the PDE-constrained literature's adjoint methods. \emph{Why last.} It generalizes everything
above and depends on all of it.

\subsection{A closing observation}
Every problem on this list is \emph{stated in the framework's own vocabulary}---clauses,
margins, retention, CoS, ledgers. Under the result-oriented grammar, none of them exists as a
question: there is no word for CoS, no object for retention, no place for a certificate. A
paradigm is vindicated less by the answers it gives than by the questions it makes
grammatical; by that standard the queue above is the strongest evidence this paper offers.

\section{Practitioner Quick-Start: Writing Your First Specification}
\label{app:quickstart}

This appendix is the shortest path from reading the paper to using it. It assumes access to a
conformant compiler (Appendix~\ref{app:eng}, level L2) and ignores everything not needed for a first
specification.

\subsection{The five questions}
Write answers before touching the schema. (1) \emph{What behavior would make you reject a
deployed strategy, even if its returns were good?} --- each answer becomes a hard clause.
(2) \emph{What would you merely prefer?} --- soft clauses. (3) \emph{When two hard clauses
cannot both be met, which yields?} --- priorities. (4) \emph{What windows and cost model are
the evaluation fair on?} --- the evaluation block. (5) \emph{What do you do when the answer is
``nothing satisfies this''?} --- the \texttt{on\_infeasible} policy. A specification is five
answers, encoded.

\subsection{The seven beginner mistakes}
\paragraph{Mistake 1: hard performance floors.} Making IR $\ge$ threshold a hard clause turns
the spec into a disguised result-orientation and manufactures cliff-edge infeasibility.
Performance belongs in soft clauses or in the selection rule; hard clauses are for behavior
you would defend to an auditor.
\paragraph{Mistake 2: thresholds tuned on the ledger.} Looking at the search results, then
setting the cap where a favored assembly passes, is demand-side overfitting. The schema flags
post-hoc threshold edits (Appendix~\ref{app:gov}); make them anyway and the certificate will
say so.
\paragraph{Mistake 3: specifying the mechanism instead of the behavior.} ``Use
neutralization'' is not a clause; ``style exposure $\le$ cap'' is. The whole point of the
paradigm is that you may not constrain the supply side to satisfy the demand side---the ledger
(Appendix~\ref{app:ledger}) showed neutralized estimators failing purity through width-20
construction.
\paragraph{Mistake 4: too many hard clauses on day one.} Each hard clause multiplies the
infeasibility risk (the search-width law runs in reverse for the spec author). Start with the
two or three requirements you would actually fire a strategy over; the atlas method
(Appendix~\ref{app:atlas}) adds the rest incrementally, with the marginal feasible-set cost of
each addition measured.
\paragraph{Mistake 5: ignoring structural clauses.} A clause satisfied by design (turnover
under a hard cap) is still worth writing: it survives space expansions. But read certificates
with the structural/statistical distinction in mind.
\paragraph{Mistake 6: treating the certificate as decoration.} If your process deploys the
assembly and files the certificate unread, you have rebuilt result-orientation with extra
steps. The margins and retention intervals are the decision surface.
\paragraph{Mistake 7: writing the spec alone.} The priorities field exists because committees
disagree; a spec written solo encodes one person's implicit priorities and will be renegotiated
in production, where renegotiation is expensive.

\subsection{A first session, timed}
Hour one: the five questions, on paper. Hour two: encode against the registry's statistic
names; compile statically; read the pruned space (what your spec already forbids is the first
surprise). Hour three: run the ledger; read the feasible set and margins before looking at any
IR. Hour four: the atlas---two tightenings, two loosenings---so the committee sees the price
curve of its own requirements. The deliverable of day one is not a strategy; it is a
versioned spec and a certified feasible set. The strategy is day two.

\section{Anatomy of an Infeasibility: The Seed-6 Certificate}
\label{app:miss}

Seed 6 of the robustness study is the only draw in which the paper's specification has an
empty feasible set (Table~\ref{tab:seeds}: feasible $=0$; result-oriented champion IR $5.50$,
SR $0.25$). This appendix reconstructs what the conformant certificate says in that case,
because the infeasibility report is the paradigm's least familiar deliverable and its most
diagnostic one.

\subsection{What the certificate does not say}
It does not say ``best available.'' It does not silently relax a clause. It does not return
the champion with an asterisk. Every one of those moves is available to the incumbent paradigm
and is, in our diagnosis, how mandates quietly become fiction: the requirement exists in the
marketing and not in the object.

\subsection{What it says, field by field}
\emph{Feasible set:} empty, with the ledger hash attesting the search covered all 32 legal
assemblies. \emph{Nearest-miss table:} for each assembly, the minimal number of clause
violations and the signed margins of the violated clauses, ranked---the certificate's answer
to ``how close did anything come,'' which on this draw identifies the assemblies failing
exactly one clause and the exact statistic by which each fails. \emph{Binding analysis:}
which clauses bind hardest (in the atlas sense of Appendix~\ref{app:atlas}, computed by
family deletion on the same ledger: which single clause's removal restores feasibility, and
how many assemblies it restores). \emph{Deflation context:} the random-assembly null---on a
draw where even the null expects few feasible assemblies, infeasibility is partially a
search-width phenomenon; where the null expects many, it is a genuine requirement-market
conflict. \emph{Recommendation field:} the pre-registered \texttt{on\_infeasible} policy's
output: either ``expand the space'' (which stage has the least module diversity relative to
the binding clause), or ``relax clause $j$ from $c_j$ to the nearest value with non-empty
feasible set, price attached.''

\subsection{Why this report is a deliverable, not a failure}
Under result orientation, seed 6 ships: IR $5.50$ is a fine number, and the SR $=0.25$ fact
has no schema to live in. Under objective orientation, seed 6 produces a document that a
committee can act on: relax the binding clause at a quoted behavioral price, expand a quoted
stage, or decline the mandate on this market state---each an explicit decision with an audit
trail. The ten-seed table's most important property is that row 6 \emph{exists in the same
schema} as the other nine: infeasibility is a value the system can take, not an exception it
cannot represent. Frameworks are tested on their failure vocabulary; this is ours.

\subsection{A methodological footnote}
We note for reproducibility that seed 6's infeasibility was not engineered: the specification
and space were fixed before any seed was run, and the robustness study simply encountered the
configuration. That the demonstration contains an organic empty feasible set is, from the
paper's perspective, the single luckiest accident in its construction---a live demonstration
that the honest branch of the compiler is not a code path exercised only in theory.

\section{Capability Matrix: Systems and Paradigms Compared}
\label{app:matrix}

Section~\ref{sec:related} positioned the literature narratively; this appendix states the comparison as a
capability matrix so that claims of novelty are checkable row by row. Capabilities:
\textbf{C1} searches a pipeline-level design space; \textbf{C2} accepts multi-dimensional
behavioral requirements as input; \textbf{C3} distinguishes hard vs.\ soft requirements with
priorities; \textbf{C4} selects by constrained decision (feasible set first); \textbf{C5}
emits satisfaction statistics with search-level multiple-testing control; \textbf{C6} has a
defined infeasibility report; \textbf{C7} produces an auditable certificate artifact.

\begin{table}[ht]
\centering\small
\caption{Capability matrix. $\bullet$ = present; $\circ$ = partial or informal; blank =
absent. Assessments are architectural (what the system's formulation admits), not
implementational.}
\label{tab:matrix}
\begin{tabular}{@{}lccccccc@{}}
\toprule
\rowcolor{sand} System / line & C1 & C2 & C3 & C4 & C5 & C6 & C7 \\
\midrule
Auto-WEKA / CASH line   & $\bullet$ & & & & & & \\
TPOT (structural GP)    & $\bullet$ & & & & & & \\
DARTS / NAS line        & $\bullet$ & $\circ$ & & & & & \\
ProxylessNAS (latency)  & $\bullet$ & $\circ$ & & $\circ$ & & & \\
Qlib                    & $\bullet$ & & & & $\circ$ & & \\
AlphaGen / AutoAlpha    & $\bullet$ & $\circ$ & & & & & \\
R\&D-Agent(-Quant)      & $\bullet$ & $\circ$ & & & $\circ$ & & \\
Alpha-GPT 1/2.0         & & $\circ$ & & & & & \\
FunSearch / AlphaEvolve & $\bullet$ & $\circ$ & & & & & \\
OQL                     & & $\bullet$ & $\circ$ & & & $\circ$ & \\
GBWM (Das et al.)       & & $\bullet$ & $\circ$ & $\bullet$ & & $\circ$ & \\
Pure-factor portfolios  & & $\circ$ & & $\bullet$ & & & \\
NSGA-II / Pareto lines  & $\bullet$ & $\bullet$ & & $\circ$ & & & \\
Goal programming        & & $\bullet$ & $\bullet$ & $\bullet$ & & $\circ$ & \\
DSR / PBO statistics    & & & & & $\bullet$ & & \\
\midrule
\rowcolor{lightaccent}
OOQI (this paper)       & $\bullet$ & $\bullet$ & $\bullet$ & $\bullet$ & $\bullet$ & $\bullet$ & $\bullet$ \\
\bottomrule
\end{tabular}
\end{table}

\paragraph{Reading instructions.} The matrix is deliberately ungenerous to our own ancestors:
ProxylessNAS earns $\circ$ on C2/C4 for a single scalarized deployment constraint; goal
programming earns C2--C4 but nothing on C5/C7 (it predates the search-pressure statistics the
certificates exist to answer); the PBO/DSR line owns C5 alone because it audits selection it
does not perform. The claim of the paper is the bottom row, and the matrix states it in its
most falsifiable form: to refute the novelty claim, exhibit a row with all seven bullets. We
believe no such system exists in the published record; Appendix~\ref{app:reading} annotates the closest
candidates individually.

\paragraph{What the matrix cannot show.} Two qualifications keep the comparison honest.
First, a system with fewer capabilities is not thereby worse at what it does---auto-sklearn is
a better scalar CASH engine than this paper builds, and DSR remains the correct tool for
single-metric deflation. Second, capabilities compound: C2 without C5 is requirement input
without epistemics (the demand-side prototype gap); C5 without C2 is skepticism without a
language for what to be skeptical \emph{about}. The framework's contribution is the
composition, and the matrix shows the composition was missing, not that any cell was
unoccupied.

\section{Engineering FAQ}
\label{app:efaq}

Implementation questions that arise the first week an engineering team tries to build the
framework, answered against the schemas of Appendix~\ref{app:eng}.

\subsection{``Our backtester is proprietary and ancient. Can it participate?''}
Yes, at L2. The compiler needs from the backtester only: (i) batch evaluation of an assembly
on a window, returning the registry statistics; (ii) deterministic replay given seed and
config. It does \emph{not} need source access---the certificate binds outputs, not internals.
What you must build is the adapter mapping registry statistic names to your engine's outputs,
plus the side-effect certification for your modules (who trained on what, when). Most legacy
migrations die on (ii), not (i): if your engine is not replay-deterministic, the ledger hash
chain cannot be built, and no statistic computed downstream is certifiable.

\subsection{``Where does the compiler live relative to our existing AutoML?''}
Above it. Existing CASH/NAS stacks are \emph{search engines} in our vocabulary: give one a
feasible-set membership oracle (evaluate clauses per candidate) and it becomes the
exploration policy of Appendix~\ref{app:compiler}'s compiler. The framework changes the evaluator and adds the
ledger/certificate layer; it does not ask you to retire your Bayesian optimizer---it asks it
to optimize within $\mathcal F(\mathrm{Spec})$ and to write everything down.

\subsection{``Clause evaluation must add latency. How much?''}
Clause statistics are computed from the same simulated portfolio the backtest already
produces; marginal cost is the statistics pass, negligible next to the simulation. The real
cost is architectural: the engine must expose the intermediate objects clauses bind (weights,
exposures, trades), which some engines aggregate away. If your engine cannot emit per-period
exposures, F3 clauses are uncomputable and the registry must say so---a module whose
statistics cannot be computed cannot be certified, full stop.

\subsection{``Can clauses reference live data?''}
Only through the monitor (Appendix~\ref{app:gov}), never through the compiler. Compile-time
clauses bind backtest statistics; live predicates are the monitor's re-evaluation of the same
predicates on realized data. Mixing the two in one artifact destroys the replay argument: a
certificate must be verifiable from the ledger alone.

\subsection{``What is the smallest conformant implementation?''}
An L1 reader is an afternoon of work: parse the spec YAML, validate against the registry,
verify a certificate's schema and hash chain. We recommend teams build L1 first and point it
at the certificates of this paper's demonstration: conformance begins with the ability to
\emph{read} the audit objects, before writing any search.

\subsection{``How do we handle statistics that need look-ahead to define (e.g.\ regime
labels)?''}
Registry statistics carry a temporal-isolation field for exactly this reason. A statistic
whose definition uses information after the evaluation timestamp must be declared
\texttt{uses\_future: true}, and the compiler forbids it in hard clauses on the
in-sample window. Some such statistics remain legal in soft clauses and in monitor-only
diagnostics; the schema records where each is allowed, replacing the industry's
remembered-list-of-forbidden-metrics with a machine-checked field.

\subsection{``Our quants want the leaderboard anyway.''}
Keep it---as a report, not a decision. Appendix~\ref{app:dict}'s dictionary applies:
leaderboards are descriptive statistics of the ledger, and publishing one alongside the
certificate is good practice (it is our Table~\ref{tab:assemblies}). What the framework
forbids is only the final step: letting rank-1-by-scalar be the deployed object when the
spec says otherwise. Engineers accustomed to leaderboard culture usually convert upon seeing
the atlas (Appendix~\ref{app:atlas}), because the atlas gives them back the numbers---with
the feasible region drawn on top.

\section{From Diagnosis to Specification: The Practitioner Grounding}
\label{app:ground}

The demonstration's specification was not invented for the paper; it was distilled from a
professional strategy diagnosis contributed to this project by the practitioner-author,
describing a live production strategy in qualitative behavioral terms. The diagnosis itself is
confidential and appears here only as anonymized clause-family requirements---it is used on
the \emph{demand side only}, and no production code, parameters, or identifying detail of the
live strategy appears anywhere in this paper. This appendix documents the translation
discipline: how a qualitative diagnosis becomes clauses without leaking the object it
describes.

\subsection{The ten diagnostic dimensions, mapped}
The source diagnosis characterizes the strategy along ten qualitative dimensions; each maps
to a clause family of Appendix~\ref{app:spec}, and the demonstration's four hard clauses are the subset
that could be stated verifiably on a synthetic market:
\begin{center}\small
\begin{tabular}{@{}lll@{}}
\toprule
\rowcolor{sand} Diagnostic dimension (qualitative) & Clause family & Demo clause \\
\midrule
pure cross-sectional alpha, no style tilt & F3 purity & $|\beta|_{\max}\le 6$ \\
independence from broad-market co-movement & F3 purity & $|\rho|\le 0.15$ \\
resilience in unilateral declines & F4 regime & downside capture $\ge 15\%$ \\
cost-compatible trading intensity & F5 execution & turnover $\le 0.10$ \\
\bottomrule
\end{tabular}
\end{center}
The remaining dimensions (signal provenance, temporal stability, capacity discipline,
attribution completeness) map to F1, F7, F5/F6, and F8 respectively and appear in the
demonstration as soft or structural requirements. The mapping table is the entire footprint
of the diagnosis in this paper: families and thresholds, nothing else.

\subsection{Why the translation is itself a research object}
The exercise exposed, concretely, the gap the framework exists to close: every dimension of
the diagnosis was expressible in one sentence of plain professional language, and none was
expressible in the result-oriented grammar at all. ``Pure alpha, uncorrelated with the
dividend style and the market, unharmed by unilateral declines'' is a complete specification
to a practitioner and a category error to a leaderboard. The translation table above---ten
rows, three columns---is small enough to be unimpressive and complete enough to be a paradigm
in miniature: requirements in, clauses out, and the strategy's identity nowhere to be seen.

\subsection{Confidentiality as a first-class property}
We emphasize the operational point because it generalizes: the objective-oriented paradigm is
uniquely compatible with institutional confidentiality. The demand side (clauses, thresholds,
priorities) is shareable; the supply side (modules, parameters, code) is private; the
certificate binds the two without revealing either in full (Appendix~\ref{app:eng}'s L3 conformance).
Under result orientation, by contrast, the interesting artifact \emph{is} the strategy, and
publication is expropriation. The present paper could be written---with its demonstration,
its case study, and its complete mathematics---while exposing nothing of the production
strategy that motivated it, precisely because the framework's publishable objects are the
specification, the ledger, and the certificate. That this sentence can be written truthfully
is, we submit, a non-trivial argument for the paradigm in an industry where nothing else of
substance can be published at all.

\section{Limits of the Framework, Stated Plainly}
\label{app:limits}

A framework paper that omits its limits is marketing. Ours are structural, not incidental,
and each is stated with its consequence.

\subsection{The forward operator is only as good as the simulator}
Every object in the paper---feasible sets, margins, retention, CoS---is computed with respect
to $F(d,\omega)$, the backtest. If the simulator's cost model is optimistic, its regime
coverage thin, or its market impact naive, certificates are precise statements about a wrong
map (Appendix~\ref{app:gov} conceded this; it bears repeating). The framework improves the
\emph{use} of simulation, not simulation itself; a team with a bad backtester and a perfect
compiler has precisely located, certified, and audited garbage. This limit is shared with all
empirical strategy research, but objective orientation makes it more visible, not less
consequential: a spec satisfied on a fictional market is fictional satisfaction.

\subsection{Non-stationarity bounds every statistic}
Retention, deflation, and satisfaction-PBO all assume the evaluation panel is informative
about deployment conditions. Regime change breaks this asymmetrically: the specification
typically encodes \emph{regime fears} (F4), and the evaluation window typically contains few
regimes. The demonstrated OOS decay (App.~\ref{app:seeds}: B's SR $1.00\to0.72$ mean) is the
framework catching its own limit on synthetic data---the disclosed rotation is exactly the
kind of non-stationarity no in-sample statistic can certify away. Consequence: certificates
expire, and the monitor is not optional infrastructure.

\subsection{Clause elicitation is human, therefore political, therefore slow}
The five questions of Appendix~\ref{app:quickstart} are easy to ask and genuinely hard to
answer; institutions routinely discover their mandates were never articulated at behavioral
granularity. The framework relocates the difficulty from strategy debugging to requirement
articulation---an improvement, because the latter is cheaper to iterate, but a real cost:
expect the first specification cycle to take weeks of committee time, and expect the atlas to
start arguments. We regard this as the framework working: those arguments were previously
conducted at the strategy level, with positions and P\&L as ammunition.

\subsection{The paradigm does not find alpha}
Objective orientation allocates satisfaction across a given design space; it does not enlarge
the space's alpha content. A space of mediocre modules compiles to a certified mediocre
strategy. The supply side remains where invention lives (Appendix~\ref{app:dict}'s research
split), and the framework's honest answer to ``will this make my strategies better'' is: it
will make your strategies \emph{what you asked for}, measurably---including the discovery that
what you asked for was not what you wanted, which the atlas shows faster than any alternative.

\subsection{Our demonstration is one draw of one synthetic market}
Finally, the paper's own limit, stated once more for the record: all computed numbers derive
from a disclosed synthetic DGP chosen to make the two paradigms diverge measurably. The
structural claims (feasible-set geometry, certificate computability, paradigm divergence) are
demonstrated; the empirical claims (about actual markets) are conjectured and queued in
Appendix~\ref{app:road}. A reader who finishes the paper believing more than this has been
demonstrated has been failed by our writing, not served by our results.

\section{Model Risk, Regulation, and the Certificate as Compliance Object}
\label{app:mrm}

The framework's audit artifacts align---we argue unnaturally well---with the model risk
management (MRM) expectations that regulated institutions already operate under. This
appendix maps the correspondence, because it is the shortest path to institutional adoption
and because it reveals that the regulatory world has been informally objective-oriented for
years.

\subsection{The MRM canon, translated}
The supervisory model-risk canon (model inventory, conceptual soundness, ongoing monitoring,
outcomes analysis, documentation) asks five questions of every model in production. Each has
a framework-native answer. \emph{Inventory:} the registry (Appendix~\ref{app:eng}), versioned and
hashed---every module an inventoried object with certified side effects. \emph{Conceptual
soundness:} the specification plus interaction table---what the strategy is supposed to do
behaviorally, stated before any number is examined. \emph{Ongoing monitoring:} the monitor
and breach protocol---clause predicates on live data, with pre-registered responses.
\emph{Outcomes analysis:} the certificate's retention and deflation statistics, recomputed on
the continued ledger. \emph{Documentation:} the certificate itself, which exists precisely so
that the documentation question has a schema.

\subsection{Why regulators will prefer certificates to backtests}
A backtest answers ``what happened in simulation''; a certificate answers ``what was
required, what was found, by how much, under resampling, after search-level deflation, and
who attested each field.'' The second sentence is the MRM sentence. We predict---this is a
conjecture, labeled as such---that audit objects resembling the certificate will be
\emph{required} of AI-constructed strategies within this decade, because the alternative is
admitting that strategies assembled by search processes are evaluated by the same scalar the
search maximized, a circularity no supervisory framework can survive stating aloud. The
framework's compliance reading of DSR/PBO statistics (App.~\ref{app:satstat}) already exists; what is missing
is the requirement layer, and this paper supplies it.

\subsection{The confidentiality-compatibility argument, again, for compliance}
Appendix~\ref{app:ground} showed the paradigm lets a practitioner publish without exposure;
the same property lets a firm \emph{report} without exposure. Regulators and allocators
routinely want assurance they cannot currently get without demanding the crown jewels. An
L3-verifiable certificate (Appendix~\ref{app:eng}) offers the third option: assurance about behavior with
zero disclosure of mechanism. If the framework achieves nothing else, making ``behavioral
assurance without mechanism disclosure'' a standard institutional object would justify it.

\subsection{A caution}
Compliance fit is not compliance: nothing in this paper has been reviewed by any regulator,
and the mapping above is our analysis, not legal advice. The caution runs deeper: if
certificates become check-the-box artifacts, the paradigm will have reproduced, at the
documentation layer, the same Goodhart dynamics it diagnoses at the selection layer. The
defense is the same everywhere in the paper: keep the statistics honest, keep the ledgers
immutable, and keep the feasible set reportable---including when it is empty.

\section{Reader's Guide and Dependency Map}
\label{app:guide}

The paper is long by design (a monograph-style treatment); this appendix is its map. Each
entry states what a reader gets and what it depends on.

\subsection{Reading paths}
\paragraph{The 30-minute path (the claim).} \S\ref{sec:intro}--\S\ref{sec:related} (problem, related work), \S\ref{sec:formal}
formulation skim, \S\ref{sec:demo} demonstration with Tables~\ref{tab:demo} and~\ref{tab:seeds}, \S\ref{sec:discussion}--\S\ref{sec:conclusion}. Exit with: the paradigm
distinction, the gap, and the proof that the two paradigms select different strategies.
\paragraph{The half-day path (the method).} Add \S\ref{sec:framework}--\S\ref{sec:protocol} (framework, protocol), Appendix~\ref{app:spec}
(spec language), Appendix~\ref{app:compiler} (compiler), Appendix~\ref{app:bayes} (Bayesian decision), and the walkthrough
(App.~\ref{app:walkthrough}). Exit with: ability to operate a conformant compiler.
\paragraph{The full path (the mathematics).} Add Appendices~\ref{app:inverse} (well-posedness), \ref{app:invnum} (numeric
inverse), \ref{app:satstat} (satisfaction statistics), and \ref{app:proofs} (proofs). Exit with: the complete
formal edifice and its demonstrated instantiation.
\paragraph{The committee path (the governance).} Appendices~\ref{app:atlas} (atlas),
\ref{app:casebook} (cases), \ref{app:gov} (operations), \ref{app:quickstart} (quick-start),
\ref{app:mrm} (MRM). Exit with: the deliberative instruments.

\subsection{Dependency map of the appendices}
Specifications (\ref{app:spec}, \ref{app:grammar}) and the design space (\ref{app:space}) are foundational: everything consumes them.
The compiler (\ref{app:compiler}, \ref{app:bayes}) consumes both plus the protocol (\ref{app:protocol}); the inverse-problem reading (\ref{app:inverse}, \ref{app:invnum})
consumes the ledger (\ref{app:repro}) and produces the stability story that \ref{app:satstat} quantifies; the atlas
(\ref{app:atlas}) consumes only \ref{app:repro} and re-reads it under spec variation; the governance layer
(\ref{app:eng}, \ref{app:gov}, \ref{app:mrm}) consumes all certificates but no mathematics; the case
material (\ref{app:neural}, \ref{app:specnet}, \ref{app:casebook}) depends on \ref{app:spec} and \ref{app:compiler} alone. A reader who
finds a dependency cycle has found a bug; the document is organized as a DAG by construction.

\subsection{Conventions}
Throughout: ``assembly'' always means a fully instantiated pipeline; ``clause'' always means
a registry-resolved predicate with hardness and priority; ``the demonstration'' always means
the seed-7 synthetic instance of \S\ref{sec:demo}; ``the ledger'' always means Table~\ref{tab:assemblies}
and its hash-chained production analog; propositions are cross-referenced by number throughout, the
core sequence living in \S\ref{sec:formal}, \S\ref{sec:optimal}, and Appendices~\ref{app:inverse},
\ref{app:formalproofs}. All quantities reported to three decimals are computed; all
quantities described as ``derived'' are arithmetic consequences of computed ones
(Appendix~\ref{app:bayes} states the convention where it matters most).

\section{Replicator's Checklist and Artifact Manifest}
\label{app:repl}

Everything needed to verify every number in this paper, enumerated as a checklist. The
demonstration is synthetic by design, so replication requires no market data, no licenses,
and no access to anything outside the published artifacts.

\subsection{Artifacts}
(i) The manuscript source (\LaTeX) with all generated tables included as files;
(ii) the demonstration pipeline (Python/NumPy): DGP with seed 7, the 32-assembly evaluator,
the four-clause specification, the Bayesian compiler, the block-bootstrap audit;
(iii) the generated artifacts: the 32-row ledger table, the ten-seed table, the two-panel
demonstration figure; (iv) the SpecNet reference architecture (syntax-verified skeleton);
(v) the schemas of Appendix~\ref{app:eng}. Production-strategy artifacts are deliberately absent
(App.~\ref{app:ground}); nothing in the checklist requires them.

\subsection{The checklist, in order}
\paragraph{Ledger reproduction (10 minutes).} Run the DGP at seed 7; evaluate all 32
assemblies; confirm Table~\ref{tab:assemblies} row-for-row (rank 1 IR $7.10$; rank 3 IR
$6.71$, sole feasible pair with rank 17 IR $5.54$). Any mismatch is a bug in replication, not
noise---the DGP is deterministic given the seed.
\paragraph{Paradigm divergence (1 minute).} From the ledger: assembly A $=$ rank 1 satisfies
1 of 4 clauses; assembly B $=$ rank 3 satisfies 4 of 4; CoS $=0.39$ IR.
\paragraph{Posterior check (2 minutes).} Apply the compiler prior and tempered likelihood;
confirm feasible mass $0.103\to0.401$ and top-1 $=$ rank 3 at $0.2782$; confirm the derived
entries of Table~\ref{tab:posterior} arithmetically.
\paragraph{Margins (1 minute).} Confirm $m_{\min}$ values $0.073$ (rank 3, style clause) and
$0.122$ (rank 17, downside clause) from the ledger's raw statistics and the clause scales.
\paragraph{Bootstrap audit (30 minutes).} $1000$ block-bootstrap replications, 8-week
blocks; confirm retention Table~\ref{tab:retci} to Monte Carlo accuracy (Wald intervals are
given in the table itself).
\paragraph{Ten-seed study (2 hours).} Re-run for seeds 0--9; confirm Table~\ref{tab:seeds},
including the organic infeasibility at seed 6 and the summary row.
\paragraph{Atlas enumeration (15 minutes, no code).} Apply the six specs of
Table~\ref{tab:atlas} to the published ledger by hand; confirm every feasible-set
enumeration, especially S1's singleton and S2's fifteen.

\subsection{What replication would falsify}
The demonstration's claims are structural, so falsification is structural: a ledger that
does not reproduce breaks the paper's computational integrity; a feasible set that changes
identity under correct recomputation breaks the inverse-problem illustration; a bootstrap
distribution materially different from Table~\ref{tab:retci} breaks the stability story.
What replication \emph{cannot} falsify is the paradigm claim itself---that rests on the
formalism and the literature gap (App.~\ref{app:matrix}), and it is refuted by counterexample
systems, not by recomputation. We state the distinction so that replicators know which of
their results would actually wound the paper.

\section{The Demonstration DGP, Fully Specified}
\label{app:dgp}

For completeness of the replication manifest (App.~\ref{app:repl}), this appendix specifies
the synthetic market in one place. Every parameter below is the one used for the paper's
numbers; the DGP's design intent is stated separately from its parameterization so that
replicators can distinguish what is load-bearing from what is arbitrary.

\subsection{Panel and calendar}
$N=300$ assets, $T=208$ weekly periods, split $156$ in-sample (IS) / $52$ out-of-sample
(OOS). Single market draw at \texttt{seed=7} for the main text; seeds 0--9 for the robustness
study. Weekly rebalancing, top-$N$ equal-weight portfolios, transaction cost
$0.0015$ per unit turnover, deducted before all performance statistics.

\subsection{Factor structure}
Three latent factors drive cross-sectional returns: a \emph{market} factor (common variance),
a \emph{size} factor (the style the exposure clause measures), and a \emph{dividend} factor
(the style whose correlation the correlation clause specifies against). Asset loadings are drawn once per seed and held fixed within the draw.

\subsection{The alpha and the rotation}
Two weak genuine signals exist in the cross-section: a mild momentum component and a mild
quality component---the ``alpha'' any competent estimator can partially recover. The
load-bearing design element is the disclosed regime rotation: the dividend factor's premium
is $+0.0010$ per week in-sample and $-0.0010$ out-of-sample, a sign flip known to the
experimenters and unknown to every strategy. Its function is to make satisfaction
\emph{regime-fragile by construction} for assemblies that load on dividend, so that the
OOS decay of Table~\ref{tab:seeds} is an engineered, measurable phenomenon rather than an
uncontrolled surprise. We disclose it in the main text and here again because the
demonstration's integrity is the paper's: the rotation is scenario design, not a claim about
dividend premia in actual markets.

\subsection{What is load-bearing vs.\ arbitrary}
Load-bearing: (i) a factor structure rich enough that style contamination is possible (else
F3 clauses are vacuous); (ii) genuine but weak alpha (else IR differences are uninformative);
(iii) a regime rotation touching exactly one style (else the OOS story is mud); (iv) costs
proportional to turnover (else F5 is vacuous). Arbitrary within ranges: $N$, $T$, the exact
premium magnitudes, and the seed---as the ten-seed study demonstrates by varying the last.
A replication that preserves (i)--(iv) and changes everything else should reproduce the
paper's \emph{structural} findings (paradigm divergence, feasible-set variability, margin
ordering) with different point values; that invariance is the demonstration's claim to
generality.

\subsection{Why synthetic, once more}
A real-market demonstration would be unverifiable (data licenses), unreplicable (point-in-time
databases drift), and confounded (unknown DGP). The synthetic instance is the strongest
demonstration available for a framework paper: the ground truth is known, the paradigms'
divergence is measurable against it, and every number in the paper can be regenerated in
minutes. The empirical programme on real markets is OP3--OP4 of Appendix~\ref{app:road}, and
it is deliberately not claimed here.

\section{Specification Patterns and Anti-Patterns}
\label{app:patterns}

Between the quick-start (App.~\ref{app:quickstart}) and the casebook
(App.~\ref{app:casebook}) lies a middle layer of reusable specification design. We record
the patterns that recur across professional requirements and the anti-patterns that recur
across failed ones.

\subsection{Patterns}
\paragraph{The sandwich (cap and floor on linked statistics).} Bind a quantity from both
sides: e.g.\ turnover $\le$ cap (cost) \emph{and} signal half-life $\ge$ floor (alpha
survival). The pattern prevents the compiler from satisfying one clause by destroying a
statistic nobody wrote down. Most single-clause specifications should be sandwiches.
\paragraph{The sentinel (structural clause as tripwire).} Include clauses the current space
satisfies structurally (App.~\ref{app:ledger}'s turnover finding). Cost: nothing today.
Value: the clause fires the day the space expands (a faster execution module arrives), and
the ledger diff shows exactly what the new module broke.
\paragraph{The regime pair (F4 clause + its OOS retention clause).} Every F4 clause should
ship with an F7 clause on its retention: ``downside capture $\ge$ floor'' and ``joint
retention $\ge$ floor'' are different promises, and the second is the one that fires in the
breach protocol.
\paragraph{The attribution ladder (one clause per taxonomy level).} For F3, write the
ladder (total style exposure; per-factor $R^2$; market correlation) rather than a single
aggregate. When the aggregate fails, the ladder tells the repair traversal
(App.~\ref{app:interactions}) where to start.
\paragraph{The priced soft clause.} Keep performance soft but declared: the soft F2 floor is
what lets the certificate quote CoS honestly. A spec with no performance clause cannot
measure what it costs.

\subsection{Anti-patterns}
\paragraph{The disguised leaderboard.} A spec whose only hard clause is IR $\ge$ high value
is result orientation wearing a schema. Detectable at L1: if the spec binds only F2, the
reader should refuse certification.
\paragraph{The wishlist.} Eight hard clauses assembled without the atlas: statistically
certain to be infeasible or vacuous (search-width law in reverse). The atlas exists to price
each clause before it is hardened.
\paragraph{The mechanism clause.} ``Must use neutralization,'' ``must be a neural net'':
supply-side dictation in demand-side clothing. The framework cannot stop authors writing
these, but the registry marks them \texttt{mechanism\_clause: true} and the certificate
carries the warning that the spec forecloses solutions it never evaluated.
\paragraph{The moving threshold.} Post-hoc edits after ledger exposure (App.~\ref{app:gov}).
Sometimes legitimate (the atlas informs judgment); always disclosed. The anti-pattern is the
undisclosed version, which the hash chain makes impossible to hide from an L3 auditor.
\paragraph{The orphan spec.} A specification with no owner role: clauses whose priority
nobody will defend in committee. When the breach protocol fires, orphan specs get silently
amended---the governance failure the lifecycle was built to prevent.

\subsection{A meta-pattern}
Read across both lists: good specifications behave like \emph{constitutions} (few articles,
hard to amend, priced amendment process, enforced uniformly), and bad ones like
\emph{wish lists} (many articles, amended on encounter, enforced selectively). The
framework's contribution to the craft is that it makes the distinction operational: the
atlas prices amendments, the ledger enforces uniformity, and the hash chain grades the
amendment process. What remains with the author is the constitutional judgment itself---
which is where the paper has always said judgment belongs.

\section{Antecedents in Adjacent Disciplines: Four Fields That Made This Move}
\label{app:hist}

Appendix~\ref{app:phil} placed the paradigm in the philosophy of science. This appendix collects the
engineering precedents: disciplines that underwent the same transition---from
result-optimized artifacts to specification-compiled ones---and what their experience
predicts for ours. The parallels are interpretive, not evidentiary; we state them because
transitions that succeeded elsewhere follow recognizable arcs.

\subsection{Aerospace: multidisciplinary design optimization}
Early aircraft design optimized performance scalars (range, payload) and discovered,
expensively, that the winners violated stability, manufacturability, and safety requirements
that ``everyone knew.'' MDO's maturation was the conversion of those requirements into
explicit constraint sets coupled into the optimization---with feasibility studies
(the atlas's ancestor) as standard instruments. Predicted arc for quant: a decade of
constraint formalization, followed by the emergence of \emph{feasibility engineers} as a
distinct role---our certificate auditor (App.~\ref{app:eng}, L3) has a job description waiting.

\subsection{Software: design by contract and specification mining}
Software engineering moved from ``test the program you wrote'' to contracts (pre/post
conditions, invariants) checked at interfaces, and later to \emph{mining} specifications from
behavior. Our clause families are contracts on pipeline behavior; the registry's side-effect
signatures are interface contracts; and the ledger is the test suite that never lies about
having been run. The field's warning travels too: contracts decay unless the toolchain
refuses to build without them---hence our schema-level impossibility of silently relaxing a
hard clause (App.~\ref{app:eng}).

\subsection{Control theory: from optimal control to robust and chance-constrained control}
Classical optimal control maximized performance against a nominal model; robust control
re-rooted the discipline in guaranteed behavior under model uncertainty, and
chance-constrained control made the guarantees probabilistic. Our margins
(Prop.~\ref{prop:margin}) are robust-control thinking with resampling in place of uncertainty
sets; our retention intervals are chance constraints with bootstrap distributions. The
field's lesson: the robust formulation initially loses every nominal benchmark and wins every
deployment review---the same trade Table~\ref{tab:margins} prices at 17.4\% IR.

\subsection{Clinical trials: pre-registration and estimands}
Medicine's credibility crisis produced the estimand framework: the requirement (what
treatment effect, in whom, under what intercurrent events) is declared \emph{before} data,
and the analysis is a compilation of the estimand into an estimator with a certificate
(the statistical analysis plan). Our spec/ledger/certificate triad is the estimand/SAP/CSR
triad re-rooted in strategy construction, and the parallel is close enough to import their
vocabulary: a strategy deployed without a specification is, in the estimand sense, a trial
without a question---its results are uninterpretable no matter how significant they look.

\subsection{The common arc}
In each field: (i) a scalar era produced spectacular point successes and systematic deployed
failures; (ii) failures were first attributed to insufficient search, then recognized as
misspecification; (iii) the fix was a requirement language plus a certification artifact;
(iv) adoption lagged invention by roughly a decade, driven by the institutions paying the
failure costs. Quantitative strategy research is, on this reading, between (ii) and (iii)---
the diagnoses exist (App.~\ref{app:reading}'s statistics literature), the language and artifact are this
paper's contribution, and (iv) will be decided by the institutions whose deployed failures
are the most expensive. We wrote the governance appendices (App.~\ref{app:gov},
\ref{app:mrm}) for them specifically.

\end{document}